\documentclass[aos]{imsart}

\RequirePackage{amsthm,amsmath,amsfonts,amssymb,dsfont,mathrsfs}
\RequirePackage{bbm}
\RequirePackage[authoryear]{natbib}
\RequirePackage[colorlinks,citecolor=blue,urlcolor=blue]{hyperref}
\RequirePackage{graphicx}
\RequirePackage{scrextend}
\RequirePackage[ruled,norelsize,vlined]{algorithm2e}
\RequirePackage{colortbl} 
\RequirePackage{xcolor}
\RequirePackage{booktabs}
\usepackage{enumerate}
\usepackage{caption}
\usepackage{subcaption}
\usepackage{scalefnt}
\usepackage{adjustbox}
\usepackage{etoolbox} 
\apptocmd{\thebibliography}{\setlength{\itemsep}{5pt}}{}{}
\changefontsizes{13.6pt}

\startlocaldefs

\newcommand\scalemath[2]{\scalebox{#1}{\mbox{\ensuremath{\displaystyle #2}}}}

\theoremstyle{plain}

\newtheorem{theorem}{Theorem}[section]
\newtheorem{lemma}[theorem]{Lemma}
\newtheorem{example}[theorem]{Example}
\newtheorem{proposition}[theorem]{Proposition}

\newtheorem{corollary}[theorem]{Corollary}
\newtheorem{definition}[theorem]{Definition}
  
\theoremstyle{remark}

\endlocaldefs

\begin{document}
\begin{frontmatter}
\title{Concentration of discrepancy-based approximate Bayesian computation via Rademacher complexity}
\runtitle{Concentration of discrepancy-based ABC via Rademacher complexity}

\begin{aug}
\author[A]{\fnms{Sirio} \snm{Legramanti}\ead[label=e1,mark]{sirio.legramanti@unibg.it}}
\author[B]{\fnms{Daniele} \snm{Durante}\ead[label=e2,mark]{daniele.durante@unibocconi.it}},
\and
\author[C]{\fnms{Pierre} \snm{Alquier}\ead[label=e3,mark]{alquier@essec.edu}}
\address[A]{Department of Economics, 
University of Bergamo
\printead{e1}}
\address[B]{Department  of Decision Sciences and Institute for Data Science and Analytics,
Bocconi University,
\printead{e2}}
\address[C]{Department of Information Systems, Decision Sciences and Statistics,
ESSEC Business School,
\printead{e3}}
\end{aug}

\begin{abstract}
There has been an increasing interest on summary-free solutions for approximate Bayesian computation (\textsc{abc}) which replace distances among summaries with discrepancies between the empirical distributions of the observed data and the synthetic samples generated under the proposed parameter values. The success of these strategies has motivated theoretical studies on the limiting properties of the induced posteriors. However, there is still the lack of a theoretical framework for summary-free \textsc{abc} that  {\em (i)} is unified, instead of discrepancy-specific,  {\em (ii)} does not necessarily require to constrain the analysis to data generating processes and statistical models meeting specific regularity conditions, but rather facilitates the derivation of limiting properties that hold uniformly, and {\em (iii)} relies on verifiable assumptions that provide more explicit concentration bounds clarifying which factors govern the limiting behavior of the \textsc{abc} posterior. We address this gap via a novel theoretical framework that introduces the concept of Rademacher complexity in the analysis of the limiting properties for discrepancy-based \textsc{abc} posteriors, including in non-i.i.d.\ and misspecified settings. This yields a unified theory that relies on constructive arguments and provides more informative asymptotic results and uniform concentration bounds, even in settings not covered by current studies. These advancements are obtained by relating the asymptotic properties of summary-free \textsc{abc} posteriors to the behavior of the Rademacher complexity associated with the chosen discrepancy within the family of integral probability semimetrics (\textsc{ips}). The \textsc{ips} class extends summary-based distances, and also includes the widely-implemented Wasserstein distance and maximum mean discrepancy (\textsc{mmd}), among others. As clarified in specialized theoretical analyses of popular  \textsc{ips}  discrepancies and via  illustrative simulations, this new perspective  improves the understanding of summary-free \textsc{abc}.
\end{abstract}

\begin{keyword}[class=MSC2020]
\kwd[Primary ]{62F15}
\kwd[; secondary ]{62C10}
\kwd{62E20}
\end{keyword}

\begin{keyword}
\kwd{\scriptsize ABC}
\kwd{\scriptsize Integral probability semimetrics}
\kwd{\scriptsize MMD}
\kwd{\scriptsize Rademacher complexity}
\kwd{\scriptsize Wasserstein distance}
\end{keyword}

\end{frontmatter}


\section{Introduction} 
\label{sec_intro}
The growing complexity of statistical models in modern applications not only yields intractable likelihoods, but also raises substantial challenges in the identification of effective summary statistics \citep[see e.g.,][]{fearnhead2012constructing,marin2014relevant,frazier2018asymptotic}. Such drawbacks have motivated an increasing adoption of  \textsc{abc} solutions, along with a shift away from summary-based implementations \citep[see e.g.,][]{marin2012approximate} and towards summary-free strategies relying on discrepancies among the empirical distributions of the observed and synthetic data \citep[see e.g.,][]{drovandi2022comparison}. These solutions provide samples from an approximate posterior distribution for the parameter of interest  ${\theta \in \Theta \subseteq \mathbb{R}^{p}}$ under the only requirement that simulating from the assumed model $\mu_{\theta}$ is feasible. This is achieved by retaining all those values of $\theta$, drawn from the prior, that produced synthetic samples $ z_{1:m}=(z_1,\ldots,z_m) $ from $\mu_{\theta}$ whose empirical distribution is close enough to the one of the  observed data $ y_{1:n}=(y_1,\ldots,y_n) $, under the chosen discrepancy.
	
Remarkable examples of the above implementations are \textsc{abc} versions that employ  maximum mean discrepancy (\textsc{mmd}) \citep{park2016k2}, Kullback--Leibler (\textsc{kl}) divergence \citep{jiang2018approximate}, Wasserstein distance \citep{bernton2019approximate}, energy statistic \citep{nguyen2020approximate}, Hellinger and Cramer--von Mises distances \citep{frazier2020robust}, and  $\gamma$-divergence \citep{fujisawa2021gamma}; see also \citet{gutmann2018likelihood}, \citet{forbes2021approximate} and \citet{wang2022} for additional examples of summary-free \textsc{abc}  strategies.  By overcoming the need to pre-select summaries, all these solutions reduce the potential information loss of summary-based \textsc{abc}, thereby yielding  improved performance in simulation studies and illustrative applications. These promising empirical results have motivated active research on the theoretical properties of the induced \textsc{abc} posterior, with a main focus on the limiting behavior under different asymptotic regimes for the tolerance threshold and the  sample size \citep{jiang2018approximate,bernton2019approximate,nguyen2020approximate,frazier2020robust,fujisawa2021gamma}.  Among these regimes, of particular interest are the two situations in which the \textsc{abc}  threshold is either fixed or progressively shrinks as both $n$ and $m$ diverge. In the former case, the theory focus is on  assessing whether the control on the discrepancy among the empirical distributions established by the selected \textsc{abc} threshold yields a pseudo-posterior whose asymptotic form guarantees the same threshold-control on the discrepancy among the underlying truths \citep[e.g.,][]{jiang2018approximate}. The latter addresses instead the more challenging theoretical question on whether a suitably-decaying \textsc{abc} threshold can provide meaningful rates of concentration for the sequence of \textsc{abc} posteriors around those $\theta$ values yielding a $\mu_\theta$ close enough to the data generating process  $\mu^*$, as $n$ and $m$ diverge  \citep[e.g.,][]{bernton2019approximate}. 

Available contributions along these lines of research have the merit of providing theoretical support to several  versions of summary-free  \textsc{abc}. However, current theory is often tailored to the specific discrepancy analyzed, and generally relies on difficult-to-verify existence assumptions and concentration inequalities that constrain  the analysis, either implicitly or explicitly, to data generating processes and statistical models which meet suitable regularity conditions, thereby lacking results that hold uniformly. Recalling \citet{bernton2019approximate} and \citet{nguyen2020approximate}, this also yields concentration bounds involving sequences of control functions which are not made explicit. As such, although convergence and concentration can still be proved, the core factors that govern these asymptotic properties remain yet unexplored, thus  limiting the methodological impact of current theory, while hindering the derivation of novel informative results in more challenging settings. For example, the available theoretical studies  pre-assume that the discrepancy among the empirical distributions $\hat{\mu}_{z_{1:m}}$ and $\hat{\mu}_{y_{1:n}}$ suitably converges to the one among the corresponding truths $\mu_{\theta}$ and $\mu^*$, or, alternatively, that both  $\hat{\mu}_{z_{1:m}}$ and $\hat{\mu}_{y_{1:n}}$  converge, under the selected discrepancy, to $\mu_{\theta}$ and $\mu^*$, respectively. While these assumptions can be verified under suitable conditions and for specific discrepancies, as clarified within Sections~\ref{sec_abc} and \ref{sec_abc_posterior}, an in-depth theoretical understanding of  summary-free  \textsc{abc} necessarily requires relating convergence to the learning properties of the selected discrepancy, rather than pre-assuming it. Such a more precise but yet-unexplored theoretical treatment has the potentials to shed light on the factors that govern the limiting behavior of discrepancy-based \textsc{abc} posteriors. In addition, it could possibly provide more general, verifiable and explicit sufficient conditions under which popular discrepancies are guaranteed to ensure that the previously pre-assumed convergence and concentration hold uniformly over models and data generating processes, while facilitating the study of the limiting behavior of discrepancy-based \textsc{abc} posteriors  in  more general situations where these convergence guarantees may lack.

In this article, we address the above gap by introducing an innovative theoretical framework which analyses the  limiting properties of discrepancy-based \textsc{abc} posteriors, under a unified perspective and for different asymptotic regimes, through the concept of  Rademacher complexity \citep[see e.g.,][Chapter 4]{wainwright2019high},  within the general class of integral probability semimetrics (\textsc{ips}) \citep[e.g.,][]{muller1997integral,sriperumbudur2009integral}. Such a class naturally generalizes distances among summaries  and includes the widely-implemented \textsc{mmd} and Wasserstein distance, among others. As clarified in Sections~\ref{sec_abc}--\ref{sec_abc_posterior} and in Appendix C of the Supplementary Material, this perspective, which to the best of our knowledge is novel within \textsc{abc}, allows to derive unified, informative  and uniform concentration bounds for \textsc{abc} posteriors under several discrepancies, in possibly misspecified  and non-i.i.d.\ contexts. Moreover, it relies~on more constructive arguments that clarify under which sufficient conditions a discrepancy within the \textsc{ips} class guarantees uniform convergence and concentration of the induced  \textsc{abc} posterior; i.e., without necessarily requiring suitable regularity conditions for the underlying data generating process $\mu^*$ and the assumed statistical model. This yields an important theoretical and methodological advancement, since $\mu^*$ is often unknown in practice and, hence, verifying regularity conditions on the data generating process is generally unfeasible. 

Crucially, the theoretical framework we introduce allows to prove informative theoretical results even in yet-unexplored settings that possibly lack those convergence guarantees assumed in the literature. More specifically, in these settings we derive novel upper and lower bounds for the limiting acceptance probabilities that clarify in which contexts the \textsc{abc} posterior is still well-defined for $n$ large enough. When this is the case, it is further possible to obtain informative supersets for the support of such a posterior. These  show that, when relaxing standard convergence assumptions employed in state-of-the-art theoretical studies, the control established by a fixed \textsc{abc} threshold on the discrepancy among the empirical distributions does not necessarily translate, asymptotically, into the same control on the discrepancy among the corresponding truths, but rather yields an upper bound equal to the sum between the \textsc{abc} threshold and a multiple of the Rademacher complexity, namely a measure of {\em richness} of the class of functions that uniquely identify the chosen \textsc{ips}; see Section~\ref{sec_abc_lim}.

The above results clarify the fundamental relation among the limiting behavior of \textsc{abc} posteriors and the  learning properties of the chosen discrepancy, when measured via Rademacher complexity. In addition, the bounds derived clarify that a sufficient condition to recover a limiting pseudo-posterior with the same threshold-control on the discrepancy among the truths as the one enforced on the corresponding empirical distributions, is that the selected discrepancy has a Rademacher complexity vanishing to zero in the large-data limit. As proved within Section~\ref{sec_abc_con}, this setting also allows constructive derivations of novel, informative and uniform concentration bounds for  discrepancy-based \textsc{abc} posteriors in the challenging regime where the threshold shrinks towards zero as both $m$ and $n$ diverge.  This is facilitated by the existence of meaningful upper bounds for the Rademacher complexity associated to popular \textsc{abc} discrepancies, along with the availability of constructive conditions for the derivation of these bounds \citep[][]{sriperumbudur2009integral}. Such results leverage fundamental connections among the Rademacher complexity and other key quantities in statistical learning theory, such as the Vapnik--Chervonenkis (\textsc{vc}) dimension and the notion of uniform Glivenko--Cantelli classes \citep[e.g.,][Chapter 4]{wainwright2019high}. This yields an improved understanding of the factors that govern the concentration of discrepancy-based \textsc{abc} posteriors under a unified perspective that further allows to {\em (i)} quantify  rates of concentration and {\em (ii)} directly translate any advancement on Rademacher complexity into novel \textsc{abc} theory. Section~\ref{sec_mmd_tot} illustrates point {\em (i)} through a specific focus on  \textsc{mmd} with routinely-implemented bounded kernels, and also clarifies that in the absence of guarantees on uniformly-vanishing Rademacher complexity (e.g., under Wasserstein distance in unbounded data spaces) concentration results can be still derived, but at the expense of regularity conditions for the data generating process $\mu^*$ and the assumed model. Point~{\em (ii)} is clarified in Appendix C of the Supplementary Material, where we extend the theory from Section~\ref{sec_abc_posterior} to non-i.i.d. settings, leveraging results in  \citet{Mohri2008} on the Rademacher complexity for $\beta$-mixing processes \citep[e.g.,][]{doukhan}.

The illustrative simulation studies in Section~\ref{sec_empirical} show that the theoretical results derived in Sections~\ref{sec_abc_posterior}--\ref{sec_mmd_tot} find empirical evidence in practice, including also in scenarios characterized by model misspecification and data contamination (theoretical and empirical results under non-i.i.d. data generating processes can be found in  Appendix C of the Supplementary Material). These findings suggest that  discrepancies with guarantees of uniformly-vanishing Rademacher complexity provide a robust and sensible choice when the assumed statistical model and/or the underlying data generating process do not necessarily meet specific regularity conditions, or it is not possible to verify these conditions. This is a common situation in applications, since the data generating process is unknown in practice.

As discussed in Section~\ref{sec_discussion}, the unexplored bridge between discrepancy-based \textsc{abc} and the Rademacher complexity introduced in this article can also be leveraged to derive even more general theory by exploiting the active literature on the Rademacher complexity. For example, combining our perspective with the recent unified treatment of \textsc{ips} and $f$-divergences \citep{agrawal2021optimal,birrell2022f} might set the premises to derive similarly-interpretable and general results for other  discrepancies employed within \textsc{abc}, such as the Kullback--Leibler  divergence \citep{jiang2018approximate} and Hellinger distance \citep{frazier2020robust}. More generally, our contribution can have  implications even beyond \textsc{abc}, in particular on generalized Bayesian inference via pseudo-posteriors based on discrepancies \citep[][]{bissiri2016general,cherief2020mmd,matsubara2021robust,frazier2024}. Proofs and additional results  can be found in the Supplementary Material.


\section{Integral probability semimetrics and  Rademacher complexity} 
\label{sec_abc}
Denote with $ y_{1:n}=(y_1,\ldots,y_n) \in \mathcal{Y}^n $ an i.i.d.\ sample from $ \mu^* \in  \mathcal{P}(\mathcal{Y}) $, where $ \mathcal{P}(\mathcal{Y}) $ is the space of probability measures on $ \mathcal{Y} $, and assume that $ \mathcal{Y} $ is a metric space endowed with  distance $\rho$; see Appendix~C in the Supplementary Material for extensions to non-i.i.d. settings.  Given a statistical model $ \{ \mu_\theta: \theta \in \Theta  \subseteq \mathbb{R}^{p} \} $ in $ \mathcal{P}(\mathcal{Y}) $ and a prior distribution $ \pi $ on $ \theta $, rejection \textsc{abc} iteratively samples $ \theta $ from $ \pi $, draws a synthetic i.i.d.\ sample $ z_{1:m}=(z_1,\ldots,z_m) $ from $ \mu_\theta $, and retains $ \theta $ as a sample from the \textsc{abc} posterior if the discrepancy $\Delta(z_{1:m},y_{1:n}) $ between $z_{1:m}$ and $y_{1:n}$ is below a user-selected threshold $\varepsilon_n \geq 0 $. Although $ m $ may differ from  $ n $, here we follow common practice in \textsc{abc} theory \citep[e.g.,][]{bernton2019approximate,frazier2020robust} and set $ {m=n} $. Rejection \textsc{abc} does not sample from the exact posterior $ \pi_n(\theta) \propto \pi(\theta) \mu_\theta^n(y_{1:n}), $ but rather from the \textsc{abc} posterior  defined as
\begin{eqnarray*}
\pi_n^{(\varepsilon_n)}(\theta) \propto \pi(\theta) \int_{\mathcal{Y}^n} \mathds{1}\{\Delta(z_{1:n},y_{1:n})\leq\varepsilon_n\} \ \mu_\theta^n(dz_{1:n}),
\end{eqnarray*}
 whose properties clearly depend on the chosen discrepancy $\Delta(\cdot,\cdot) $. Within summary-based \textsc{abc}, $ \Delta(z_{1:n},y_{1:n}) $ is a suitable distance, typically Euclidean, among summaries computed from the synthetic sample $z_{1:n}$ and the observed data $y_{1:n}$. However, recalling, e.g.,  \citet{marin2014relevant} and \citet{frazier2018asymptotic}, the identification of summaries that do not lead to information loss is challenging for those complex models requiring \textsc{abc} implementations.

To overcome these challenges,  \textsc{abc} literature has progressively moved towards the adoption of discrepancies $ \mathcal{D}: \mathcal{P}(\mathcal{Y}) \times \mathcal{P}(\mathcal{Y}) \to [0,\infty] $ between the empirical distributions of the synthetic  and observed data, that is
\begin{eqnarray*}\label{eq_dist_empirical}
\Delta(z_{1:n},y_{1:n}) = \mathcal{D}( \hat{\mu}_{z_{1:n}}, \hat{\mu}_{y_{1:n}})=\mathcal{D} ( n^{-1}\textstyle{\sum\nolimits_{i=1}^n} \delta_{z_i}, n^{-1}\textstyle{\sum\nolimits_{i=1}^n} \delta_{y_i} ),
\end{eqnarray*}
where $\delta_{x}$ is the Dirac delta at a generic $x \in \mathcal{Y}$. Popular examples are \textsc{abc} versions based on \textsc{mmd}, \textsc{kl} divergence, Wasserstein distance, energy statistic, Hellinger and Cramer--von Mises distances, and $\gamma$-divergence, whose limiting properties have been studied in \citet{park2016k2}, \citet{jiang2018approximate}, \citet{bernton2019approximate}, \citet{nguyen2020approximate}, \citet{frazier2020robust}  and \citet{fujisawa2021gamma} under different asymptotic regimes and with a common reliance on existence assumptions to  ease the proofs. As a first step toward our unified and constructive theoretical framework, we shall emphasize that, although most of the above contributions treat the discrepancies separately, some of these choices share, in fact, a common origin. For example, \textsc{mmd}, Wasserstein distance and the energy statistic all belong to the \textsc{ips} class \citep{muller1997integral} in Definition~\ref{def_IPS}. Such a class also includes summary-based distances.

\begin{definition}[Integral probability semimetric --- \textsc{ips}]  \label{def_IPS} 
Let $ \mathfrak{F}$ be a class of measurable functions $ f: \mathcal{Y} \to \mathds{R} $. Then, the integral probability semimetric $\mathcal{D}_\mathfrak{F}$ between $ \mu_1,\mu_2 \in \mathcal{P}(\mathcal{Y}) $ is defined as
\begin{eqnarray*}
\mathcal{D}_\mathfrak{F}(\mu_1,\mu_2):=\sup\nolimits_{f\in\mathfrak{F}} \left| \int_{\mathcal{Y}} f d\mu_1 - \int_{\mathcal{Y}} f d\mu_2 \right|,
\end{eqnarray*}
where $\mathfrak{F}$ is the pre-specified class of functions that characterizes $\mathcal{D}_\mathfrak{F}$.
\end{definition}

Examples~\ref{exm_wass}--\ref{exm_sum} show that routinely-employed discrepancies both  in summary-free and summary-based \textsc{abc} \citep[e.g.,][]{park2016k2,bernton2019approximate,nguyen2020approximate,drovandi2022comparison} are, in fact, \textsc{ips} with a known characterizing family $\mathfrak{F}$ that uniquely identifies each discrepancy. 

\begin{example} \label{exm_wass}
{\em (Wasserstein-1 distance)}.	Let us consider the Lipschitz seminorm defined as $ || f ||_L := \sup \{ |f(x)-f(x')|/\rho(x,x'): x \neq x' \mbox{ in } \mathcal{Y} \} $. If $ \mathfrak{F} = \{f : ||f||_L \leq 1 \} $, then $ \mathcal{D}_\mathfrak{F} $ coincides with the Kantorovich metric. Recalling the Kantorovich--Rubinstein theorem \citep[e.g.,][]{dudley2018real}, when $\mathcal{Y}$ is
separable this metric coincides with the dual representation of the Wasserstein-1 distance, which is therefore recovered as an example of \textsc{ips}. Recall also that such a distance is defined  on the space of probability measures having finite first moment \citep[e.g.,][]{sriperumbudur2009integral,bernton2019approximate}.
\end{example}

\begin{example} {\em (\textsc{mmd} and energy distance).} \label{exm_mmd}
Given a positive-definite kernel $ k(\cdot,\cdot) $ on $ \mathcal{Y} \times \mathcal{Y} $, let $ \mathfrak{F}=\{f:  ||f||_\mathcal{H} \leq 1\} $, where $ \mathcal{H} $ is the reproducing kernel Hilbert space corresponding to $ k(\cdot,\cdot) $. Then $\mathcal{D}_\mathfrak{F} $ is the \textsc{mmd} \citep[e.g.,][]{muandet2017kernel}. When $ k(\cdot,\cdot) $ is characteristic, then $\mathcal{D}_\mathfrak{F} $ is not only a semimetric, but also a metric --- i.e., $\mathcal{D}_\mathfrak{F}(\mu_1,\mu_2)=0 $ implies $\mu_1=\mu_2$. Relevant examples of routinely-implemented characteristic kernels in $\mathcal{Y} =\mathbb{R}^d$ are the Gaussian {$\exp(-\|x-x'\|^2/\sigma^2)$} and Laplace {$\exp(-\|x-x'\|/\sigma)$} ones \citep[]{muandet2017kernel}. Note that \textsc{mmd} is also directly related to the energy distance, due to the correspondence among positive-definite kernels and negative-definite functions \citep[][]{sejdinovic2013equivalence}.
\end{example}

\begin{example}\label{exm_sum} {\em (Summary-based distances).} 
Classical summary-based \textsc{abc} employs distances among a finite set of  summaries $f_1, \ldots, f_K$  \citep[e.g.,][]{drovandi2022comparison}. Interestingly, these distances can be recovered as a special case of \textsc{mmd} with  suitably-defined kernel and, hence, are guaranteed to belong to the \textsc{ips} class. In particular, for a pre-selected vector of summaries $f(x)=[f_1(x), \ldots, f_K(x)] \in \mathcal{H}=\mathbb{R}^K$ equipped with the standard Euclidean norm $\|f(x)\|_{\mathcal{H}}^2 = \left< f(x),f(x)\right> = f^2_1(x)+ \cdots +f^2_K(x)$, one can define the kernel $k(x,x')=\left< f(x),f(x')\right>$ to obtain classical summary-based \textsc{abc}  recasted within the \textsc{mmd}  framework. As a result, popular   kernels, such as the Gaussian one, relying on an infinite dimensional feature space can be interpreted as a limiting version of summary-based  \textsc{abc}.
\end{example}

Although Examples \ref{exm_wass}--\ref{exm_sum}  characterize the most popular instances of \textsc{ips} discrepancies employed in \textsc{abc}, it shall be emphasized that other  interesting semimetrics belong to the \textsc{ips} class \citep[e.g.,][]{sriperumbudur2009integral,birrell2022f}.  Two relevant ones are the total variation  (\textsc{tv})  and Kolmogorov--Smirnov distances discussed in the Supplementary Material.

While \textsc{abc} based on discrepancies, such as those presented above, overcomes the need to pre-select summaries, it still requires the choice of a discrepancy. This motivates theory  on the asymptotic  properties of the induced \textsc{abc} posterior, under general discrepancies, for  $n \rightarrow \infty$ and relevant thresholding schemes. Considering here a unified perspective which applies to the whole \textsc{ips} class, the  current theory focuses both on fixed $\varepsilon_n=\varepsilon$ settings, and also on $\varepsilon_n \rightarrow \varepsilon^*$ regimes, where $\varepsilon^* = \inf\nolimits_{\theta \in \Theta} \mathcal{D}_\mathfrak{F} (\mu_\theta, \mu^*)$ is the lowest attainable discrepancy between the assumed model and the data generating process \citep{jiang2018approximate, bernton2019approximate, nguyen2020approximate,frazier2020robust,frazier2020model,fujisawa2021gamma}. Under these settings, available theory investigates whether the $\varepsilon_n$-control on $ \mathcal{D}_\mathfrak{F} ( \hat{\mu}_{z_{1:n}}, \hat{\mu}_{y_{1:n}})$   established by rejection-\textsc{abc}  translates into meaningful convergence results, and upper  bounds on the rates of concentration for the sequence of discrepancy-based \textsc{abc} posteriors around the true data generating process $\mu^*$, both under correctly specified models where $\mu^*=\mu_{\theta^*}$ for some $\theta^* \in \Theta$, and  in misspecified contexts where $\mu^*$ does not belong to $ \{ \mu_\theta: \theta \in \Theta  \subseteq \mathbb{R}^{p} \} $.

A common strategy to derive the aforementioned results is to study  convergence of the empirical distributions $\hat{\mu}_{z_{1:n}}$ and $\hat{\mu}_{y_{1:n}}$ to the corresponding truths $\mu_\theta$ and  $\mu^*$, in the chosen discrepancy $ \mathcal{D}_\mathfrak{F}$, after noticing that, since $\mathcal{D}_\mathfrak{F}$ is a semimetric, we have
\begin{eqnarray}
\mathcal{D}_\mathfrak{F} (\mu_\theta, \mu^*) \leq  \mathcal{D}_\mathfrak{F} (\hat{\mu}_{z_{1:n}},\mu_\theta) +\mathcal{D}_\mathfrak{F} (\hat{\mu}_{z_{1:n}},\hat{\mu}_{y_{1:n}})+  \mathcal{D}_\mathfrak{F}(\hat{\mu}_{y_{1:n}},\mu^*).
\label{ineq}
\end{eqnarray}
To this end, available theoretical studies pre-assume suitable convergence results for the empirical measures along with specific concentration inequalities regulated by non-explicit sequences of control functions; see Assumptions 1--2 in Proposition 3 of {\citet{bernton2019approximate}, and Assumptions A1--A2 in \citet{nguyen2020approximate}. Alternatively, it is possible to directly require that  $\mathcal{D}_\mathfrak{F} (\hat{\mu}_{z_{1:n}},\hat{\mu}_{y_{1:n}}) \rightarrow \mathcal{D}_\mathfrak{F} (\mu_\theta, \mu^*)$ almost surely as $n \rightarrow \infty$  \citep[see][Theorems 1 and 2, respectively]{jiang2018approximate, nguyen2020approximate}. However, as highlighted by the same authors, these conditions {\em (i)} can be  difficult to verify for several discrepancies, {\em (ii)}  do not allow to assess whether some of these discrepancies can achieve convergence and concentration uniformly over $ \mathcal{P}(\mathcal{Y})$,  and  {\em (iii)}  often yield bounds which hinder an in-depth understanding of the factors regulating these limiting properties. In addition, when the above assumptions are not met, the asymptotic behavior of  \textsc{abc} posteriors remains yet unexplored. 

As a first step towards addressing the above issues, notice that the aforementioned results, when contextualized within the \textsc{ips} class, are inherently related to the richness of the known class of functions $\mathfrak{F}$ that identifies each $ \mathcal{D}_\mathfrak{F} $; see Definition~\ref{def_IPS}. Intuitively, if $\mathfrak{F}$ is too rich, it is always possible to find a function $f \in \mathfrak{F}$ yielding large discrepancies between the empiricals and the corresponding truths even when $\hat{\mu}_{z_{1:n}}$ and $\hat{\mu}_{y_{1:n}}$ are arbitrarily close to $\mu_\theta$ and  $\mu^*$, respectively. Hence, $\mathcal{D}_\mathfrak{F} (\hat{\mu}_{z_{1:n}},\mu_\theta)$ and $ \mathcal{D}_\mathfrak{F}(\hat{\mu}_{y_{1:n}},\mu^*)$ will  remain large with positive probability, making the  triangle  inequality in \eqref{ineq} of limited interest, since a low $\mathcal{D}_\mathfrak{F} (\hat{\mu}_{z_{1:n}},\hat{\mu}_{y_{1:n}})$ will not necessarily imply a small $\mathcal{D}_\mathfrak{F} (\mu_\theta, \mu^*)$. In fact, in this context, it is clearly even not guaranteed that $\mathcal{D}_\mathfrak{F} (\hat{\mu}_{z_{1:n}},\hat{\mu}_{y_{1:n}})$ will converge to $\mathcal{D}_\mathfrak{F} (\mu_\theta, \mu^*)$. This suggests that the limiting properties of the \textsc{abc} posterior are inherently related to the richness of $\mathfrak{F}$.  In Section~\ref{sec_abc_posterior} we prove that this is the case when such a richness is measured through the notion of Rademacher complexity clarified in Definition~\ref{def_rade}; see also  Chapter~4 in \citet{wainwright2019high} for an introduction.

\begin{definition}[Rademacher complexity] \label{def_rade}
Given a vector $ x_{1:n} = (x_1, \ldots, x_n) $ of  i.i.d.\ random variables with distribution  $\mu \in  \mathcal{P}(\mathcal{Y})$, the Rademacher complexity $ \mathfrak{R}_{\mu,n}(\mathfrak{F}) $ of a class $\mathfrak{F}$ of real-valued measurable functions is 
\begin{eqnarray*}
 \mathfrak{R}_{\mu,n}(\mathfrak{F}) = \mathbb{E}_{x_{1:n},\epsilon_{1:n}}\left[ \sup\nolimits_{f\in\mathfrak{F}}\left| (1/n) {\textstyle \sum\nolimits}_{i=1}^{n} \epsilon_i f(x_i) \right| \right], 
 \end{eqnarray*} 
where $\epsilon_{1},\dots,\epsilon_n$ denote i.i.d.\ Rademacher variables, that is $\mathbb{P}(\epsilon_i=1)=\mathbb{P}(\epsilon_i=-1)=1/2$.
\end{definition}

As is clear from the above definition, high values of the Rademacher complexity $\mathfrak{R}_{\mu,n}(\mathfrak{F})$ mean that $\mathfrak{F}$ is rich enough to contain functions that can closely capture, on average, even the behavior of full-noise vectors. Hence, if $\mathfrak{R}_{\mu,n}(\mathfrak{F})$ is bounded away from zero for every $n$, then $\mathfrak{F}$ might be excessively rich for statistical purposes.  Conversely, if $\mathfrak{R}_{\mu,n}(\mathfrak{F})$ goes to zero as $n$ diverges, then $\mathfrak{F}$ is a more parsimonious class. Lemma~\ref{lemma_rademacher} formalizes this intuition.

\begin{lemma}[Theorem 4.10 and Proposition 4.12 in \citet{wainwright2019high}]
\label{lemma_rademacher}
Let $x_{1:n}$ be i.i.d.\ from some distribution $\mu \in  \mathcal{P}(\mathcal{Y})$. Then, for any $b$-uniformly bounded class $\mathfrak{F}$, i.e., any class $ \mathfrak{F} $ of functions $f$ such that $ \|f\|_\infty\leq b$, any integer $n \geq 1$ and  scalar $\delta \geq 0$, it holds that
 \begin{eqnarray}\label{eq1}
\mathbb{P}_{x_{1:n}} \left[\mathcal{D}_{\mathfrak{F}}(\hat{\mu}_{x_{1:n}},\mu) \leq  2 \mathfrak{R}_{\mu,n}(\mathfrak{F})+\delta \right]\geq 1-\exp(-n \delta^2/2b^2),
\end{eqnarray}
and $\mathbb{P}_{x_{1:n}} [ \mathcal{D}_{\mathfrak{F}}(\hat{\mu}_{x_{1:n}},\mu) \geq  \mathfrak{R}_{\mu,n}(\mathfrak{F})/2- \sup_{f \in \mathfrak{F}}| \mathbb{E}(f)|/2 n^{1/2}-\delta ]\geq 1-\exp(-n \delta^2/2b^2)$.
\end{lemma}

Lemma~\ref{lemma_rademacher} provides bounds for the probability that $\mathcal{D}_{\mathfrak{F}}(\hat{\mu}_{x_{1:n}},\mu)$ takes values below a multiple and above a fraction of the Rademacher complexity of $\mathfrak{F}$. Recalling the previous discussion on concentration of \textsc{abc} posteriors, this result is crucial to study the convergence of both $\hat{\mu}_{z_{1:n}}$ and $\hat{\mu}_{y_{1:n}}$ to $\mu_\theta$ and $\mu^*$, respectively. Our theory in Section~\ref{sec_abc_posterior} proves that this is possible when $ \mathfrak{R}_{n}(\mathfrak{F}) := \sup_{\mu \in  \mathcal{P}(\mathcal{Y})} \mathfrak{R}_{\mu,n}(\mathfrak{F}) $ goes to $0$ as $n \rightarrow \infty$. According to Lemma~\ref{lemma_rademacher}, this is a necessary and sufficient condition for $\mathcal{D}_{\mathfrak{F}}(\hat{\mu}_{x_{1:n}},\mu) \rightarrow 0 $ as $n \rightarrow \infty$, in $ \mathbb{P}_{x_{1:n}}$--probability, uniformly over $\mathcal{P}(\mathcal{Y})$. See Appendix C in the Supplementary Material for extensions to non-i.i.d.\ settings.


\section{Asymptotic properties of discrepancy-based \textsc{abc} posterior distributions}
\label{sec_abc_posterior}
As anticipated in Section~\ref{sec_abc}, the asymptotic results we derive leverage Lemma~\ref{lemma_rademacher} to connect the properties of the \textsc{abc} posterior under an \textsc{ips} discrepancy $\mathcal{D}_\mathfrak{F} $  with the behavior of the Rademacher complexity for the associated family $ \mathfrak{F} $. Section~\ref{sec_abc_lim} clarifies the importance of this bridge by studying the limiting properties of the \textsc{abc} posterior when  $\varepsilon_n=\varepsilon$ is fixed and $n \rightarrow \infty$, including in yet-unexplored situations where convergence guarantees for the empirical measures and for the discrepancy among such measures are not necessarily available. The theoretical results we derive in this case suggest that the use of discrepancies whose Rademacher complexity decays to zero in the limit is a sufficient condition to obtain strong convergence guarantees.  On the basis of these findings, Section~\ref{sec_abc_con} focuses on the regime  $\varepsilon_n \rightarrow \varepsilon^*= \inf\nolimits_{\theta \in \Theta} \mathcal{D}_{\mathfrak{F}}(\mu_\theta, \mu^*)$ as $n \rightarrow \infty$, and derives unified and informative concentration bounds that  hold uniformly over $\mathcal{P}(\mathcal{Y})$ and are based on more constructive conditions than those employed in the available discrepancy-specific theory. More specifically,  to prove our theoretical results in Sections~\ref{sec_abc_lim}--\ref{sec_abc_con},  we rely on some or all of the following assumptions:
\begin{enumerate}[(I)]
	\item \label{C1} the observed data $ y_{1:n} $	are i.i.d.\ from a data generating process $ \mu^* $;
	\item \label{C2} there exist $L,c_\pi \in  \mathds{R}^+$ such that, for $\bar{\varepsilon}$ small, $ \pi(\{\theta: \mathcal{D}_{\mathfrak{F}}(\mu_\theta,\mu^*) \leq \varepsilon^* + \bar{\varepsilon}\}) \geq c_\pi \bar{\varepsilon}^L$;
	\item \label{C3} $ \mathfrak{F}$ is $b$-uniformly bounded; i.e., there exists $b \in \mathds{R}^+$ such that $ ||f||_{\infty} {\leq b} $ for any $ f \in \mathfrak{F} $;
	\item \label{C4} $\mathfrak{R}_{n}(\mathfrak{F}) := \sup_{\mu \in  \mathcal{P}(\mathcal{Y})} \mathfrak{R}_{\mu,n}(\mathfrak{F}) $ goes to zero as $ n \to \infty  $.
\end{enumerate}

Condition (\ref{C1}) is the only assumption made on the data generating process and is present in, e.g., \citet{nguyen2020approximate} and in the supplementary materials of \citet{bernton2019approximate}. Although the theory we derive in Appendix C of the Supplementary Material relaxes (\ref{C1}) to study convergence and concentration also beyond the i.i.d.\ context, it shall be emphasized that some of the assumptions considered in the literature may not hold even in i.i.d.\ settings. Hence,  an improved understanding of the \textsc{abc} properties under (\ref{C1}) is crucial to clarify the range of applicability and potential limitations  of available existence theory under more complex, possibly non-i.i.d.\ regimes. In fact, as shown within Section~\ref{sec_abc_lim}, certain discrepancies may yield posteriors that are either not well-defined in the limit or lack strong convergence guarantees even in  i.i.d.\ settings. Assumption (\ref{C2}) is standard in Bayesian asymptotics and \textsc{abc} theory \citep[][]{bernton2019approximate, nguyen2020approximate,frazier2020robust}, and requires that sufficient prior mass is placed on those parameter values yielding models $\mu_{\theta}$ close to $\mu^*$, under $ \mathcal{D}_{\mathfrak{F}}$.  Finally, (\ref{C3})--(\ref{C4}) provide conditions on $\mathfrak{F}$ to ensure that $ \mathcal{D}_{\mathfrak{F}}$ yields meaningful concentration bounds, leveraging Lemma~\ref{lemma_rademacher}.  Crucially, (\ref{C3})--(\ref{C4}) are made directly on the user-selected discrepancy $\mathcal{D}_\mathfrak{F} $, rather than on the data generating process or on the statistical model, thereby providing general constructive conditions that are more realistic to verify in practice compared to regularity assumptions on the unknown $\mu^*$ or on $ \{ \mu_\theta: \theta \in \Theta  \subseteq \mathbb{R}^{p} \} $, while guiding the choice of $\mathcal{D}_ \mathfrak{F}$. 

Notice that (\ref{C3})--(\ref{C4}) depend, in principle, on  $\mathcal{P}(\mathcal{Y})$. Nonetheless, as clarified within Section~\ref{sec_ass}, available bounds on the Rademacher complexity ensure that these two conditions  hold for the relevant \textsc{ips} discrepancies in Examples~\ref{exm_wass}--\ref{exm_sum} under either no  constraints on $\mathcal{P}(\mathcal{Y})$ --- thereby circumventing such a dependence --- or for suitable spaces $\mathcal{Y}$ (e.g., bounded ones). Unlike for constraints on  the unknown $\mu^*$ and on $ \{ \mu_\theta: \theta \in \Theta  \subseteq \mathbb{R}^{p} \} $, conditions on $\mathcal{Y}$ can be verified in practice; see  Examples \ref{exm_wass1}--\ref{exm_sum1}, and refer to the Supplementary Material for the analysis of other relevant \textsc{ips} discrepancies. Moreover, Assumptions (\ref{C3}) and (\ref{C4}) are generally not stronger than those implicit in the current theory on discrepancy-based \textsc{abc}, and are inherently related with the notion of uniform Glivenko--Cantelli classes \citep[see e.g.,][Chapter~4]{wainwright2019high}, arguably a minimal sufficient requirement for establishing convergence and concentration properties that, unlike those currently derived,  hold uniformly over $\mathcal{P}(\mathcal{Y})$. This latter point is further clarified in Propositions \ref{corr_MMD_unb} and \ref{corr_Wasserstein_unb}, which state concentration results for \textsc{mmd} with unbounded kernel in $\mathbb{R}^{d}$ (which also includes \textsc{abc} with unbounded summary statistics) and Wasserstein-1 distance in $\mathbb{R}^{d}$, respectively. Although these two  discrepancies meet (\ref{C3})--(\ref{C4}) for suitable constraints on $\mathcal{Y}$, in general unbounded $\mathcal{Y}$ settings these conditions are no more guaranteed for such discrepancies. In this case, Propositions~\ref{corr_MMD_unb} and \ref{corr_Wasserstein_unb}  clarify that  concentration properties can still be derived, but at the expense of additional constraints on $\mu^*$ and $ \{ \mu_\theta: \theta \in \Theta  \subseteq \mathbb{R}^{p} \} $, which yield results that hold uniformly only in a suitable subset of $\mathcal{P}(\mathcal{Y})$, rather than over the entire  $\mathcal{P}(\mathcal{Y})$. It shall be also emphasized that the convergence theory derived in Section~\ref{sec_abc_lim} does not necessarily require a vanishing Rademacher complexity to obtain informative results. However, it also clarifies that, when this condition is verified, the limiting pseudo-posterior inherits the same threshold-control on $\mathcal{D}_\mathfrak{F} (\mu_\theta, \mu^*)$ as the one established by rejection \textsc{abc}  on $\mathcal{D}_\mathfrak{F} (\hat{\mu}_{z_{1:n}},\hat{\mu}_{y_{1:n}})$. 

As anticipated above, Sections~\ref{sec_abc_lim}--\ref{sec_abc_con} consider the two scenarios where the sample size $ n $ diverges to infinity, while the \textsc{abc} threshold $\varepsilon_n$ is either fixed at $\varepsilon$ or progressively shrinks to $\varepsilon^* = \inf\nolimits_{\theta \in \Theta} \mathcal{D}_{\mathfrak{F}}(\mu_\theta, \mu^*)$. If the model is well-specified, i.e., $ \mu^*= \mu_{\theta^*} $ for some, possibly not unique, $ \theta^* \in \Theta $, then $ \varepsilon^* =0 $. Although the regime characterized by $ \varepsilon_n \to 0 $ with fixed $ n $ is also of interest, this setting can be readily addressed under the unified \textsc{ips} class via a direct adaptation of available  results in, e.g., \citet{jiang2018approximate} and \citet{bernton2019approximate}; see also  \citet{miller2019robust} for results in the context of coarsened posteriors under the regime with  $ n \to \infty $ and fixed $ \varepsilon_n= \varepsilon$, related to those in  \citet{jiang2018approximate}.  


\subsection{Limiting behavior with fixed $\varepsilon_n=\varepsilon$ and $n \rightarrow \infty$}
\label{sec_abc_lim}
Current \textsc{abc} theory for the setting $\varepsilon_n=\varepsilon$ fixed and $n \rightarrow \infty$, studies convergence of the discrepancy-based \textsc{abc} posterior
\begin{eqnarray*}
{\pi^{(\varepsilon)}_{n}(\theta)} \propto  \pi(\theta) \int_{\mathcal{Y}^n} \mathds{1}\{\mathcal{D}_\mathfrak{F}(\hat{\mu}_{z_{1:n}},\hat{\mu}_{y_{1:n}})\leq\varepsilon\}  \mu_\theta^n(dz_{1:n}),
\end{eqnarray*}
 under the key assumption that $\mathcal{D}_\mathfrak{F} (\hat{\mu}_{z_{1:n}},\hat{\mu}_{y_{1:n}}) \rightarrow \mathcal{D}_\mathfrak{F} (\mu_\theta, \mu^*)$ almost surely as $n \rightarrow \infty$ \citep[see, e.g.,][Theorem 1]{jiang2018approximate}. When such a condition is met, it can be shown that $\lim_{n \rightarrow \infty} {\pi^{(\varepsilon)}_{n}(\theta)}  \propto \pi(\theta)\mathds{1}\{\mathcal{D}_\mathfrak{F}(\mu_{\theta},\mu^*)\leq\varepsilon\} $. This means that the limiting \textsc{abc} posterior coincides with the prior constrained within the support $\{\theta: \mathcal{D}_\mathfrak{F}(\mu_{\theta},\mu^*)\leq\varepsilon\}$, where $\varepsilon$ is the same threshold employed for $\mathcal{D}_\mathfrak{F}(\hat{\mu}_{z_{1:n}},\hat{\mu}_{y_{1:n}})$. 

Although this is a key result, verifying the convergence assumption behind the above finding can be difficult for several discrepancies \citep[e.g.,][]{jiang2018approximate}. More crucially, such theory fails to provide informative results in those situations where $\mathcal{D}_\mathfrak{F} (\hat{\mu}_{z_{1:n}},\hat{\mu}_{y_{1:n}})$ is not guaranteed to converge to $\mathcal{D}_\mathfrak{F} (\mu_\theta, \mu^*)$. In fact, under this setting, the asymptotic properties of discrepancy-based \textsc{abc} posteriors have been overlooked, and it is  not even clear whether a well-defined ${\pi^{(\varepsilon)}_{n}}$  exists in the  limit. Indeed, for specific choices of the discrepancy and threshold,  the \textsc{abc} posterior may not be well-defined even for fixed $n$. For example, when $\mathcal{D}_\mathfrak{F}$ is the \textsc{tv} distance (see Appendix A in the Supplementary Material), then $\mathcal{D}_\mathfrak{F}(\hat{\mu}_{z_{1:n}},\hat{\mu}_{y_{1:n}})=1$, almost surely, whenever  $z_{1:n}$ and $y_{1:n}$ are from two continuous distributions $\mu_{\theta}$ and $\mu^*$, respectively.  Hence, for any $\varepsilon < 1$, the \textsc{abc} posterior is not defined, even if  $\mu_{\theta}$ and $\mu^*$ are not mutually singular. Note that, as clarified in Appendix A, the Rademacher complexity of the \textsc{tv} distance is always bounded away from zero in such a setting. In this case and, more generally, for \textsc{ips} discrepancies whose $\mathfrak{R}_{n}(\mathfrak{F})$ does not vanish to zero, Lemma~\ref{lemma_rademacher} together with the triangle inequality in \eqref{ineq} still allow to derive new informative upper and lower bounds for limiting acceptance probabilities which further guarantee that $\mathcal{D}_\mathfrak{F} (\mu_\theta, \mu^*)$ can be bounded from above,  for a sufficiently large $n$, by $\varepsilon + 4\mathfrak{R}_{n}(\mathfrak{F})$, even without assuming that $\mathcal{D}_\mathfrak{F} (\hat{\mu}_{z_{1:n}},\hat{\mu}_{y_{1:n}}) \rightarrow \mathcal{D}_\mathfrak{F} (\mu_\theta, \mu^*)$ almost surely for $n \rightarrow \infty$, as in the current convergence theory. Theorem~\ref{conv}  formalizes this intuition for the whole class of  \textsc{abc} posteriors arising from discrepancies in the \textsc{ips} family.

 \begin{theorem} \label{conv}
Let $\mathcal{D}_{\mathfrak{F}}$ be an \textsc{ips} as in Definition~\ref{def_IPS}. Moreover, assume (\ref{C1}), (\ref{C3}), and let $c_{\mathfrak{F}}= 4 \lim\sup \mathfrak{R}_n(\mathfrak{F})$. Then, the acceptance probability $ p_n=\int_{\Theta}\int_{\mathcal{Y}^n} \mathds{1}\{\mathcal{D}(\hat{\mu}_{z_{1:n}},\hat{\mu}_{y_{1:n}})\leq\varepsilon\}  \mu_\theta^n(dz_{1:n})\pi(d\theta) $ of  rejection \textsc{abc} with discrepancy $\mathcal{D}_{\mathfrak{F}}$ satisfies, for any fixed $ \varepsilon >0$, 
\begin{eqnarray}
\label{eq_pABC}
\qquad \pi\{\theta: \mathcal{D}_{\mathfrak{F}}(\mu_\theta,\mu^* )\leq  \varepsilon - c_{\mathfrak{F}}\} - e_n  \leq   p_n   \leq    \pi\{\theta: \mathcal{D}_{\mathfrak{F}}(\mu_\theta,\mu^* )\leq  \varepsilon + c_{\mathfrak{F}}\} + e_n,
\end{eqnarray}
almost surely with respect to $y_{1:n} \stackrel{\mbox{\normalfont \scriptsize i.i.d.}}{\sim} \mu^{*}$, as $n \rightarrow \infty$, where $e_n=\exp(-\sqrt{n}/(2b^2))=o(1)$.

In particular, whenever $\varepsilon > \tilde{\varepsilon}+c_{\mathfrak{F}} $, with  $\tilde{\varepsilon}= \inf \{ \epsilon>0: {\pi\{\theta:\mathcal{D}_{\mathfrak{F}}(\mu_\theta,\mu^* ) \leq \epsilon \}}>0 \}$,  the probability $p_n$ is bounded away from $0$, for  $n$ large enough. This implies that, for every fixed $\varepsilon>\tilde{\varepsilon}+c_{\mathfrak{F}}$, the \textsc{abc}  posterior  is always well-defined for a large-enough $n$ and its support is almost surely included asymptotically  in the set $\{\theta: \mathcal{D}_{\mathfrak{F}}(\mu_\theta,\mu^*) \leq \varepsilon + 4 \mathfrak{R}_n(\mathfrak{F}) \}$. Namely, 
\begin{eqnarray*}
 \pi_n^{(\varepsilon)}\{\theta: \mathcal{D}_{\mathfrak{F}}(\mu_\theta,\mu^* ) \leq \varepsilon + 4 \mathfrak{R}_n(\mathfrak{F})\} \rightarrow 1, 
 \end{eqnarray*}
 for any fixed $\varepsilon>\tilde{\varepsilon}+c_{\mathfrak{F}}$, almost surely with respect to $y_{1:n} \stackrel{\mbox{\normalfont \scriptsize i.i.d.}}{\sim} \mu^{*}$, as $n \rightarrow \infty$. 
 \end{theorem}

Theorem~\ref{conv} clarifies that, even when relaxing the assumptions behind current theory, the  Rademacher complexity framework still allows to obtain guarantees on the acceptance probabilities,  existence and limiting support  of the \textsc{abc} posterior. Such a relaxation  highlights the key role of the richness of $\mathfrak{F}$ in driving these properties. As previously discussed, the higher is {$ \mathfrak{R}_n(\mathfrak{F})$} the lower are the guarantees that $\mathcal{D}_\mathfrak{F}(\hat{\mu}_{z_{1:n}},\hat{\mu}_{y_{1:n}})$ is close enough to $\mathcal{D}_{\mathfrak{F}}(\mu_\theta,\mu^*)$ in the large data limit. In fact, recalling Section~\ref{sec_abc}, when {$ \mathfrak{R}_n(\mathfrak{F})$} does not vanish with $n$ there are no guarantees that the convergence of  $\mathcal{D}_\mathfrak{F} (\hat{\mu}_{z_{1:n}},\hat{\mu}_{y_{1:n}})$ to $\mathcal{D}_\mathfrak{F} (\mu_\theta, \mu^*)$ assumed in the literature is verified. Nonetheless, the inequalities in Lemma \ref{lemma_rademacher}  still yield informative results that   translate these arguments into an inflation of the superset containing the support of the  limiting \textsc{abc} posterior, and in a more conservative choice for the threshold $\varepsilon$ to ensure that such a posterior is well-defined for large $n$. Notice that in Theorem~\ref{conv} the dependence on the dimension of the parameter and data space is embedded within $\mathcal{D}_{\mathfrak{F}}(\mu_\theta,\mu^* )$; see Equation (6) in the Appendix of \citet{cherief2022finite} for an example making this dependence explicit.

Recall that, for the \textsc{abc} posterior to be  well-defined in the limit (with a fixed threshold $\varepsilon$),  the acceptance probability $p_n$ needs to be bounded away from zero. According to Theorem~\ref{conv}, as $n \rightarrow \infty$ this is guaranteed whenever $ \varepsilon >\tilde{\varepsilon}+c_{\mathfrak{F}}$. For instance, recalling Section~\ref{sec_abc}, in misspecified models $ \mathcal{D}_{\mathfrak{F}}(\mu_\theta,\mu^* ) \geq \varepsilon^*$. As a consequence, $ \varepsilon > c_{\mathfrak{F}}$ is not sufficient to guarantee that $ \pi\{\theta: \mathcal{D}_{\mathfrak{F}}(\mu_\theta,\mu^* )\leq  \varepsilon - c_{\mathfrak{F}}\}$ in \eqref{eq_pABC} is bounded away from zero (in fact, when  $c_{\mathfrak{F}}<\varepsilon <c_{\mathfrak{F}}+  \varepsilon^*$ the lower bound in \eqref{eq_pABC} approaches zero from below). To achieve this result,  $ \varepsilon - c_{\mathfrak{F}}$ needs to define a ball with radius not lower than the one to which the prior is already guaranteed~to~assign strictly positive mass. This implies $\varepsilon >\tilde{\varepsilon}+c_{\mathfrak{F}} $, with $c_{\mathfrak{F}}= 4 \lim\sup \mathfrak{R}_n(\mathfrak{F})$.

The results in Theorem~\ref{conv} further clarify that, when $\mathfrak{R}_n(\mathfrak{F}) \rightarrow 0$ as $n \rightarrow \infty$, stronger and uniform convergence statements can be derived, both in correctly-specified and in misspecified models. More specifically, in this setting, the acceptance probability $p_n$ in Theorem~\ref{conv} approaches $ \pi\{\theta: \mathcal{D}_{\mathfrak{F}}(\mu_\theta,\mu^* )\leq  \varepsilon \}$ in the limit. Moreover, the \textsc{abc} posterior is well-defined for any $ \varepsilon>\tilde{\varepsilon}$ and its asymptotic support has the same $\varepsilon$-control as the one established by the \textsc{abc} procedure on the discrepancy among the empirical distributions. These results allow to establish the almost sure convergence of the \textsc{abc} posterior stated in Corollary~\ref{conv_rad}.

 \begin{corollary} \label{conv_rad}
 Under the conditions  in Theorem~\ref{conv}, if also $\mathfrak{R}_n(\mathfrak{F}) \rightarrow 0$ as $n \rightarrow \infty$ (i.e., Assumption (\ref{C4}) is satisfied), then, for any fixed $ \varepsilon>\tilde{\varepsilon}$, it holds that
  \begin{eqnarray}
   \pi_n^{(\varepsilon)}(\theta)  \rightarrow \pi( \theta \mid \mathcal{D}_{\mathfrak{F}}(\mu_\theta,\mu^*) \leq \varepsilon ) \propto  \pi( \theta ) \mathds{1}\left\{\mathcal{D}_{\mathfrak{F}}(\mu_\theta,\mu^*) \leq \varepsilon \right\},
   \label{conv_ABC}
  \end{eqnarray}
almost surely with respect to $y_{1:n} \stackrel{\mbox{\normalfont \scriptsize i.i.d.}}{\sim} \mu^{*}$, as $n \rightarrow \infty$. 
\end{corollary}
 
Corollary~\ref{conv_rad} clarifies that to obtain within the whole \textsc{ips} class a convergence result in line  with those stated in the current theory it is not necessary to assume that $\mathcal{D}_\mathfrak{F} (\hat{\mu}_{z_{1:n}},\hat{\mu}_{y_{1:n}})$ converges almost surely to $\mathcal{D}_\mathfrak{F} (\mu_\theta, \mu^*)$. In fact, it is sufficient to verify that $\mathfrak{R}_n(\mathfrak{F}) \rightarrow 0$ as $n \rightarrow \infty$. Crucially, this condition is imposed directly on the selected discrepancy within the \textsc{ips} class, rather than on the model or the underlying data generating process and, therefore, can be constructively verified for all the discrepancies in Examples~\ref{exm_wass}--\ref{exm_sum}. For instance, while \citet{jiang2018approximate} is unable to provide conclusive guidelines on whether \eqref{conv_ABC} holds under Wasserstein distance and \textsc{mmd}, our Corollary~\ref{conv_rad} can be directly applied to prove convergence for both these divergences, leveraging the results and discussions in Section~\ref{sec_ass}.  
 
Corollary~\ref{conv_rad} accounts also for misspecified models. In fact, the limiting pseudo-posterior in \eqref{conv_ABC} is well-defined only when  $ \varepsilon> \tilde{\varepsilon}$. For example, when $\mu^*$ is not within $ \{ \mu_\theta: \theta \in \Theta  \subseteq \mathbb{R}^{p} \} $, then $ \mathcal{D}_{\mathfrak{F}}(\mu_\theta,\mu^* ) \geq \varepsilon^*$. Hence, for any $\varepsilon \leq  \varepsilon^*$, we have that  $\mathds{1}\{\mathcal{D}_{\mathfrak{F}}(\mu_\theta,\mu^*) \leq \varepsilon \}=0$ for any $\theta \in \Theta$ and the limiting pseudo-posterior is not well-defined.
 
Motivated by the convergence result in Corollary~\ref{conv_rad}, the theory in Section~\ref{sec_abc_con} provides an in-depth study of the concentration properties for the \textsc{abc} posterior when the tolerance threshold progressively shrinks. As is clear from \eqref{conv_ABC}, employing a vanishing threshold might guarantee that, in the limit, the \textsc{abc} posterior concentrates all its mass around $\mu^*$, in correctly specified settings, or around the distribution $\mu_{\theta^*}$ closest to $\mu^*$, among those in $ \{ \mu_\theta: \theta \in \Theta  \subseteq \mathbb{R}^{p} \} $ for misspecified regimes. These results and the associated rates of concentration are provided in Section~\ref{sec_abc_con} by leveraging, again, Rademacher complexity theory.
 

\subsection{Concentration as $\varepsilon_n \rightarrow \varepsilon^*$ and $n \rightarrow \infty$}
\label{sec_abc_con}
Theorem~\ref{thm_rademacher}  states our main concentration result. As in related \textsc{abc} theory (see e.g., Proposition 3 in \citet{bernton2019approximate} and \citet{nguyen2020approximate}), we also leverage the triangle inequality  \eqref{ineq}. However, we crucially avoid pre-assuming convergence of $\mathcal{D}_\mathfrak{F}(\hat{\mu}_{y_{1:n}},\mu^*)$, and we do not rely on concentration inequalities for  $\mathcal{D}_\mathfrak{F} (\hat{\mu}_{z_{1:n}}{,}\mu_\theta)$  regulated by non-explicit sequences of control functions $c_n(\theta)f_n(\bar{\varepsilon}_n)$  \citep{bernton2019approximate}  and $c_n(\theta)s_n(\bar{\varepsilon}_n)$  \citep{nguyen2020approximate}. Rather, we leverage Lemma~\ref{lemma_rademacher} under the proposed Rademacher complexity framework  to obtain  more direct and informative results, that are also guaranteed to hold uniformly over $\mathcal{P}(\mathcal{Y})$. Indeed, Lemma~\ref{lemma_rademacher} ensures that both $\mathcal{D}_\mathfrak{F}(\hat{\mu}_{y_{1:n}},\mu^*)$ and $ \mathcal{D}_\mathfrak{F} (\hat{\mu}_{z_{1:n}},\mu_\theta)$ exceed $2\mathfrak{R}_n(\mathfrak{F})$ with vanishing probability. Therefore, when $\mathfrak{R}_n(\mathfrak{F}) \rightarrow 0$, we have that $\mathcal{D}_\mathfrak{F} (\hat{\mu}_{z_{1:n}},\hat{\mu}_{y_{1:n}}) \approx \mathcal{D}_\mathfrak{F} (\mu_\theta, \mu^*)$, by combining  \eqref{ineq} with 
\begin{eqnarray*}
\mathcal{D}_\mathfrak{F} (\mu_\theta, \mu^*) \geq - \mathcal{D}_\mathfrak{F} (\hat{\mu}_{z_{1:n}},\mu_\theta) +\mathcal{D}_\mathfrak{F} (\hat{\mu}_{z_{1:n}},\hat{\mu}_{y_{1:n}})-  \mathcal{D}_\mathfrak{F}(\hat{\mu}_{y_{1:n}},\mu^*). 
\end{eqnarray*}
This means that, if $\mathcal{D}_\mathfrak{F} (\hat{\mu}_{z_{1:n}},\hat{\mu}_{y_{1:n}})$ is small, then  $\mathcal{D}_{\mathfrak{F}}(\mu_\theta,\mu^*)$ is also small with a high probability.  This clarifies the importance of Assumption (\ref{C4}), which is further supported~by~the~fact that, if $\mathfrak{R}_n(\mathfrak{F})$ does not shrink to zero, by the second inequality in Lemma~\ref{lemma_rademacher}, the two empirical measures $\hat{\mu}_{y_{1:n}} $ and $ \hat{\mu}_{z_{1:n}}$ do not converge to $ \mu^* $ and $ \mu_{\theta} $, respectively, and hence there is no guarantee that a small $\mathcal{D}_\mathfrak{F} (\hat{\mu}_{z_{1:n}},\hat{\mu}_{y_{1:n}})$ would imply a vanishing $\mathcal{D}_{\mathfrak{F}}(\mu_\theta,\mu^*)$.

\begin{theorem}
	\label{thm_rademacher}
 Let $\bar{\varepsilon}_n \to 0$ when $ n \to\infty $, with {$n\bar{\varepsilon}^2_n \to \infty $} and {$\bar{\varepsilon}_n / \mathfrak{R}_{n}(\mathfrak{F}) \to \infty $}. Then, if $\mathcal{D}_\mathfrak{F}$ is from the \textsc{ips} class in Definition~\ref{def_IPS} and Assumptions (\ref{C1})--(\ref{C4}) hold,  the \textsc{abc} posterior with threshold $\varepsilon_n = \varepsilon^* + \bar{\varepsilon}_n $ satisfies 
	 \begin{eqnarray*}
	 {\pi_n^{(\varepsilon^*+\bar{\varepsilon}_n)}\Bigl(\Bigl\{\theta:  \mathcal{D}_{\mathfrak{F}}(\mu_\theta,\mu^*)} > {\varepsilon^* + \frac{4\bar{\varepsilon}_n}{3}+2\mathfrak{R}_{n}(\mathfrak{F}) +\Bigl(\frac{2b^2}{n}\log \frac{n}{\bar{\varepsilon}_n^L}\Bigr)^{1/2}}\Bigr\}\Bigr) \leq \frac{2 \cdot 3^L}{c_\pi n},
	 \end{eqnarray*}
	with {$\mathbb{P}_{y_{1:n}}$}--probability going to $1$ as $n \rightarrow \infty.$
\end{theorem}

The proof of Theorem~\ref{thm_rademacher} can be found in Appendix D of the Supplementary Material and follows similar arguments considered to establish the concentration results in \citet{bernton2019approximate} and \citet{nguyen2020approximate}, which, in turn, extend those in \citet{frazier2018asymptotic}. However, as mentioned above, those proofs are specific to a single discrepancy, pre-assume the convergence of $\mathcal{D}_\mathfrak{F}(\hat{\mu}_{y_{1:n}},\mu^*)$, and rely on concentration inequalities for $\mathcal{D}_\mathfrak{F} (\hat{\mu}_{z_{1:n}},\mu_\theta)$  that depend on non-explicit sequences of control functions. Theorem~\ref{thm_rademacher} overcomes these issues and proves a unified theory based on the single concentration inequality in Lemma~\ref{lemma_rademacher}. This yields technical differences in the proof and, more importantly, it introduces a novel and broadly-impactful perspective for the analysis of  concentration properties of discrepancy-based \textsc{abc} posteriors. 

Note that in Theorem~\ref{thm_rademacher} the  constant $b$  can be typically set equal to $1$ either by definition or upon normalization of the class of $b$-uniformly bounded functions. Moreover, as clarified in the Supplementary Material, Theorem~\ref{thm_rademacher} also holds when replacing $n/\bar{\varepsilon}_n^L$ and $c_\pi n$ with $M_n/\bar{\varepsilon}_n^L$ and $c_\pi M_n$, respectively, for any sequence $M_n>1$. Nonetheless, to ensure concentration, it suffices to let $M_n = n$. In such a case, the quantities $4\bar{\varepsilon}_n/3$, $2\mathfrak{R}_{n}(\mathfrak{F})$, $[(2b^2/n)\log(n/\bar{\varepsilon}_n^L)]^{1/2} $ and $2 \cdot 3^L/c_\pi n$ converge to $0$ as $n \rightarrow \infty$, under the settings of Theorem~\ref{thm_rademacher}. This implies that the \textsc{abc} posterior asymptotically concentrates around those $\theta$ values yielding a  $ \mu_\theta $ within discrepancy $ \varepsilon^* $ from  $\mu^*$. Appendix B in the Supplementary Material translates these concentration results from  the space of distributions to the space of parameters.

In contrast to available theory, the concentration result in Theorem~\ref{thm_rademacher} holds uniformly over $\mathcal{P}(\mathcal{Y})$, and replaces the currently-employed non-explicit control functions with a known and well-studied quantity, i.e., the Rademacher complexity. Notice that, although Theorem~\ref{thm_rademacher} provides a unified statement for all the \textsc{ips} discrepancies, the bound we derive depends on $\mathfrak{R}_{n}(\mathfrak{F})$, which is specific to each discrepancy $\mathcal{D}_\mathfrak{F}$ and plays a fundamental role in controlling the rate of concentration of the \textsc{abc} posterior. In particular,  to make the bound as tight as possible, we must choose an $\bar{\varepsilon}_n$ such that $4\bar{\varepsilon}_n/3$ and {$[(2b^2/n)\log(n/\bar{\varepsilon}_n^L)]^{1/2}$} are of the same order. By neglecting all the terms in $\log \log n$, such a choice leads to setting {$\bar{\varepsilon}_n$} of the order {$[\log(n)/n]^{1/2}$}. In this case, the constraint $\bar{\varepsilon}_n / \mathfrak{R}_{n}(\mathfrak{F}) \to \infty $ may not be satisfied, thereby requiring a larger $\bar{\varepsilon}_n$, such as $\mathfrak{R}_{n}(\mathfrak{F}) \log \log(n)$. Summarizing, when $\bar{\varepsilon}_n = \max\{[\log(n)/n]^{1/2}, \mathfrak{R}_{n}(\mathfrak{F}) \log \log(n) \} $ it follows, under the conditions of  Theorem~\ref{thm_rademacher}, that
\begin{eqnarray*}
{\pi_n^{(\varepsilon^*+\bar{\varepsilon}_n)}(\{\theta:  \mathcal{D}_{\mathfrak{F}}(\mu_\theta,\mu^*)} > \varepsilon^* + \mathcal{O}( \bar{\varepsilon}_n )\}) \leq 2 \cdot 3^L/(c_\pi n).
\end{eqnarray*}

Notice that the conditions {$n\bar{\varepsilon}^2_n \to \infty $} and {$\bar{\varepsilon}_n / \mathfrak{R}_{n}(\mathfrak{F}) \to \infty $} do not allow to set $\bar{\varepsilon}_n<1/\sqrt{n}$. Although this regime is of interest, we are not aware of explicit results in the discrepancy-based \textsc{abc} literature for such a setting. In fact, available studies rely on similar restrictions for some sequence of functions $f_n(\bar{\varepsilon}_n)$ which is not made explicit except for specific examples that still point toward setting $\bar{\varepsilon}_n=[\log (n)/n]^{1/2}>1/\sqrt{n}$  \citep[e.g.,][supplementary materials]{bernton2019approximate}. Faster rates for the \textsc{abc} threshold have been considered by \citet{lifearn2018} in the context of summary-based \textsc{abc}, but with a substantially different theoretical focus and asymptotic regime relative to the one considered here. Recalling, e.g., Theorem~1 in  \citet{frazier2018asymptotic}, within our theory setting the rate $1/\sqrt{n}$ (up to log terms) cannot be improved for summary-based \textsc{abc}; see also Section~\ref{sec_mmd} for additional discussion. It shall be also emphasized that $\bar{\varepsilon}_n<1/\sqrt{n}$ would not yield a faster concentration rate in Theorem~\ref{thm_rademacher} because of the terms $2\mathfrak{R}_{n}(\mathfrak{F})$ and $[(2b^2/n)\log(n/\bar{\varepsilon}_n^L)]^{1/2} $ in the bound.

Theorem~\ref{thm_rademacher} holds under both well-specified and misspecified models. In the former case, $ \mu^*= \mu_{\theta^*} $ for some $ \theta^* \in \Theta $. Therefore, $ \varepsilon^* =0 $ and the \textsc{abc} posterior concentrates around those $\theta$ yielding $ \mu_\theta=\mu^* $. Conversely, when the model is misspecified, the \textsc{abc} posterior concentrates on those $\theta$ yielding a $ \mu_\theta $ within discrepancy $ \varepsilon^* $ from $\mu^*$. Since $ \varepsilon^*  = \inf\nolimits_{\theta \in \Theta} \mathcal{D}_{\mathfrak{F}}(\mu_\theta, \mu^*)$,  the \textsc{abc} posterior will concentrate on those $\theta$ such that $\mathcal{D}_{\mathfrak{F}}(\mu_\theta,\mu^*)$ is arbitrarily close to $\varepsilon^*$. Example~\ref{exm_huber} provides an explicit bound for $\varepsilon^*$ in a simple yet noteworthy class of misspecified models.

\begin{example}
{\em (Huber contamination model)}. \label{exm_huber}
In this model, with probability $ 1-\alpha_n$, the data are  from a distribution $\mu_{\theta^*}$ belonging to the assumed model $ \{ \mu_\theta: \theta \in \Theta  \subseteq \mathbb{R}^{p} \} $, while with probability $\alpha_n$ arise from a contaminating distribution $\mu_{C}$. Therefore, the data generating process is $ \mu^* = (1-\alpha_n) \mu_{\theta^*} + \alpha_n \mu_C, $ with $\alpha_n \in [0,1)$ controlling the amount of contamination. In such a context, Definition~\ref{def_IPS} and Assumption (\ref{C3}) imply  
\begin{eqnarray*}
\varepsilon^*  \leq  \mathcal{D}_{\mathfrak{F}}(\mu_{\theta^*}, \mu^*)= \mathcal{D}_{\mathfrak{F}}(\mu_{\theta^*},(1-\alpha_n) \mu_{\theta^*} + \alpha_n \mu_C)= \alpha_n \mathcal{D}_{\mathfrak{F}}(\mu_{\theta^*}, \mu_C) \leq 2b \alpha_n.
\end{eqnarray*}
 By plugging this bound into Theorem~\ref{thm_rademacher}, we can obtain the same statement with $\varepsilon^*$ replaced by $2b \alpha_n$, where $ \alpha_n \in [0,1)$ is the amount of contamination. This means that the \textsc{abc} posterior asymptotically contracts in a neighborhood of the contaminated model $ \mu^*$ of radius at most $2b \alpha_n$. The previous bound on $\varepsilon^*$ combined with a triangle inequality also implies 
 \begin{eqnarray*}
 \mathcal{D}_{\mathfrak{F}}(\mu_\theta, \mu^*) \geq  \mathcal{D}_{\mathfrak{F}}(\mu_\theta, \mu_{\theta^*})-\mathcal{D}_{\mathfrak{F}}(\mu_{\theta^*},\mu^*) \geq  \mathcal{D}_{\mathfrak{F}}(\mu_\theta, \mu_{\theta^*})-2b \alpha_n.
 \end{eqnarray*}
  Therefore, replacing $\mathcal{D}_{\mathfrak{F}}(\mu_\theta, \mu^*)$ with $\mathcal{D}_{\mathfrak{F}}(\mu_\theta, \mu_{\theta^*})$ guarantees concentration also around the uncontaminated model $\mu_{\theta^*}$, with $\varepsilon^*$ replaced by $2b \alpha_n+2b \alpha_n=4b \alpha_n$, thus ensuring robustness to Huber contamination.
\end{example}


\subsection{Validity of the assumptions}
\label{sec_ass}
The main theorems in Section~\ref{sec_abc_lim}--\ref{sec_abc_con} (i.e., Theorems~\ref{conv} and \ref{thm_rademacher}) leverage Assumptions (\ref{C1})--(\ref{C4}). As anticipated in Section~\ref{sec_abc_posterior}, Assumption (\ref{C1}) is useful to formally clarify the range of applicability along with the possible limitations of the current existence theory even in simple i.i.d.\ settings, and it will be relaxed in Appendix C of the Supplementary Material to study uniform convergence and concentration of discrepancy-based \textsc{abc} posteriors beyond the i.i.d.\ context. Assumption  (\ref{C2}) is not specific to our framework. Rather, it defines a standard minimal requirement routinely employed in Bayesian asymptotics and \textsc{abc} theory \citep[e.g.,][]{bernton2019approximate, nguyen2020approximate,frazier2020robust}. Conversely, Assumptions  (\ref{C3})--(\ref{C4}) provide sufficient conditions on the \textsc{ips} discrepancy $\mathcal{D}_{\mathfrak{F}}$ to obtain guarantees of uniform convergence and concentration under the proposed Rademacher complexity perspective. These two conditions essentially replace the pre-assumed convergence for $\mathcal{D}_\mathfrak{F}(\hat{\mu}_{y_{1:n}},\mu^*)$ and the non-explicit bounds within the concentration inequalities for $\mathcal{D}_\mathfrak{F} (\hat{\mu}_{z_{1:n}},\mu_\theta)$ leveraged by the current discrepancy-specific theory. While these latter assumptions may implicitly require regularity conditions on $ \{ \mu_\theta: \theta \in \Theta  \subseteq \mathbb{R}^{p} \} $ and on the unknown data generating process $\mu^*$, Assumptions (\ref{C3})--(\ref{C4}) are made directly on the user-selected discrepancy $\mathcal{D}_{\mathfrak{F}}$. Leveraging the available bounds on the Rademacher complexity, these two conditions  crucially allow to state results that hold uniformly over $\mathcal{P}(\mathcal{Y})$. 

Notice that, to derive these uniform convergence and concentration properties,  an arguably minimal sufficient requirement is that $\mathfrak{F}$ is a uniform Glivenko--Cantelli class \citep[see e.g.,][Proposition 10]{dud1991}, i.e., for $x_{1:n}$ i.i.d. from $\mu \in \mathcal{P}(\mathcal{Y})$,
\begin{eqnarray}
\label{GCC}
\sup_{\mu \in \mathcal{P}(\mathcal{Y})}\sup_{f\in\mathfrak{F}} \left| \frac{1}{n}\sum_{i=1}^n f(x_i) - \int_{\mathcal{Y}} f d \mu \right| =\sup_{\mu \in \mathcal{P}(\mathcal{Y})}\mathcal{D}_\mathfrak{F} (\hat{\mu}_{x_{1:n}},\mu) \to 0,
\end{eqnarray}
 in $ \mathbb{P}_{x_{1:n}}$--probability as $n \to \infty$. In fact, as already discussed within Section~\ref{sec_abc}, the lack of uniform convergence guarantees for both $\mathcal{D}_\mathfrak{F}(\hat{\mu}_{y_{1:n}},\mu^*)$ and $\mathcal{D}_\mathfrak{F} (\hat{\mu}_{z_{1:n}},\mu_\theta)$ would fail to ensure that the control established by the \textsc{abc} threshold on $\mathcal{D}_\mathfrak{F}(\hat{\mu}_{z_{1:n}},\hat{\mu}_{y_{1:n}})$   applies, asymptotically, to  $\mathcal{D}_\mathfrak{F}(\mu_{\theta},\mu^*)$, uniformly over $\mathcal{P}(\mathcal{Y})$. Interestingly, although the uniform Glivenko--Cantelli property in \eqref{GCC} might seem more general and weaker than Assumptions (\ref{C3})--(\ref{C4}), by the upper and lower bounds in Lemma~\ref{lemma_rademacher} along with the subsequent discussion, it follows immediately that \eqref{GCC} is exactly equivalent to Assumption (\ref{C4}), under (\ref{C3}); see also Chapter 4 in \citet{wainwright2019high}. Regarding (\ref{C3}), notice that, as discussed in \citet{dud1991}, when $\mathfrak{F}$ is a uniform Glivenko--Cantelli class, then $\bar{\mathfrak{F}} :=\{\bar{f}(\cdot) := f(\cdot)-\inf_{x}f(x), f\in\mathfrak{F} \}$ is uniformly bounded and $\mathcal{D}_{\bar{\mathfrak{F}}} = \mathcal{D}_{\mathfrak{F}}$. Indeed, for any $f\in\mathfrak{F}$, $ \int \bar{f} d \mu_1  - \int \bar{f} d \mu_2  = \int [f -\inf_x f(x)]  d \mu_1 - \int [f(x) -\inf_x f(x)]  d \mu_2 = \int f  d \mu_1 - \int f  d \mu_2  $.
 Thus, when the uniform Glivenko–Cantelli property in  \eqref{GCC}  holds, it is always possible to re-define $\mathfrak{F}$, without affecting $\mathcal{D}_\mathfrak{F}$, in order to ensure that  (\ref{C3}) is verified, and hence also (\ref{C4}) as a consequence of the above discussion.

The above connection clarifies that Assumptions (\ref{C3})--(\ref{C4}) are arguably at the core  of the uniform convergence and concentration properties of discrepancy-based \textsc{abc} posteriors. Moreover, although \eqref{GCC} is inherently related to  (\ref{C3})--(\ref{C4}), such a uniform Glivenko--Cantelli property only states a convergence in probability result which can be crucially refined through the notion of Rademacher complexity under the more precise concentration inequalities  in Lemma~\ref{lemma_rademacher}. Recalling the theoretical results in Sections~\ref{sec_abc_lim} and \ref{sec_abc_con}, this allows not only to state convergence and concentration of specific \textsc{abc} posteriors, but also to clarify the factors governing these limiting properties and possibly derive the associated rates.

As clarified in Examples \ref{exm_wass1}--\ref{exm_sum1}, Assumptions (\ref{C3})--(\ref{C4}) can be generally verified for the key \textsc{ips} discrepancies presented in Examples~\ref{exm_wass}--\ref{exm_sum} leveraging known upper bounds on the Rademacher complexity, along with connections between this measure and other well-studied quantities in statistical learning theory. See, in particular, Chapter 4.3 of \citet{wainwright2019high} for an overview of several useful techniques for upper-bounding the Rademacher complexity via, e.g., the notion of polynomial discrimination and  \textsc{vc} dimension. The validity of (\ref{C3})--(\ref{C4}) for two other instances of  the \textsc{ips} class (i.e., the total variation distance and the Kolmogorov--Smirnov distance) is discussed in  Appendix A within the Supplementary Material. Notice that, albeit interesting, these two discrepancies are less common in \textsc{abc} implementations relative to Wasserstein distance, \textsc{mmd} and summary-based distances.

\begin{example}
{\em (Wasserstein-1 distance)}. \label{exm_wass1}
When $ \mathcal{Y} $ is a bounded subset of $\mathbb{R}^d$, Assumptions (\ref{C3})--(\ref{C4}) hold without further constraints under the Wasserstein-1 distance. In particular, (\ref{C3}) follows immediately from the definition  $ \mathfrak{F} = \{f : ||f||_L \leq 1 \} $, together with the fact that the diameter of $ \mathcal{Y}$ is finite in this case \citep[e.g.,][Remark 1.15]{villani2021topics}. Assumption (\ref{C4}) is instead a direct consequence of the  bounds in \citet{sriperumbudur2009integral}. Although it would be desirable to remove such a constraint on $ \mathcal{Y} $, it shall be emphasized that this condition is  ubiquitous in state-of-the-art  concentration results of empirical measures, under the Wasserstein distance, that are guaranteed to hold uniformly over $ \mathcal{P}(\mathcal{Y})$ \citep[e.g.,][]{tal1994,sriperumbudur2009integral,ram2017,weed2019sharp}. One possibility to preserve (\ref{C3})--(\ref{C4}) under the Wasserstein-1 distance beyond bounded $ \mathcal{Y} $ is to consider a variable transformation via a monotone function $g(\cdot)$ (e.g., logistic transform) mapping from $ \mathcal{Y}=\mathbb{R}^d$ to a bounded subset of $\mathbb{R}^d$. In the original unbounded $\mathcal{Y}$, this transformation induces a Wasserstein-1 distance based on a bounded $\bar{\rho}(x,x')=\rho(g(x),g(x'))$. As such,  when $ \mathcal{Y}$ is not a bounded subset of $\mathbb{R}^d$, defining $\mathfrak{F}$ as the class of Lipschitz functions with respect to $\bar{\rho}(x,x')=\rho(g(x),g(x'))$ satisfies  (\ref{C3}) and (\ref{C4}). However, this requires care and further research on how  $g(\cdot)$ affects the properties of the discrepancy. Alternatively, as shown in Proposition~\ref{corr_Wasserstein_unb} within Section~\ref{sec_wass_unb}, concentration results for Wasserstein-1 distance in unbounded $ \mathcal{Y} \subseteq \mathbb{R}^d$ can be still derived, but at the expense of some restrictions on $\mu \in \mathcal{P}(\mathcal{Y})$.
\end{example}

\begin{example} {\em (\textsc{mmd}).} \label{exm_mmd1}
	The properties of  \textsc{mmd}  inherently depend on the selected kernel $ k(\cdot,\cdot) $.  This is  evident from the  inequalities {$\mathfrak{R}_{\mu,n}(\mathfrak{F}) \leq [\mathbb{E}_{x} k(x,x)/n]^{1/2}$}, with $x \sim \mu$  \citep[see e.g., Lemma 22 in][]{bartlett2002rademacher}, and $|f(x)| \leq[ k(x,x)]^{1/2} ||f||_\mathcal{H} $ for every $x \in \mathcal{Y}$ \citep[see e.g., equation 16 in][]{hofmann2008kernel}. By these two inequalities, all bounded kernels, including, e.g., the commonly-implemented Gaussian {$\exp(-\|x-x'\|^2/\sigma^2)$} and the Laplace  {$\exp(-\|x-x'\|/\sigma)$} ones, ensure that Assumptions (\ref{C3}) and (\ref{C4}) are always met without requiring additional regularity conditions on $\mu  \in \mathcal{P}(\mathcal{Y})$, nor constraints on $\mathcal{Y}$. Instead, when $ k(\cdot,\cdot) $ is unbounded, the inequality $\mathfrak{R}_{\mu,n}(\mathfrak{F}) \leq [\mathbb{E}_{x} k(x,x)/n]^{1/2}$ is only informative for those $\mu$ such that $\mathbb{E}_{x} k(x,x) < \infty$, with $ x \sim \mu $, whereas the bound on $ |f(x)| $ does not guarantee that $\mathfrak{F}$ is a uniformly bounded class in general, unless additional conditions are made. Due to the relevance of these \textsc{mmd} instances and the direct connections with the classical summary-based \textsc{abc} implementations leveraging unbounded summaries, Proposition~\ref{corr_MMD_unb} in Section~\ref{sec_mmd} derives specific theory proving that concentration results similar to those in Theorem~\ref{thm_rademacher} can be derived, under different assumptions, including conditions on $\mu \in \mathcal{P}(\mathcal{Y})$, also for unbounded kernels.
\end{example}

\begin{example}\label{exm_sum1} {\em (Summary-based distance).}
As discussed in Example~\ref{exm_sum}, classical \textsc{abc} implementations relying on a finite set of summaries $f_1, \ldots, f_K$ with $K < \infty$, can be seen as a special case of \textsc{mmd} by letting $f(x) = [f_1(x), \ldots, f_K(x)]$ and $k(x,x') = \left<f(x),f(x')\right>$. Leveraging this bridge and the results for  \textsc{mmd} discussed in Example~\ref{exm_mmd1}, it is clear that, if $\sup_{x\in\mathcal{Y}} \left<f(x),f(x)\right>$ is finite --- i.e., the induced kernel is bounded --- then  (\ref{C3}) and (\ref{C4})  are satisfied without requiring  regularity conditions for $\mu$ or constrains on  $\mathcal{Y}$. While this result clarifies that \textsc{abc} with bounded summaries achieves uniform convergence and concentration, classical  \textsc{abc}  implementations often employ unbounded summaries, such as moments, i.e., $f(x) = [x, x^2 \ldots, x^K]$. In this case  (\ref{C3})  is not satisfied. In fact, recalling again   Example~\ref{exm_sum}, such a setting is a special case of \textsc{mmd} with  unbounded kernel and, hence, lacks guarantees that  (\ref{C3})--(\ref{C4})  hold, unless further conditions are imposed, e.g., on $\mathcal{Y}$. Nonetheless, this connection also clarifies that the concentration theory we derive in Proposition~\ref{corr_MMD_unb} for \textsc{mmd} with unbounded kernels directly applies to classical \textsc{abc} with unbounded summaries.
\end{example} 

Examples \ref{exm_wass1}--\ref{exm_sum1} show that (\ref{C3})--(\ref{C4}) can be realistically verified for the key instances of the \textsc{ips} class presented in Section~\ref{sec_abc}, and generally hold under either no additional conditions or for suitable constraints on $\mathcal{Y}$ which can be directly checked simply on the basis of the support of the data analyzed. From a practical perspective, this is an important gain relative to the need of verifying more sophisticated regularity conditions on the assumed model and on the unknown data generating process. Notice that the boundedness condition on $ \mathcal{Y} $  in Example~\ref{exm_wass1} interestingly  relates to Assumptions 1 and 2 of \citet{bernton2019approximate} which have been verified when $ \mathcal{Y} $ is a bounded subset of $\mathbb{R}^d$ by, e.g., \cite{weed2019sharp}. In this context, our Rademacher complexity perspective further refines the important results in \citet{bernton2019approximate} by clarifying that it is possible to derive convergence and concentration results for Wasserstein-\textsc{abc}  that are regulated by a known complexity measure and hold uniformly over the space of probability measures defined on a bounded $ \mathcal{Y} $. 

Notice that an alternative possibility to verify Assumptions 1--2 in \citet{bernton2019approximate}  is to leverage the results in \citet{fournier2015rate} who replace the boundedness condition in  \cite{weed2019sharp} with assumptions on the existence of exponential moments; see also  the supplementary materials in \citet{bernton2019approximate}. A similar direction within our Rademacher complexity framework would be to assume that the class of functions $\mathfrak{F}$ defining the Wasserstein distance admits a uniform Glivenko--Cantelli property over a subset $\tilde{\mathcal{P}}(\mathcal{Y})$ of $\mathcal{P}(\mathcal{Y})$ that comprises probability measures on $\mathcal{Y}$ meeting some suitable regularity conditions. Recalling the previous discussion on the connection among \eqref{GCC} and our assumptions, this would imply a relaxation of  (\ref{C3})--(\ref{C4}) allowing the theory in Sections~\ref{sec_abc_lim} and \ref{sec_abc_con} to still hold for statistical models and data generating processes belonging to  $\tilde{\mathcal{P}}(\mathcal{Y})$. Propositions \ref{corr_MMD_unb} and \ref{corr_Wasserstein_unb} explore results along these lines for  \textsc{mmd} with unbounded kernel in $\mathbb{R}^{d}$ (which also includes \textsc{abc} with unbounded summary statistics), and for the Wasserstein-1 distance in $\mathbb{R}^{d}$, respectively, which do not meet  (\ref{C3}) and (\ref{C4}) when $\mathcal{Y}$ is unbounded. These two propositions clarify that concentration results can still be stated under alternative, yet related, proofs based on suitable relaxations of (\ref{C3})--(\ref{C4}). Nonetheless, these relaxations require checking that $ \{ \mu_\theta: \theta \in \Theta  \subseteq \mathbb{R}^{p} \} $ and $\mu^*$ meet the conditions characterizing $\tilde{\mathcal{P}}(\mathcal{Y})$, which can be difficult since, again, $\mu^*$ is generally not known in practice. Conversely, when (\ref{C3})--(\ref{C4}) hold  (e.g., in \textsc{mmd} with bounded kernels and Wasserstein-1 distance in bounded subsets of $\mathbb{R}^d$) the convergence and concentration results in Section~\ref{sec_abc_lim}--\ref{sec_abc_con} are guaranteed without the need to worry about the peculiar properties of the assumed model and of the generally-unknown data generating process.


\section{Asymptotic properties of ABC with maximum mean discrepancy and Wasserstein-1 distance}
\label{sec_mmd_tot}

Sections~\ref{sec_mmd}--\ref{sec_wass_unb} specialize the theory derived within Section~\ref{sec_abc_posterior} to two remarkable distances in the \textsc{ips} class; namely \textsc{mmd} (which  includes distances among summaries~as~a~special case) and Wasserstein-1 distance, respectively. Recalling Examples~\ref{exm_wass1}--\ref{exm_sum1} these discrepancies are covered by the general results in Section~\ref{sec_abc_posterior}, under the assumption that the kernel~or the space $\mathcal{Y}$ are bounded. For completeness, we extend such concentration results also to situations where these two conditions are not met; see Propositions~\ref{corr_MMD_unb} and \ref{corr_Wasserstein_unb}.


\subsection{Asymptotic properties of ABC with maximum mean discrepancy}
\label{sec_mmd}
As discussed in Sections~\ref{sec_intro}--\ref{sec_abc_posterior}, \textsc{mmd} stands out as a prominent example of discrepancy within summary-free \textsc{abc}. Nonetheless, an in-depth and comprehensive study on the limiting properties of \textsc{mmd}-\textsc{abc}  is still lacking. In fact, no theory is available in the original article proposing \textsc{mmd}-\textsc{abc} methods \citep{park2016k2}, while convergence in the fixed $\varepsilon_n=\varepsilon$ and $n \rightarrow \infty$ regime is explored in  \citet{jiang2018approximate} but without conclusive results.  \citet{nguyen2020approximate} study convergence and concentration of \textsc{abc}  with the energy distance in both fixed and vanishing $\varepsilon_n$ settings. Recalling Example~\ref{exm_mmd}, the direct correspondence between \textsc{mmd} and the energy distance would allow to translate these results to the \textsc{mmd} framework. However, as highlighted by the same authors, the theory derived relies on  difficult-to-verify existence assumptions which yield bounds depending on control functions that are not made explicit.

Leveraging Theorems~\ref{conv}--\ref{thm_rademacher} and Corollary~\ref{conv_rad} along with the available upper bounds on the Rademacher complexity of \textsc{mmd} --- see Example~\ref{exm_mmd1} --- Corollaries \ref{corr_MMD_2}--\ref{corr_MMD} substantially refine and expand available knowledge on the limiting properties of \textsc{mmd}-\textsc{abc} with routinely-implemented bounded kernels. Crucially,  \textsc{mmd} with such kernels automatically satisfies (\ref{C3})--(\ref{C4}) without additional constraints on the model or on the data generating process. Notice that Corollaries \ref{corr_MMD_2}--\ref{corr_MMD} also hold for summary-based distances with bounded summaries as a direct consequence of the discussion in Examples~\ref{exm_sum} and \ref{exm_sum1}.

\begin{corollary}
	\label{corr_MMD_2}
	Consider the \textsc{mmd}  with a bounded kernel $ k(\cdot,\cdot) $ defined on $\mathbb{R}^d$, where $|k(x,x)| \leq 1$ for any $x \in \mathbb{R}^d$.  Then, under (\ref{C1}),  the acceptance probability $ p_n$ of the rejection-based \textsc{abc} routine employing discrepancy $\mathcal{D}_{\textsc{mmd}}$ and the threshold $\varepsilon$ satisfies
\begin{eqnarray*}
 p_n \rightarrow   \pi\{\theta: \mathcal{D}_{\textsc{mmd} }(\mu_\theta,\mu^* )\leq  \varepsilon\},
\end{eqnarray*}
almost surely with respect to $y_{1:n} \stackrel{\mbox{\normalfont \scriptsize i.i.d.}}{\sim} \mu^{*}$, as  $n \rightarrow \infty$, for any $\mu_\theta, \mu^* \in \mathcal{P}(\mathcal{Y}) $. Moreover, let $\tilde{\varepsilon}$ be defined as in Theorem~\ref{conv}. Then,  for any fixed $ \varepsilon>\tilde{\varepsilon}$, it holds that
  \begin{eqnarray*}
   \pi_n^{(\varepsilon)}(\theta)  \rightarrow \pi( \theta \mid \mathcal{D}_{\textsc{mmd}}(\mu_\theta,\mu^*) \leq \varepsilon ) \propto  \pi( \theta ) \mathds{1}\left\{\mathcal{D}_{\textsc{mmd}}(\mu_\theta,\mu^*) \leq \varepsilon \right\},
  \end{eqnarray*}
  almost surely with respect to $y_{1:n} \stackrel{\mbox{\normalfont \scriptsize i.i.d.}}{\sim} \mu^{*}$, as $n \rightarrow \infty$.
\end{corollary}

The above result applies, e.g., to  routinely-implemented Gaussian $k(x,x')=\exp(-\|x-x'\|^2/\sigma^2)$ and Laplace $k(x,x')=\exp(-\|x-x'\|/\sigma)$ kernels on $\mathbb{R}^d$, which are both bounded by $1$, thereby implying $ ||f||_{\infty}\leq 1 $ and  $\mathfrak{R}_{n}(\mathfrak{F}) \leq n^{-1/2}$. These results are also crucial to prove the concentration statement in Corollary~\ref{corr_MMD} below. 

\begin{corollary}
	\label{corr_MMD}
	Consider the \textsc{mmd}  with a bounded kernel $ k(\cdot,\cdot) $ defined on $\mathbb{R}^d$, where $|k(x,x)| \leq 1$ for any $x \in \mathbb{R}^d$.  Then, for any $\mu_\theta, \mu^* \in \mathcal{P}(\mathcal{Y}) $, under (\ref{C1})--(\ref{C2}) and the settings of  Theorem~\ref{thm_rademacher}, with {$\bar{\varepsilon}_n=[\log (n)/n]^{1/2}$}, we have that 
	\begin{eqnarray*}
	\pi_n^{(\varepsilon^*+\bar{\varepsilon}_n)}\Bigl(\Bigl\{\theta:  \mathcal{D}_{\textsc{mmd}}(\mu_\theta,\mu^*) > \varepsilon^* + \Bigl(\frac{10}{3}+(L+2)^{1/2}\Bigr)\cdot \Bigl(\frac{\log n}{n}\Bigr)^{1/2}\Bigl\} \Bigr) \leq \frac{2 \cdot 3^L}{c_\pi n},
	\end{eqnarray*}
	 with $\mathbb{P}_{y_{1:n}}$--probability going to $1$ as $n \rightarrow \infty$.
\end{corollary}

Notice that   $\mathfrak{R}_{n}(\mathfrak{F}) \leq n^{-1/2}$ implies $\mathfrak{R}_{n}(\mathfrak{F}) \log \log(n) \leq \log \log(n)/n^{1/2} \leq  [\log(n)/n]^{1/2}$ for any $n \geq 1$. Hence, as a consequence of the previous discussion, the concentration rate is essentially minimized by setting $\bar{\varepsilon}_n=[\log (n)/n]^{1/2}$. This result interestingly aligns with the optimal rate derived by \citet{frazier2018asymptotic} for summary-based \textsc{abc} under sub-Gaussian assumptions. The reason for such an agreement is immediately clear after noticing  that  \textsc{mmd}  with bounded kernels includes, as a special case, \textsc{abc} with bounded summaries.

Corollaries~\ref{corr_MMD_2}--\ref{corr_MMD} are effective examples of the potentials of  Theorems~\ref{conv}--\ref{thm_rademacher} and Corollary~\ref{conv_rad},  which can be readily specialized to any discrepancy within the \textsc{ips} class. For example, in the context of \textsc{mmd} with Gaussian and Laplace kernels, Corollary~\ref{corr_MMD} ensures informative posterior  concentration  without requiring assumptions on $ \{ \mu_\theta: \theta \in \Theta  \subseteq \mathbb{R}^{p} \} $ or $\mu^*$. Similar results can be obtained for all \textsc{ips} discrepancies as long as  (\ref{C3})--(\ref{C4}) are satisfied and  $\mathfrak{R}_{n}(\mathfrak{F})$ admits explicit upper bounds.  For instance,  if $\mathcal{Y}$ is bounded, this is possible for the Wasserstein-1 distance in Example~\ref{exm_wass1}, leveraging the  bounds for  $\mathfrak{R}_{n}(\mathfrak{F})$ in \citet{sriperumbudur2009integral}. 

While  (\ref{C3}) and (\ref{C4}) hold for \textsc{mmd} with a bounded kernel without additional assumptions, the currently-available bounds on the Rademacher complexity ensure that \textsc{mmd} with an unbounded kernel meets the above conditions only under specific models and data generating processes, even within the i.i.d.\ setting.  In this context, it is however possible to revisit the results for the Wasserstein case in Proposition 3 of  \citet{bernton2019approximate} under the new Rademacher complexity framework introduced in the present article. In particular, as shown in Proposition~\ref{corr_MMD_unb}, under \textsc{mmd} with an unbounded kernel, the existence Assumptions 1 and 2 in \citet{bernton2019approximate} can be directly related to constructive conditions on the kernel, inherently related to our Assumption  (\ref{C4}). This in turn yields informative concentration inequalities that are reminiscent of those in Theorem~\ref{thm_rademacher} and Corollary~\ref{corr_MMD}. Notice that these inequalities also hold for summary-based \textsc{abc} with routinely-used unbounded summaries (e.g., moments) as a direct consequence of the discussion in Example~\ref{exm_sum1}.

\begin{proposition}
	\label{corr_MMD_unb}
	Consider the \textsc{mmd} with unbounded kernel $ k(\cdot,\cdot) $ on $\mathbb{R}^d$. 
	Assume (\ref{C1})--(\ref{C2}) along with (A1) $\mathbb{E}_{y} \left[ k(y,y) \right]<\infty$, (A2) $\int_{\Theta} \mathbb{E}_{z} \left[ k(z,z) \right] \pi({\rm d}\theta) <\infty$, and (A3) there exist constants $\delta_0>0$ and $c_0>0$ such that $\mathbb{E}_{z} \left[ k(z,z) \right]<c_0$ for any $\theta$ satisfying $ (\mathbb{E}_{z,z'}\left[ k(z,z') \right] - 2 \mathbb{E}_{z,y}\left[ k(y,z) \right] + \mathbb{E}_{y,y'}\left[ k(y,y') \right])^{1/2} \leq \varepsilon^* + \delta_0$, where $z,z' \sim \mu_\theta$ and $y,y' \sim \mu^*$. Then, when $n \rightarrow \infty$, $\bar{\varepsilon}_n \rightarrow 0$ and $n\bar{\varepsilon}^2_n \rightarrow \infty$,  for some $C \in (0, \infty)$ and any $M_n \in (0, \infty)$, it holds that 
	\begin{eqnarray*}
	\pi_n^{(\varepsilon^*+\bar{\varepsilon}_n)}\Bigl(\Bigl\{\theta: \mathcal{D}_{\textsc{mmd}}(\mu_\theta,\mu^*) > \varepsilon^* + \frac{4\bar{\varepsilon}_n}{3} + \Bigr(\frac{M_n}{n \bar{\varepsilon}_n^L }\Bigr)^{1/2} \Bigr\}\Bigr) \leq \frac{C}{M_n},
	\end{eqnarray*}
with $\mathbb{P}_{y_{1:n}}$--probability going to $1$ as $n \rightarrow \infty$.
\end{proposition}

A popular example of unbounded kernel is provided by the polynomial one, which is defined as $k(x,x') = (1+a\left<x,x'\right>)^{q}$ for some integer $q\in\{2,3,\dots\}$, and constant $a>0$. Under such a kernel, it can be easily shown that if $\mathbb{E}_{y}[\|y \|^q] < \infty$, and $\theta\mapsto \mathbb{E}_{z}(\|z\|^q)  $ is $\pi$-integrable, then Assumptions $(A1)$--$(A3)$ are satisfied.  The latter conditions essentially require that the kernel has finite expectation under both $\mu^*$ and $\mu_\theta$ for suitable $\theta \in \bar{\Theta} \subseteq \Theta$, and is uniformly bounded for those $ \mu_\theta $ close to $ \mu^* $. Recalling Example~\ref{exm_mmd} and the bound $\mathfrak{R}_{\mu,n}(\mathfrak{F}) \leq [\mathbb{E}_{x \sim \mu} k(x,x)/n ]^{1/2}$, $(A1)$--$(A3)$ are inherently related to (\ref{C4}), which however additionally requires that these expectations are finite for any $\mu \in \mathcal{P}(\mathcal{Y})$. Notice that while the concentration result in Proposition~\ref{corr_MMD_unb} does not make explicit the dependence on the regularity of the different moments \citep[e.g.,][]{frazier2018asymptotic}, as clarified in the polynomial kernel example, such a dependence on $q$ is embedded within the definition of the kernel, which in turn enters the expression for $\mathcal{D}_{\textsc{mmd}}$. As such, translating the results in Proposition~\ref{corr_MMD_unb} from our overarching focus on the space of distributions to the one of parameters would make this dependence explicit. Such a mapping is challenging to obtain analytically in general settings. Hence, further refinements along these lines are left for future research.

Before moving to the Wasserstein-1 distance, notice that by Proposition~\ref{corr_MMD_unb} a sensible setting for  $\bar{\varepsilon}_n$ in the unbounded-kernel case would be {$\bar{\varepsilon}_n=(M_n/n)^{1/(2+L)}$}. This yields 
\begin{eqnarray*}
{\pi_n^{(\varepsilon^*+\bar{\varepsilon}_n)}}(\{\theta:  \mathcal{D}_{\textsc{mmd}}(\mu_\theta,\mu^*) > \varepsilon^* + (7/3)(M_n/n)^{1/(2+L)}\}) \leq C/M_n,
\end{eqnarray*}
 which is essentially the tighter possible order of magnitude for the bound. In the unbounded-kernel setting, $M_n=n$ would not be suitable, but any $M_n\rightarrow \infty$ slower than $n$ can work; e.g., $M_n=n^{1/2}$ yields {$  \pi_n^{(\varepsilon^*+\bar{\varepsilon}_n)}(\{\theta:  \mathcal{D}_{\textsc{mmd}}(\mu_\theta,\mu^*) > \varepsilon^* + (7/3)(1/n)^{1/(2L+4)} \}) \leq C/n^{1/2}.$} 

\subsection{Asymptotic properties of \textsc{abc} with Wasserstein-1 distance}
\label{sec_wass_unb}
As pointed out earlier, the  general results in Section~\ref{sec_abc_posterior} apply directly to the Wasserstein-1 distance ($\mathcal{D}_{\textsc{wass}}$) when the data space $\mathcal{Y}$ is bounded. As clarified in Proposition~\ref{corr_Wasserstein_unb}, if $\mathcal{Y}$ is unbounded, concentration results can still be derived for $\mathcal{D}_{\textsc{wass}}$, but at the expense of regularity conditions on  $ \{ \mu_\theta: \theta \in \Theta  \subseteq \mathbb{R}^{p} \} $ and $\mu^*$.  These results and the associated proof mirror those for the \textsc{mmd} with unbounded kernel in Proposition~\ref{corr_MMD_unb}, and rely on  exponential moment assumptions that allow to leverage the recent concentration bounds for $\mathcal{D}_{\textsc{wass}}$ in \citet{lei2020w}.
\begin{proposition}
	\label{corr_Wasserstein_unb}
	Consider the Wasserstein-1 distance $\mathcal{D}_{\textsc{wass}}$ (see Example~\ref{exm_wass}) on $\mathcal{Y}=\mathbb{R}^d$, with ground distance $\rho(x,x')=\|x-x'\|$.
	Besides (\ref{C1})--(\ref{C2}), assume there exists a small enough $c>0$ such that 
	(A1')~$ \ \mathbb{E}_y(\exp(c \|y\|)) <\infty $ and 
	(A2')~${\sup_{\mu_{\theta},\theta\in\Theta}} \mathbb{E}_z(\exp(c \|z\|)) <~\infty $,  where $z \sim \mu_\theta$ and $y \sim \mu^*$. Then, when $n \rightarrow \infty$ and $\bar{\varepsilon}_n \rightarrow 0$ such that  $n\bar{\varepsilon}^2_n \rightarrow \infty$ and also {$ n^{-1/ \max(d,3)} \ll \bar{\varepsilon}_n$}, for some $C \in (0, \infty)$ and any $M_n \in (0, \infty)$, 
	it holds that
	\begin{eqnarray*}
	\pi_n^{(\varepsilon^*+\bar{\varepsilon}_n)}\Bigl(\Bigl\{\theta: \mathcal{D}_{\textsc{wass}}(\mu_\theta,\mu^*) > \varepsilon^* + \frac{4\bar{\varepsilon}_n}{3} + \Bigl( \frac{1}{c'n}  \log\frac{M_n}{\bar{\varepsilon}_n^L} \Bigr)^{1/2}  + c_1 n^{-{1}/{ \max(d,3) }} \Bigr\}\Bigr) \leq \frac{C}{M_n},
	\end{eqnarray*}
with $\mathbb{P}_{y_{1:n}}$--probability going to $1$ as $n \rightarrow \infty$.
\end{proposition}

The results and proof of Proposition~\ref{corr_Wasserstein_unb} are inherently related to  those in Section 2.1 of the supplementary materials in \citet{bernton2019approximate}. However, rather than leveraging the concentrations bounds by \cite{fournier2015rate}, Proposition~\ref{corr_Wasserstein_unb} exploits the more recent ones in \citet{lei2020w} under the exponential moment assumptions in (A1')--(A2') that are generally satisfied for sub-exponential random variables, including sub-Gaussian ones. We refer to   \citet{lei2020w} for a more detailed discussion on the advantages of these  bounds compared to those in \cite{fournier2015rate}. Notice that, under Proposition~\ref{corr_Wasserstein_unb}, setting  $\bar{\varepsilon}_n = n^{-1/\max(d,3)} \sqrt{\log(n) }$ is an admissible choice to ensure concentration, but the induced rate would be worse than the one in Theorem~\ref{thm_rademacher}. Deriving optimal rates for the Wasserstein-1 distance in unbounded $\mathcal{Y}$ is not the scope of Proposition~\ref{corr_Wasserstein_unb}, whose aim is, instead,  to clarify that concentration beyond the settings considered in Section~\ref{sec_abc_posterior} can still be achieved, but at the cost of regularity conditions on  $ \{ \mu_\theta: \theta \in \Theta  \subseteq \mathbb{R}^{p} \} $ and $\mu^*$ that do not allow to state results uniformly over $\mathcal{P}(\mathcal{Y})$. As such, we leave questions on the optimality of the bounds for the Wasserstein distance in unbounded spaces $\mathcal{Y}$ as future research,  and refer to  \citet{bernton2019approximate} for more in-depth and specialized results on the properties of \textsc{abc} under such a distance.


\vspace{10pt}
\section{Illustrative simulation in i.i.d.\ settings}
\label{sec_empirical}
Let us illustrate the theory we derived in the previous sections through a simple, yet insightful, simulation in i.i.d.\ settings (see Appendix~C in the Supplementary Material for empirical results under non-i.i.d.\ regimes). Note that several empirical studies have already compared the performance of summary-free \textsc{abc} under different discrepancies and complex non-i.i.d.\ models. Recalling   \citet{drovandi2022comparison}, all these analyses clarify the practical feasibility of \textsc{abc} based on various discrepancies, including those within the \textsc{ips} class,  and in several i.i.d.\ and non-i.i.d.\ scenarios which require  \textsc{abc} procedures. This feasibility  is also supported by recent softwares \citep[e.g.,][]{Dutta2021}.

Rather than replicating the available  studies on benchmark examples, we complement current  empirical evidence on summary-free \textsc{abc} by focusing on a misspecified and contaminated scenario that clarifies the possible challenges in convergence and concentration encountered even in basic  i.i.d.\ settings. As clarified in Table~\ref{tab:caption} and Figure~\ref{figure:post}, this scenario also showcases a key consequence of the novel theoretical results  in Sections~\ref{sec_abc}--\ref{sec_mmd_tot}. Namely that effective \textsc{ips} discrepancies with guarantees of uniform convergence and concentration are a safe and sensibile choice in the absence of knowledge on the specific properties of $\mu^*$. 

 \begin{table}[b]
	\centering
		\caption{Concentration and runtimes in seconds (for a single discrepancy evaluation) of \textsc{abc} under \textsc{mmd} with Gaussian kernel, Wasserstein-1 distance, summary-based distance (mean) and \textsc{kl}  divergence for a misspecified Huber contamination model with a varying $\alpha\in\{0.05, 0.10, 0.15\}$. {$ \textsc{mse}= \hat{\mathbb{E}}_{\mu^*}[\hat{\mathbb{E}}_{\textsc{abc}}(\theta-\theta_0)^2]$}.}
	\begin{tabular}{lcccc}
		\hline
		& \textsc{mse} ($\alpha=0.05$)& \textsc{mse} ($\alpha=0.10$) & \textsc{mse} ($\alpha=0.15$)& time  \\
		\hline
		(\textsc{ips}) \textsc{mmd}     &{\bf 0.024} & {\bf 0.027} &   {\bf 0.031}& $<$ 0.01''  \\
		(\textsc{ips})  Wasserstein-1      &  0.027  &  0.067 &0.122 & $<$ 0.01''  \\
		(\textsc{ips})  summary (mean)    &  0.841  &  2.648 &2.835 & $<$ 0.01''\\
			\cline{1-5}
		(non-\textsc{ips}) \textsc{kl}      &    0.073 &  0.076  &0.077   & $<$ 0.01'' \\
		\hline
	\end{tabular}
	\label{tab:caption}
\end{table}

Recalling Example~\ref{exm_huber}, we consider, in particular, an uncontaminated bivariate Student's $t$ distribution  $\mu_{\theta_0}$ with $3$ degrees of freedom, mean vector $(1,1)$, and dispersion matrix having entries $\sigma_{11}=\sigma_{22}=1$ and $\sigma_{12}=\sigma_{21}=0.5.$ Such an uncontaminated data generating process is then perturbed with three different levels $\alpha_n=\alpha\in\{0.05, 0.10, 0.15\}$ of contamination from a Student's~$t$ distribution $\mu_C$ having the same parameters as $\mu_{\theta_0}$, except for the mean vector which is set to $(20,20)$. As such, the data $y_{1:n}$ from $\mu^* = (1-\alpha) \mu_{\theta_0} + \alpha \mu_C$ are obtained  by sampling $n=100$ draws from the bivariate Student's~$t$ $\mu_{\theta_0}$ and then replacing $ (100 \cdot \alpha) \%$ of these draws with  samples from the contaminating Student's~$t$ distribution $\mu_C$. For Bayesian inference, we focus on the parameter $\theta \in \mathbb{R}$ defining the unknown location vector $[1,1]^{\intercal}\theta$, and consider a misspecified bivariate Gaussian model $\mu_\theta$ with mean vector $[1,1]^{\intercal}\theta$ and known covariance matrix coinciding with that of the uncontaminated Student's~$t$  data generating process. Such a choice is interesting in providing a model that is slightly misspecified even when the data are not contaminated. Notice that, although this model does not necessarily require an \textsc{abc} approach to allow Bayesian inference, as discussed above, the issues outlined in Table~\ref{tab:caption} and Figure~\ref{figure:post} for certain discrepancies, even in such a basic example, provide a useful empirical insight that complements those in the extensive quantitive studies already available in the literature for more complex settings.

\begin{figure}[t]
	\begin{center} 
		\includegraphics[scale=0.73]{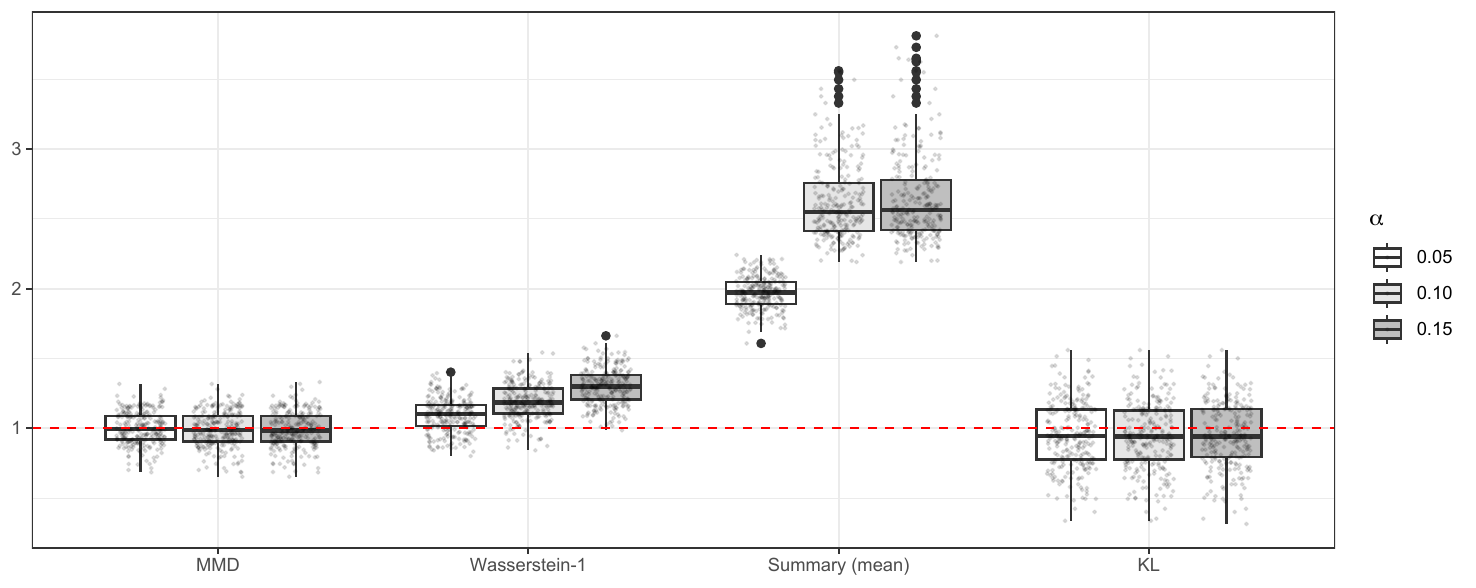}
		\caption{\footnotesize Graphical representation of the \textsc{abc} posterior for $\theta$ under \textsc{mmd} with Gaussian kernel, Wasserstein-1 distance, summary-based distance (mean) and \textsc{kl}  divergence for one simulated dataset from a misspecified Huber contamination model with a varying $\alpha\in\{0.05, 0.10, 0.15\}$; see also Example~\ref{exm_huber} for details. The red dashed line corresponds to the location parameter $\theta_0=1$ of the uncontaminated model. }
		\label{figure:post}
	\end{center}
\end{figure}

In performing \textsc{abc} under the above model and for the different discrepancies of interest, we employ rejection-based \textsc{abc}  with $m = n=100$ and a $\mbox{N}(0,1)$ prior for $\theta$. Following the standard practice in comparing discrepancies \citep{drovandi2022comparison}, we specify a common budget of $T=25{,}000$ simulations and define the \textsc{abc} threshold  to retain, for every discrepancy analyzed, the $1\%$ (i.e., $250$) values of $\theta$ that generated the synthetic data closest to the observed ones, under such a discrepancy. Although the theory in Section~\ref{sec_abc_posterior} can potentially guide the choice of the   threshold, similar guidelines are not yet available beyond the \textsc{ips} class. Hence, to ensure a fair assessment across all the discrepancies, we rely on the recommended practice in comparing the different \textsc{abc} implementations. As clarified in Table~\ref{tab:caption} and Figure~\ref{figure:post}, the discrepancies assessed are those most commonly used in \textsc{abc}, namely, the Wasserstein-1 distance as in Example~\ref{exm_wass}  \citep{bernton2019approximate}, \textsc{mmd} with Gaussian kernel defined in  Example~\ref{exm_mmd}  \citep{park2016k2,nguyen2020approximate}, and a summary-based distance leveraging the sample mean as the summary statistics (see Example~\ref{exm_sum}). For comparison, we also consider a popular discrepancy that does not belong to the \textsc{ips} class, i.e., the Kullback--Leibler divergence \citep{jiang2018approximate}. All these discrepancies can be effectively implemented in \texttt{R}, leveraging, in particular, the libraries \texttt{transport} and \texttt{eummd}. For \textsc{mmd}, the choice of the length-scale parameter $\sigma^2$ is based on the median heuristic \citep{gret2012} automatically implemented in the function \texttt{mmd}.

Leveraging the samples  from the \textsc{abc} posterior for $\theta$ under each discrepancy, we first estimate the mean squared error $\hat{\mathbb{E}}_{\textsc{abc}}(\theta-\theta_0)^2$ with respect to the location $\theta_0=1$ of the uncontaminated Student's~$t$, under the different discrepancies, thereby assessing performance  with a focus on a common metric on the space of parameters.  Table \ref{tab:caption} displays these error estimates, under each discrepancy and level of contamination, further averaged over $50$ simulated datasets to obtain a comprehensive assessment based on  replicated studies.  The results in Table \ref{tab:caption} together with the graphical representation in Figure \ref{figure:post} of the \textsc{abc} posteriors for one simulated dataset from the contaminated model with varying $\alpha\in\{0.05, 0.10, 0.15\}$ effectively illustrate an important practical consequence of the theory derived in Sections~\ref{sec_abc}--\ref{sec_mmd_tot}. Namely, when one is not sure, or cannot check, whether the assumed statistical model and/or the underlying data generating process  meet specific regularity conditions, effective \textsc{ips} discrepancies with guarantees of uniform convergence and concentration (e.g., \textsc{mmd} with bounded kernel) provide a robust and safe  default choice. When the level of contamination is mild ($\alpha=0.05$),  Wasserstein-\textsc{abc} achieves comparable concentration. This result aligns with Proposition~\ref{corr_Wasserstein_unb}, which however depends on the specific properties of the model and the underlying data generating process. Such a dependence is evident when the amount of contamination grows to $\alpha=0.10$ and $\alpha=0.15$. In these settings, the performance of Wasserstein-\textsc{abc} tends to slightly deteriorate, mainly due to a location shift in the induced posterior. This shift is even more evident for summary-based \textsc{abc} relying on the sample mean, that is not robust to location contaminations.  Conversely, the Kullback--Leibler divergence preserves robustness but at the expense of a lower concentration. Note that all the discrepancies analyzed in Table \ref{tab:caption} and Figure \ref{figure:post} are well-defined under the Student's~$t$ and Gaussian distributions considered in this simulation study. In particular, since the Student's $t$ has $3$ degrees of freedom, its mean and variance  are finite. Thus, Wasserstein-1 is well-defined for both the assumed model and the underlying data generating process.

Considering the running times, as displayed in Table \ref{tab:caption} all the discrepancies under analysis can be evaluated in the order of milliseconds for a sample of size $n=m=100$  on a standard laptop. This enables scalable and effective \texttt{R} implementations.


\vspace{13pt}
\section{Discussion}
\label{sec_discussion}
This article provides theoretical advancements with respect to the recent literature on the asymptotic properties of discrepancy-based \textsc{abc} posteriors by connecting these properties with the  behavior of  the Rademacher complexity associated with the chosen discrepancy. As clarified in the article and in the Supplementary Material, although the Rademacher complexity has never been considered in \textsc{abc}, this notion yields a powerful and promising perspective to derive general, informative and uniform convergence and concentration properties. 

While the above contribution already provides key advancements, the proposed perspective based on Rademacher complexity  has broader scope and sets the premises for additional future research. For example, as clarified in this article, any novel result and bound on the Rademacher complexity of specific discrepancies can be directly applied to \textsc{abc} theory through our framework. This may yield tighter or  more explicit bounds, possibly holding under milder assumptions and more general discrepancies.   For instance, to  our knowledge, informative bounds for the Rademacher complexity of the Wasserstein-1 distance are currently  available only  for bounded $\mathcal{Y}$ and, hence, it would be of interest to leverage future findings on the  unbounded $\mathcal{Y}$ case to further refine Proposition~\ref{corr_Wasserstein_unb} and  broaden the range of models for which our theory, when specialized to  Wasserstein-1 distance, applies. To this end, it might also be promising to explore  results on local Rademacher complexities \citep[e.g.,][]{bartlett2005local}, along with the proposed variable transformation strategy in Example~\ref{exm_wass1}. This solution requires the choice of a mapping $g(\cdot)$ which clearly influences the learning properties  of the induced distance  and, as such, requires further investigation. Finally, although Appendix~C in the Supplementary Material provides important extensions to non-i.i.d. settings, other relaxations beyond $\beta$-mixing processes could be of interest, such as, for instance, the case of independent but not identically distributed data. This setting can be addressed via the residual reconstruction strategy \citep[e.g.,][Section 4.2.3.]{bernton2019approximate} that would imply studying discrepancies among empirical distributions of residuals, for which i.i.d. assumptions are again  reasonable.

Although our focus is on the active convergence and concentration theory for discrepancy-based \textsc{abc}, other properties such as accuracy in uncertainty quantification of credible intervals and the limiting shapes of \textsc{abc}  posteriors, in correctly-specified models, have attracted interest in the context of summary-based and summary-free  \textsc{abc} \citep[][]{frazier2018asymptotic,frazier2020robust,wang2022}. While this direction goes beyond the scope of our article, extending and unifying such results, as done for concentration properties, is also of interest. 

While the \textsc{ips} class is broad, it  does not cover all discrepancies employed in \textsc{abc}. For example, the \textsc{kl} divergence \citep{jiang2018approximate} and the Hellinger distance \citep{frazier2020robust} are not \textsc{ips}, but rather belong to the class of $f$-divergences.  While this latter family has important differences relative to the \textsc{ips} class, it is worth investigating $f$-divergences in the light of our results. To accomplish this goal, an option is to exploit the unified treatment  of these two classes in, e.g., \citet{agrawal2021optimal} and \citet{birrell2022f}. More generally, our results could also stimulate methodological and theoretical advancements in  generalized likelihood-free Bayesian inference via discrepancy-based pseudo-posteriors \citep[e.g.,][]{bissiri2016general,jewson2018principles,miller2019robust,cherief2020mmd,matsubara2021robust}. The recent contribution by \citet{frazier2024} establishes, among other key results, important connections between such a latter framework and \textsc{abc} that could facilitate these advancements.

Finally, notice that, for the sake of simplicity and ease of  comparison with related studies, we have focused on rejection \textsc{abc}, and constrained the number $ m $ of synthetic samples to be equal to the sample size $ n $ of the observed data. While these settings are standard in state-of-the-art theory \citep[][]{bernton2019approximate,frazier2020robust}, other \textsc{abc} routines and alternative scenarios where $ m $ grows, e.g., sub-linearly, with $ n $ deserve further investigation. This latter regime would be  of interest in settings where the simulation of synthetic data is computationally expensive.


 \begin{acks}[Acknowledgments]
 We are grateful to the Editor, the Associate Editor and the referees for the constructive feedbacks, which helped us in improving the preliminary version of the article. 
  \end{acks}

\vspace{90pt}
\changefontsizes{14pt}

\begin{center}
\textcolor{white}{.}
\\
\textcolor{white}{.}
\\
\textcolor{white}{.}
\\
\textcolor{white}{.}
\\
\large
\uppercase
{\bf Supplementary Material to ``Concentration of discrepancy-based approximate Bayesian 	\\ computation  via Rademacher complexity"}
\end{center}

\numberwithin{equation}{section}
\numberwithin{table}{section}
\numberwithin{figure}{section}

\vspace{50pt}

\appendix
\section{Additional Integral Probability Semimetrics}\label{app1}
\vspace{10pt}

\subsection{Two additional examples of   \textsc{ips} discrepancies}

While \textsc{mmd}, Wasserstein-1 and summary-based distances provide the most notable examples of \textsc{ips} discrepancies employed in \textsc{abc},  two other relevant   \textsc{ips} instances are the total variation (\textsc{tv}) distance and the Kolmogorov--Smirnov (\textsc{ks}) distance, discussed below.

\begin{example}
 {\em (Total variation distance).} \label{exm_tot}
 Although the total variation distance is not a common choice within discrepancy-based \textsc{abc}, it still provides a notable example of  \textsc{ips}, obtained when $ \mathfrak{F}$ is the class of measurable functions whose sup-norm is bounded by $1$; i.e.  $ \mathfrak{F} = \{f : ||f||_\infty \leq 1 \} $.
\end{example}

\begin{example} {\em (Kolmogorov--Smirnov distance).}\label{exm_kol}
When $\mathcal{Y}=\mathbb{R}$ and $ \mathfrak{F} = \{\mathds{1}_{(-\infty,a]}\}_{a\in\mathbb{R}} $, then $\mathcal{D}_\mathfrak{F} $ is the Kolmogorov--Smirnov distance, which can also be written as $\mathcal{D}_\mathfrak{F}(\mu_1,\mu_2) = \sup_{x \in \mathcal{Y}} |F_1(x) - F_2(x)| $, where $F_1$ and $F_2$ are the cumulative distribution functions associated with $\mu_1$ and $\mu_2$, respectively.
\end{example}

\vspace{5pt}
\subsection{Validity of (\ref{C3})--(\ref{C4})  for the  \textsc{tv} distance and Kolmogorov--Smirnov distance}~Examples~\ref{exm_tot1} and \ref{exm_kol1} verify the validity of assumptions (\ref{C3}) and (\ref{C4}) under the \textsc{tv} distance and the Kolmogorov--Smirnov distance, respectively.

\begin{example}
{\em (Total variation distance).}
\label{exm_tot1} 
The \textsc{tv} distance satisfies  (\ref{C3}) by definition, but  in general not Assumption (\ref{C4}), unless the cardinality $|\mathcal{Y}|$ of $\mathcal{Y}$ is finite. In fact, when $\mathcal{Y}=\mathbb{R}$ and $\mu \in \mathcal{P}(\mathcal{Y})$ is continuous, the probability that there exists an index $i\neq i'$ such that $x_i=x_{i'}$ is zero.~Hence, with probability $1$, for any vector $\epsilon_{1:n}$ of Rademacher variables there always exists a~function $f_\epsilon$ from $\mathcal{Y}$ to $\{0;1\}$ such that $f_\epsilon(x_i) = \mathds{1}_{\{\epsilon_i=1\}}$. Therefore, $ \sup_{f\in\mathfrak{F}} | (1/n) \sum_{i=1}^{n} \epsilon_i f(x_i)| \geq  (1/n) \sum_{i=1}^n \mathds{1}_{\{\epsilon_i = 1\}} $, which  implies that the Rademacher complexity $ \mathfrak{R}_{\mu,n}(\mathfrak{F})$ is bounded below by $(1/n) \sum_{i=1}^n  \mathbb{P} (\epsilon_i=1) = 1/2. $ Nonetheless, as mentioned above, the \textsc{tv} distance can still satisfy (\ref{C4}) in specific contexts. For instance, leveraging the bound in Lemma 5.2 of \citet{massart2000some}, when the cardinality $|\mathcal{Y}|$ of $\mathcal{Y}$ is finite, there will be replicates in $[f(x_1), \ldots, f(x_n)]$  whenever $n > |\mathcal{Y}|$. Hence, as $n \rightarrow \infty$, it will be impossible to find a function in $ \mathfrak{F}$ which can interpolate any noise vector  of Rademacher variables with $[f(x_1), \ldots, f(x_n)]$, thus ensuring $\mathfrak{R}_{n}(\mathfrak{F}) \rightarrow 0$.
\end{example}

\begin{example}
{\em (Kolmogorov--Smirnov distance).}
\label{exm_kol1} 
The \textsc{ks} distance meets (\ref{C3}) by definition and, similarly to \textsc{mmd} with bounded kernels, also condition (\ref{C4}) is satisfied without the need to impose additional constraints on the model $\mu_\theta$ or on the data-generating process. More specifically, Assumption (\ref{C4}) follows from the inequality $ \mathfrak{R}_{\mu,n}(\mathfrak{F})  \leq 2 [\log(n+1)/n]^{1/2}$  in Chapter 4.3.1 of  \citet{wainwright2019high}. This is a consequence of the bounds on $ \mathfrak{R}_{\mu,n}(\mathfrak{F})$ when $\mathfrak{F}$ is a class of $b$-uniformly bounded functions such that, for  some $\nu \geq 1$, it holds~$ \mbox{\normalfont card}\{f(x_{1:n}): f \in \mathfrak{F} \}\leq (n+1)^{\nu}$ for any $n$ and  $x_{1:n}$ in $\mathcal{Y}^n$. When $ \mathfrak{F} = \{1_{(-\infty,a]}\}_{a\in\mathbb{R}} $ each $x_{1:n}$ would divide the real line in at most $n+1$ intervals and every indicator function within $\mathfrak{F} $ will take value $1$ for all $x_i \leq a$ and zero otherwise, meaning that $\mbox{\normalfont card}\{f(x_{1:n}): f \in \mathfrak{F} \} \leq (n+1)$. Therefore, by applying Equation (4.24) in \mbox{\citet{wainwright2019high}}, with $b=1$ and $\nu=1$, yields  $ \mathfrak{R}_{\mu,n}(\mathfrak{F})  \leq 2 [\log(n+1)/n]^{1/2}$ for any $\mu \in \mathcal{P}(\mathcal{Y}) $, which implies that Assumption \mbox{(\ref{C4})} is~met. These derivations clarify the usefulness of the available techniques for upper bounding the Rademacher complexity~\citep[e.g.,][Chapter 4.3]{wainwright2019high}, leveraging, in this case, the notion of polynomial discrimination and the closely-related \textsc{vc} dimension.
\end{example}

\vspace{5pt}
\section{Concentration in the space of parameters}
\label{sec_conc_theta_space}
Theorem~\ref{thm_rademacher} in the main article is stated for neighborhoods within the space of distributions. Although such a perspective is in line with the overarching focus of current theory for discrepancy-based \textsc{abc} \citep[][]{jiang2018approximate,bernton2019approximate,nguyen2020approximate,frazier2020robust,fujisawa2021gamma}, it shall be emphasized that similar results can be also derived in the space of parameters. To this end, it suffices to adapt Corollary~1 in \citet{bernton2019approximate}  to our general framework, under the same additional assumptions, which are adapted below to the whole \textsc{ips}~class.

\vspace{2pt}
\begin{enumerate}[(I)]
\setcounter{enumi}{4}	
\item \label{C5} The minimizer $\theta^*$ of  $\mathcal{D}_{\mathfrak{F}}(\mu_\theta,\mu^*)$ exists and is well separated, meaning that  for any $\delta>0$ there is a $\delta'>0$ such that $  \inf_{\theta\in\Theta: d(\theta,\theta^*) > \delta} \mathcal{D}_{\mathfrak{F}}(\mu_\theta,\mu^*) > \mathcal{D}_{\mathfrak{F}}(\mu_{\theta^*},\mu^*) + \delta'$;
\vspace{-4pt}
\item \label{C6} The parameters $\theta$ are identifiable, and there exist positive constants $K>0$, $\nu>0$ and an open neighborhood $U\subset \Theta$ of $\theta^*$ such that, for any $\theta\in U$, it holds that $ d(\theta,\theta^*) \leq K \left[ \mathcal{D}_{\mathfrak{F}}(\mu_\theta,\mu^*)-\varepsilon^*\right]^\nu. $
\end{enumerate}
\vspace{2pt}

Assumptions (\ref{C5}) and (\ref{C6}) essentially require that the parameters $\theta$ are identifiable, sufficiently well-separated, and that the distance $d(\cdot,\cdot)$ between parameter values~has~some~reasonable correspondence with the discrepancy $\mathcal{D}_{\mathfrak{F}}(\cdot,\cdot)$ among the associated distributions. Although these two assumptions introduce a condition on the model, it shall be emphasized that (\ref{C5}) and (\ref{C6}) are not specific to our framework  \citep[e.g.,][]{frazier2018asymptotic,bernton2019approximate,frazier2020robust}. On the contrary, these identifiability conditions are  arguably customary and minimal requirements in parameter inference. Moreover,  these two assumptions have been checked in~\citet{cherief2022finite} for \textsc{mmd} and in \citet{bernton2019approximate} for Wasserstein distance, which are arguably the two most remarkable examples of \textsc{ips} employed in the \textsc{abc}~context. Under  (\ref{C5}) and (\ref{C6}), it is possible to state Corollary~\ref{cor_1}.

\begin{corollary}\label{cor_1}
Assume (\ref{C1})--(\ref{C4}) along with (\ref{C5})--(\ref{C6}),  and that $\mathcal{D}_\mathfrak{F}$ denotes a discrepancy within the \textsc{ips} class in Definition~\ref{def_IPS}. Moreover, take $\bar{\varepsilon}_n \to 0$ as $ n \to \infty $, with $n\bar{\varepsilon}^2_n \to \infty$ and $\bar{\varepsilon}_n / \mathfrak{R}_{n}(\mathfrak{F}) \to \infty $. Then, the \textsc{abc} posterior with threshold $\varepsilon_n= \varepsilon^* + \bar{\varepsilon}_n $ satisfies 
\begin{eqnarray*}
\smash{ \pi_n^{(\varepsilon^*+\bar{\varepsilon}_n)}}\Bigl(\Bigl\{\theta: d(\theta,\theta^*) > K\Bigl[ \frac{4\bar{\varepsilon}_n}{3}+2\mathfrak{R}_{n}(\mathfrak{F}) +\Bigl(\frac{2b^2}{n}\log\frac{n}{\bar{\varepsilon}_n^L}\Bigr)^{1/2}\Bigr]^{\nu}\Bigr\} \Bigr) \leq \frac{2 \cdot 3^L}{c_\pi n},
\end{eqnarray*}
 with $\mathbb{P}_{y_{1:n}}$--probability going to $1$ as $n \rightarrow \infty$.
\end{corollary}

As for Theorem~\ref{thm_rademacher}, also Corollary~\ref{cor_1} holds more generally when replacing both $n/\bar{\varepsilon}_n^L$ and $c_\pi n$ with $M_n/\bar{\varepsilon}_n^L$ and $c_\pi M_n$, respectively, for any $M_n>1$. The proof of Corollary~\ref{cor_1} follows directly from Theorem~\ref{thm_rademacher} and Assumptions (\ref{C5})--(\ref{C6}), thereby allowing to inherit the discussion after Theorem~\ref{thm_rademacher}, also when the concentration is measured directly within the parameter space via $d(\theta,\theta^*) $. For instance, when $d(\cdot,\cdot)$ is the Euclidean distance and~$\nu=1$, this implies that whenever $\mathfrak{R}_{n}(\mathfrak{F}) =\mathcal{O}(n^{-1/2})$ the contraction rate will be in the order of $\mathcal{O}([\log(n)/n]^{1/2})$, which is the expected rate in parametric models.

\vspace{8pt}

\section{Extension to non-i.i.d. settings}
\label{sec_non_iid}
Although the theoretical results in Sections~\ref{sec_abc}--\ref{sec_mmd_tot} provide an improved understanding of~the limiting properties of discrepancy-based \textsc{abc} posteriors, the i.i.d. assumption in (\ref{C1}) rules out important settings which often require \textsc{abc}. A remarkable case is that of time-dependent observations   \citep[e.g.,][]{fearnhead2012constructing,bernton2019approximate,nguyen2020approximate,drovandi2022comparison}.~

Section~\ref{iid_def} clarifies that the  theory derived under i.i.d.\  assumptions  in Section~\ref{sec_abc}--\ref{sec_mmd_tot} can be naturally extended~to these non-i.i.d.\ settings leveraging results for Rademacher complexity in $\beta$-mixing~stochastic processes \citep{Mohri2008}. Examples~\ref{ex_doe}--\ref{ex_ch} below show that such a class embraces several processes of direct practical interest. Extensions beyond this class, albeit relevant, are challenging even when the focus is on proving~simpler, non-uniform, concentration results for a single discrepancy. Hence, these extensions are left for future research, that could be facilitated by the derivation of Rademacher complexity bounds for general  processes  beyond the $\beta$-mixing ones studied in  \citet{Mohri2008}.

\subsection{Convergence and concentration beyond i.i.d.\ settings}
\label{iid_def} Let us assume again that  $ \mathcal{Y} $ is a metric space endowed with distance~$\rho$. However, unlike the i.i.d.\ setting considered in Section~\ref{sec_abc}, we now focus on the situation in which the observed data $ y_{1:n}=(y_1,\ldots,y_n) \in \mathcal{Y}^n $ are dependent and drawn from the joint distribution $ \mu^{*(n)} \in  \mathcal{P}(\mathcal{Y}^n) $, where $ \mathcal{P}(\mathcal{Y}^n) $ is the space of probability measures on $ \mathcal{Y}^n$. Under this more general framework, the i.i.d. case is recovered by assuming that $ \mu^{*(n)}$ can be expressed  as a product, i.e., $ \mu^{*(n)} = \prod_{i=1}^n\mu^{*}$. 

In the following, the above product structure is not imposed. Instead, we  only assume~that the marginal of $ \mu^{*(n)} $ is constant, and  denoted with $\mu^*$. Such an assumption is met whenever $y_{1:n}$ is extracted from a stationary stochastic process $(y_t)_{t\in\mathbb{Z}}$, thus embracing~a~broader variety~of applications of direct interest. Under these settings, a statistical~model is defined as a collection of distributions in $\mathcal{P}(\mathcal{Y}^n) $, i.e., \smash{$ \{ \mu_\theta^{(n)}: \theta \in \Theta  \subseteq \mathbb{R}^{p} \} $}, with a constant marginal denoted by $\mu_{\theta}$. Notice that these assumptions of constant marginals $\mu^*$ and $\mu_\theta$ are made also in the available concentration theory under non-i.i.d.\ settings \citep[see e.g.,][]{bernton2019approximate,nguyen2020approximate} when requiring convergence of $ \mathcal{D}_\mathfrak{F} (\hat{\mu}_{y_{1:n}},\mu^*)$ and suitable concentration inequalities for $ \mathcal{D}_\mathfrak{F} (\hat{\mu}_{z_{1:n}},\mu_\theta)$. As a result, the settings we consider are not more~restrictive than those addressed in discrepancy-specific theory. In fact, both \citet{bernton2019approximate} and \citet{nguyen2020approximate} explicitly refer to stationary  processes when discussing the validity of the assumptions on $ \mathcal{D}_\mathfrak{F} (\hat{\mu}_{y_{1:n}},\mu^*)$ and $ \mathcal{D}_\mathfrak{F} (\hat{\mu}_{z_{1:n}},\mu_\theta)$ in non-i.i.d. contexts.

Given the above statistical model, a prior  $ \pi $ on~$ \theta $ and a generic \textsc{ips} discrepancy~$ \mathcal{D}_\mathfrak{F}  $,~the~\textsc{abc} posterior with threshold $\varepsilon_n \geq 0$ is defined as 
\begin{eqnarray*}
\pi_n^{(\varepsilon_n)}(\theta) \propto \pi(\theta) \int_{\mathcal{Y}^n} \mathds{1}\{\mathcal{D}_\mathfrak{F} (\hat{\mu}_{z_{1:n}},\hat{\mu}_{y_{1:n}})\leq\varepsilon_n\} \ \mu_\theta^{(n)}(dz_{1:n}).
\end{eqnarray*}
This definition is the same as the one  in Section~\ref{sec_abc}, with~the~only~difference that~\smash{$\mu_\theta^n=\prod_{i=1}^n \mu_\theta$} is replaced by the joint  \smash{$\mu^{(n)}_\theta$}, since now the data are no more assumed to be independent.

In order to extend the convergence result in Corollary~\ref{conv_rad} together with the concentration statement in Theorem~\ref{thm_rademacher} to the above framework, we require an analog of~Equation~\eqref{eq1} in Lemma~\ref{lemma_rademacher} for time-dependent data. This generalization can be derived leveraging results  in \citet{Mohri2008} under the notion of $\beta$-mixing coefficients.
\vspace{-3pt}
\begin{definition}[$\beta$-mixing]
\label{beta_mix}
Consider the stationary sequence  $(x_t)_{t\in\mathbb{Z}}$ of random variables, and let \smash{$\sigma^{j'}_{j}$} be the $\sigma$--algebra generated by the random variables $x_k$, $j \leq k \leq j'$, for any $j, j' \in \mathbb{Z} \cup \{-\infty, +\infty\}$. Then, for any integer $k>0$, the $\beta$-mixing coefficient of  $(x_t)_{t\in\mathbb{Z}}$ is defined as
\begin{eqnarray*}
 \beta(k) = \mbox{\normalfont sup}_{t\in\mathbb{Z}} \mathbb{E}[ \mbox{\normalfont  sup}_{A \in \sigma_{t+k}^\infty} |\mathbb{P}(A \mid \sigma_{-\infty}^t)- \mathbb{P}(A)|].
 \end{eqnarray*}  
If $\beta(k) \to 0$ as $k \to \infty$, then the stochastic process $(x_t)_{t\in\mathbb{Z}}$ is said to be $\beta$-mixing.
\end{definition}

Intuitively, $\beta(k)$ measures the dependence between the past (before $t$) and the future (after $t+k$) of the process. When such a dependence is weak, we expect that $\beta(k)$ will decay to $0$ fast when $k\rightarrow\infty$. In the most extreme case, when the $x_t$'s are i.i.d., we have $\beta(k)=0$ for all $k>0$. More generally, as clarified in Definition~\ref{beta_mix}, a process having $\beta(k) \rightarrow 0$ when $k\rightarrow\infty$ is named $\beta$-mixing. We refer the reader to~\cite{doukhan}  for an in-depth study of the main properties of $\beta$-mixing processes along with a more comprehensive discussion of relevant examples. The most remarkable ones will be also presented in the following. 

Leveraging the notion of $\beta$-mixing coefficient, Lemma~\ref{lemma_rademacher_dependent_full} extends Lemma~\ref{lemma_rademacher} to the dependent setting. The proof  can be found in Appendix D and combines  Proposition 2 and Lemma 2 in~\cite{Mohri2008}. For readability, let us also~introduce the notation $s_n = \lfloor n/(2 \lfloor \sqrt{n} \rfloor) \rfloor $. Note that $s_n \sim \sqrt{n}/2$ as $n\to\infty$, and thus $s_n\to\infty$.
\begin{lemma}
\label{lemma_rademacher_dependent_full}
Define $s_n = \lfloor n/(2 \lfloor \sqrt{n} \rfloor) \rfloor$. Moreover, consider the stationary stochastic process $(x_t)_{t\in\mathbb{Z}}$ and denote with $\beta(k), k \in \mathbb{N}$, its $\beta$-mixing coefficients. Let $\mu^{(n)}$ be the joint distribution of a sample $x_{1:n}$ extracted from $(x_t)_{t\in\mathbb{Z}}$ and denote with $\mu=\mu^{(1)}$ its constant marginal. Then, for any $b$-uniformly bounded class $\mathfrak{F}$, any integer $n \geq 1$ and scalar $\delta \geq 0$,
\begin{eqnarray}\label{eq1-dep-iid}
\qquad \  \ \ \mathbb{P}_{x_{1:n}} \left[\mathcal{D}_{\mathfrak{F}}(\hat{\mu}_{x_{1:n}},\mu) \leq  2 \mathfrak{R}_{\mu,s_n}(\mathfrak{F}) + \frac{4b}{\sqrt{n}} +\delta \right]\geq 1-2{\cdot}{\exp}\left[-\frac{s_n \delta^2}{2 b^2}\right]-2s_n \beta({\lfloor }\sqrt{n} {\rfloor}),
\end{eqnarray}
with $ \mathfrak{R}_{\mu,s_n}(\mathfrak{F})$ the Rademacher complexity in Definition~\ref{def_rade} for an~i.i.d.~sample of~size~$s_n$~from~$\mu$.
\end{lemma}

Equation~\eqref{eq1-dep-iid} extends  \eqref{eq1} beyond i.i.d.\ settings. This extension provides a bound that still depends on the Rademacher complexity in Definition~\ref{def_rade} for an i.i.d. sample --- in this case from the common marginal $\mu$ of the process $(x_t)_{t\in\mathbb{Z}}$. As such, Assumption (\ref{C4}) requires no modifications, and no additional validity checks relative to those discussed in Section~\ref{sec_ass}. This suggests that the Rademacher complexity framework might also be leveraged to derive improved convergence and concentration results for discrepancy-based \textsc{abc} posteriors in  more general situations which do not necessarily meet Assumption  (\ref{C1}). To prove these results we leverage again Assumptions (\ref{C2}), (\ref{C3}) and (\ref{C4}), and replace (\ref{C1}) with condition (\ref{C1dep}).
\begin{enumerate}[(I)] \addtocounter{enumi}{6}
\item \label{C1dep} The data $ y_{1:n} $ are from a $\beta$-mixing stochastic process $(y_t)_{t\in\mathbb{Z}}$  with mixing coefficients \smash{$\beta(k) \leq C_\beta e^{-\gamma k^{\xi}} $}  for some $C_\beta,\gamma, \xi>0$, common marginal $\mu$, and generic joint \smash{$\mu^{*(n')}$} for a sample $ y_{1:n'} $ from $(y_t)_{t\in\mathbb{Z}}$  for any $n'\in\mathbb{N}$. The same  $\beta$-mixing conditions hold also for the process $(z_t)_{t\in\mathbb{Z}}$ associated with the  synthetic data  $ z_{1:n} $  from the assumed model. In this case, the joint distribution for a generic sample  $ z_{1:n'} $ is \smash{$\mu^{(n')}_\theta$},~$\theta \in \Theta$, and the common marginal is denoted by $\mu_\theta$. For simplicity and without loss of generality, we also assume that the constants $C_\beta$, $\gamma$, and $\xi$ are the same for $(y_t)_{t\in\mathbb{Z}}$ and $(z_t)_{t\in\mathbb{Z}}$.
\end{enumerate}

Assumption (\ref{C1dep}) is clearly more general than  (\ref{C1}). As discussed previously, it embraces several stochastic processes of substantial interest in practical applications, including those in Examples~\ref{ex_doe}--\ref{ex_ch} below; see \citet{doukhan} for additional examples and discussion.

\begin{example}[Doeblin-recurrent Markov chains]
\label{ex_doe}
Let $(x_t)_{t\in\mathbb{Z}}$ be a Markov chain on $\mathcal{Y}\subset \mathbb{R}^d$ with transition kernel $P(\cdot,\cdot)$. Such a Markov chain is said to be Doeblin-recurrent if there exists a probability measure $q$, a constant $0<c \leq 1$ and an integer $r>0$ such that, for any measurable set $A$ and any $x\in\mathbb{R}^d$, $P^r (x,A) \geq c q(A)$. When this is the case, $(x_t)_{t\in\mathbb{Z}}$ is $\beta$-mixing with $\beta(k)\leq 2(1-c)^{k/r}$; see e.g., Theorem 1 in page 88 of \citet{doukhan}.
\end{example}

\begin{example}[Hidden Markov chains]
\label{ex_ch}
Assume $(x_t)_{t\in\mathbb{Z}}$ is a $\beta$-mixing stochastic process with coefficients $\beta_x(k), k\in\mathbb{N}$. If $\tilde{x}_t = F(x_t,\varepsilon_t)$ with $\varepsilon_t$ i.i.d., then the $\beta$-mixing coefficients of $(\tilde{x}_t)_{t\in\mathbb{Z}}$ satisfy $\beta_{\tilde{x}}(k) = \beta_x(k)$. Therefore, $(\tilde{x}_t)_{t\in\mathbb{Z}}$ is also $\beta$-mixing and inherits the bounds on $\beta_x(k)$. These processes are often used in practice with  $(x_t)_{t\in\mathbb{Z}}$ being a Markov chain. In this case $(\tilde{x}_t)_{t\in\mathbb{Z}}$ is called a Hidden Markov chain.
\end{example}

Section 2.4.2 of~\cite{doukhan} also provides conditions on $F$ and on the i.i.d. sequence $(\varepsilon_t)_{t\in\mathbb{Z}}$ ensuring that a stationary process $(x_t)_{t\in\mathbb{Z}}$ satisfying $ x_t = F(x_{t-1},\dots,x_{t-k},\varepsilon_t) $ exists and is $\beta$-mixing. Lemma~\ref{lemma_ar} specializes such a result in the context of Gaussian \textsc{ar}(1)~processes, which will be considered in the empirical study in Section~\ref{sec_empirical_noiid}.

\begin{lemma}[Gaussian \textsc{ar}(1) process]
\label{lemma_ar}
Consider a generic sequence $(\varepsilon_t)_{t\in\mathbb{Z}}$ of i.i.d.~random variables from a $\mbox{\normalfont N}(0,\sigma^2)$. Moreover, let $-1<\theta<1$ and $\psi\in\mathbb{R}$. Then the stationary solution to $x_t = \psi + \theta x_{t-1} + \varepsilon_t $ is  $\beta$-mixing and has coefficients $\beta(k) \leq \smash{|\theta|^k/(2\sqrt{1-\theta^2})}=(2\sqrt{1-\theta^2})^{-1}\exp(-k \log( 1/|\theta|))$, $k \in \mathbb{N}$, thus meeting (\ref{C1dep}).
\end{lemma}

Notice that in the empirical study in Section~\ref{sec_empirical_noiid} the focus will be on inference for the \textsc{ar} parameter $\theta$.~Clearly, in this case it is not sufficient to focus on the marginal distribution of each $x_t$. Rather, one should leverage the bivariate distribution for the pairs $\tilde{x}_t:=(x_t,x_{t+1})$; see also  \citet{bernton2019approximate} where such a strategy is named delay reconstruction. This~procedure simply changes the focus to the bivariate stochastic process $(\tilde{x}_t)_{t\in\mathbb{Z}}$, but does not alter the mixing properties. In particular, if $\beta_x(k)$ and $\beta_{\tilde{x}}(k)$ are the mixing coefficients of $(x_t)_{t\in\mathbb{Z}}$ and $(\tilde{x}_t)_{t\in\mathbb{Z}}$, respectively, then from Definition~\ref{beta_mix} we have $\beta_{\tilde{x}}(k) = \beta_x(k-1)$, for $k\geq 1$. Notice that identifiability is a key to ensure concentration in the space of parameters of interest as in  Corollary~\ref{cor_1}. This motivates further research, beyond the scope of this article, to derive delay reconstruction strategies ensuring identifiability in more complex processes.

Leveraging Lemma~\ref{lemma_rademacher_dependent_full} along with the newly-introduced Assumption~(\ref{C1dep}), Proposition \ref{conv_rad_dep} states convergence of the \textsc{abc} posterior when $\varepsilon_n=\varepsilon$ is fixed and $n \to \infty$.
\begin{proposition}
\label{conv_rad_dep}
  Under Assumptions (\ref{C3}), (\ref{C4}) and (\ref{C1dep}), for any $ \varepsilon>\tilde{\varepsilon}$, it holds that
  \begin{eqnarray}
   \pi_n^{(\varepsilon)}(\theta)  \rightarrow \pi( \theta \mid \mathcal{D}_{\mathfrak{F}}(\mu_\theta,\mu^*) \leq \varepsilon ) \propto  \pi( \theta ) \mathds{1}\left\{\mathcal{D}_{\mathfrak{F}}(\mu_\theta,\mu^*) \leq \varepsilon \right\},
   \label{conv_ABC_dep}
  \end{eqnarray}
almost surely with respect to $y_{1:n} \sim \mu^{*(n)}$, as $n \rightarrow \infty$.
\end{proposition}

According to Proposition~\ref{conv_rad_dep}, replacing Assumption (\ref{C1})  with (\ref{C1dep}), does not alter the~uniform convergence properties of the \textsc{abc} posterior originally stated in Corollary~\ref{conv_rad} under the i.i.d.\ assumption. This  allows to inherit~the discussion after  Corollary~\ref{conv_rad} also beyond i.i.d.\ settings, while suggesting that similar extensions would be possible in the regime $\varepsilon_n\rightarrow\varepsilon^*$ and $n\rightarrow\infty$. These extensions are stated in Theorem~\ref{thm_rademacher_dep}, which provides an important~generalization~of~Theorem~\ref{thm_rademacher} beyond the i.i.d.~case.

\begin{theorem}
	\label{thm_rademacher_dep}
	Let $\varepsilon_n = \varepsilon^* + \bar{\varepsilon}_n$, and  assume (\ref{C2}), (\ref{C3}), (\ref{C4}) and (\ref{C1dep}). Then, if \smash{$\bar{\varepsilon}_n\to 0$} is such that \smash{$ \sqrt{n} \bar{\varepsilon}_n^2 \to\infty $} and \smash{$\bar{\varepsilon}_n / \mathfrak{R}_{s_n}(\mathfrak{F}) \to \infty $}, with $s_n = \lfloor n/(2 \lfloor \sqrt{n} \rfloor) \rfloor$, we have
\begin{eqnarray*}
\smash{\pi_n^{(\varepsilon^*+\bar{\varepsilon}_n)}\Bigl(\Bigl\{\theta:  \mathcal{D}_{\mathfrak{F}}(\mu_\theta,\mu^*)} > \varepsilon^* + \frac{4\bar{\varepsilon}_n}{3}+2 \mathfrak{R}_{s_n}(\mathfrak{F}) + \frac{4b}{\sqrt{n}} +\Bigl(\frac{2b^2}{s_n}\log \frac{n}{\bar{\varepsilon}_n^L}\Bigr)^{1/2}  \Bigr\}\Bigr)
\leq \frac{4 \cdot 3^L}{c_\pi n},
\end{eqnarray*}
	with \smash{$\mathbb{P}_{y_{1:n}}$}--probability going to $1$ as $n \rightarrow \infty$, where $ \mathfrak{R}_{s_n}=\sup_{\mu \in \mathcal{P}(\mathcal{Y})}\mathfrak{R}_{\mu,s_n}$.
\end{theorem}

As for  Proposition~\ref{conv_rad_dep}, also Theorem~\ref{thm_rademacher_dep} shows that informative concentration inequalities similar to those derived in Section~\ref{sec_abc_con}, can be obtained beyond the i.i.d. setting. These results provide  insights comparable to those in Theorem~\ref{thm_rademacher} with the only difference that in this case we require \smash{$ \sqrt{n} \bar{\varepsilon}_n^2 \to\infty $} rather than \smash{$ n \bar{\varepsilon}_n^2 \to\infty $} and the term $2b^2/n$ within the bound in Theorem~\ref{thm_rademacher} is now replaced by $2b^2/s_n$ with $s_n \sim \sqrt{n}/2$~as $n\to\infty$. This means that $\bar{\varepsilon}_n$ must shrink to zero with a rate at least $n^{1/4}$ slower than the one allowed in the i.i.d.~setting. This is an interesting result which clarifies that when moving beyond i.i.d.~regimes concentration can still be achieved, although with a slower rate. Such a rate might be pessimistic in some models and we believe it may be improved under future refinements of Lemma~\ref{lemma_rademacher_dependent_full}.

Notice that (\ref{C1dep}) could be relaxed to include  $\beta$-mixing processes whose coefficients $\beta(k)$ vanish to zero, but at a non-exponential rate, e.g., $\beta(k) \sim 1/(k+1)^\xi$ for some $\xi>0$. In this case, we could still use Lemma~\ref{lemma_rademacher_dependent_full} to prove concentration,~but~with~a~smaller~$s_n$,~that~would lead to even slower rates. However, we did not provide the most general result for the sake of readability. As for processes that are not $\beta$-mixing, we are not aware of results similar to Lemma~\ref{lemma_rademacher_dependent_full} in this context. This is an important direction for future research.

\vspace{-2pt}

\subsection{Illustrative simulation in non-i.i.d. settings}
\label{sec_empirical_noiid}
Let us  illustrate the results in Section~\ref{iid_def}  on a simple simulation study focusing on a contaminated Gaussian $\textsc{ar}(1)$ process. More specifically, the uncontaminated data are generated from the model $y^*_t=0.5 y^*_{t-1} +\varepsilon_t$ for $t=1, \ldots, 100$ with $\varepsilon_t \sim \mbox{N}(0,1)$ independently, and initial state $y^*_0 \sim  \mbox{N}(0,1)$. Then,~similarly to the simulation study in Section~\ref{sec_empirical}, these data are contaminated with a growing fraction $\alpha \in \{0.05, 0.10, 0.15\}$ of independent realizations from a $\mbox{N}(20,1)$. As such, each observed data point $y_t$  is either equal to $y_t^*$  or to a sample from $\mbox{N}(20,1)$, for $t=1, \ldots, 100$. For Bayesian inference, we assume an $\textsc{ar}(1)$ model $z_t=\theta z_{t-1} +\varepsilon_t$, with $\varepsilon_t \sim \mbox{N}(0,1)$, and focus on learning $\theta$ via discrepancy-based \textsc{abc} under a uniform prior on $[-1,1]$ for $\theta$.

 \begin{table}[b]
 \vspace{-7pt}
	\centering
		\caption{Concentration and runtimes in seconds (for a single discrepancy evaluation) of \textsc{abc} under \textsc{mmd} with Gaussian kernel, Wasserstein-1 distance, summary-based distance (covariance) and \textsc{kl}  divergence for an $\textsc{ar}(1)$ Huber contamination model with $\alpha\in\{0.05, 0.10, 0.15\}$. \smash{$ \textsc{mse}= \hat{\mathbb{E}}_{\mu^{*(n)}}[\hat{\mathbb{E}}_{\textsc{abc}}(\theta-\theta_0)^2]$, $\theta_0=0.5$}.}
	\begin{tabular}{lcccc}
		\hline
		& \textsc{mse} ($\alpha=0.05$)& \textsc{mse} ($\alpha=0.10$) & \textsc{mse} ($\alpha=0.15$)& time  \\
		\hline
		(\textsc{ips}) \textsc{mmd}     &{\bf 0.029} & {\bf 0.036} &   {\bf 0.049}& $<$ 0.01'' \\
		(\textsc{ips})  Wasserstein-1      &  0.043  &  0.091 &0.180 & $<$ 0.01'' \\
		(\textsc{ips})  summary (covariance)    &  0.575  &  0.998 &1.001 & $<$ 0.01'' \\
			\cline{1-5}
		(non--\textsc{ips}) \textsc{kl}      &   0.058 &  0.060  &0.061   & $<$ 0.01''  \\
		\hline
	\end{tabular}
	\label{tab:nid}
	\vspace{-10pt}
\end{table}

Rejection \textsc{abc} is implemented under the same settings and discrepancies considered in Section~\ref{sec_empirical}. However, as discussed in Section~\ref{iid_def}, in this case we focus on distances among the empirical distributions of the $n=m=100$ observed $(y_0,y_1), (y_1, y_2), \ldots, (y_{99}, y_{100})$~and synthetic $(z_0,z_1), (z_1, z_2), \ldots, (z_{99}, z_{100})$ pairs. This is consistent with the delay reconstruction strategy in \citet{bernton2019approximate} and is motivated by the fact that information on~$\theta$ is in the bivariate distributions, rather than in the marginals. For the same reason, in implementing summary-based \textsc{abc} we consider the sample covariance rather~than~the~sample~mean.

 Table~\ref{tab:nid} summarizes the concentration achieved by the different discrepancies analyzed under the aforementioned non-i.i.d. data generating process and model, at varying contamination $\alpha \in \{0.05, 0.10, 0.15\}$. The results are coherent with those displayed in  Table~\ref{tab:caption} for the i.i.d. scenario and further clarify that discrepancies with guarantees of uniform convergence and concentration generally provide a robust choice, including in non-i.i.d. contexts.

\vspace{5pt}
\section{Proofs of Theorems, Corollaries and Propositions}\label{app2}
\begin{proof}[Proof of Theorem \ref{conv}]
Note that, by leveraging the first inequality in Lemma~\ref{lemma_rademacher}, we have $\mathbb{P}_{y_{1:n}} [\mathcal{D}_{\mathfrak{F}}(\hat{\mu}_{y_{1:n}},\mu^*) > 2 \mathfrak{R}_{\mu^*,n}(\mathfrak{F}) + \delta ]\leq \exp(-n \delta^2/2b^2).$ Hence, setting $\delta=1/n^{1/4}$, and recalling that $ \mathfrak{R}_{\mu^*,n}(\mathfrak{F})\leq \mathfrak{R}_{n}(\mathfrak{F})$, it follows  $\mathbb{P}_{y_{1:n}} [\mathcal{D}_{\mathfrak{F}}(\hat{\mu}_{y_{1:n}},\mu^*) > 2 \mathfrak{R}_{n}(\mathfrak{F}) + 1/n^{1/4} ]\leq \exp(-\sqrt{n}/2b^2)$; note that $\sum_{n\geq 0} \exp(-\sqrt{n} /2b^2) < \infty$. Therefore, if we define the event 
\vspace{-2pt}
\begin{equation}
\label{equa:proof:nonzero0}
 \mbox{E}_n = \{ \mathcal{D}_{\mathfrak{F}}(\hat{\mu}_{y_{1:n}},\mu^*) \leq 2 \mathfrak{R}_{n}(\mathfrak{F}) + 1/n^{1/4}\},
\end{equation}
then $ \mathds{1}\{\mbox{E}_n^c\} \rightarrow 0$ almost surely with respect to \smash{$y_{1:n} \stackrel{\scriptsize \mbox{i.i.d.}}{\sim} \mu^*$} as $n \rightarrow \infty$. Now, notice that
\begin{equation*}
\scalemath{0.96}{\pi_n^{(\varepsilon)}\{\theta:\mathcal{D}_{\mathfrak{F}}(\mu_\theta,\mu^* ) \leq \varepsilon \} \\
= \pi_n^{(\varepsilon)}\{\theta:\mathcal{D}_{\mathfrak{F}}(\mu_\theta,\mu^* ) \leq \varepsilon \} \mathds{1}\{\mbox{E}_n\} +\pi_n^{(\varepsilon)}\{\theta:\mathcal{D}_{\mathfrak{F}}(\mu_\theta,\mu^* ) \leq \varepsilon \} \mathds{1}\{\mbox{E}_n^c\}.}
\end{equation*}
Hence, in the following we focus on $\pi_n^{(\varepsilon)}\{\theta:\mathcal{D}_{\mathfrak{F}}(\mu_\theta,\mu^* ) \leq \varepsilon \} \mathds{1}\{\mbox{E}_n\}$. To this end, recall that
\vspace{-19pt}
 \begin{align*}
 \pi_n^{(\varepsilon)}(\theta) & \propto \pi( \theta ) \int \mathds{1}\left\{  \mathcal{D}_{\mathfrak{F}}(\hat{\mu}_{y_{1:n}},\hat{\mu}_{z_{1:n}}) \leq \varepsilon  \right\} \mu_{\theta}^n(dz_{1:n})
 \\
 & = \pi( \theta ) \int \mathds{1}\left\{ \mathcal{D}_{\mathfrak{F}}(\mu_\theta,\mu^*)    \leq \varepsilon + W_{\mathfrak{F}}(z_{1:n})  \right\} \mu_{\theta}^n(dz_{1:n})=: \pi(\theta) p_n(\theta),
 \end{align*}
 where $W_{\mathfrak{F}}(z_{1:n}) =\mathcal{D}_{\mathfrak{F}}(\mu_\theta,\mu^*) - \mathcal{D}_{\mathfrak{F}}(\hat{\mu}_{y_{1:n}},\hat{\mu}_{z_{1:n}})$, whereas $p_n(\theta)$ denotes the probability of generating a sample $z_{1:n}$ from $\mu_{\theta}^n$ which leads to accept the parameter value $\theta$. Note that, by applying the triangle inequality twice, we have $$- \mathcal{D}_{\mathfrak{F}}(\hat{\mu}_{z_{1:n}},\mu_\theta) - \mathcal{D}_{\mathfrak{F}}(\hat{\mu}_{y_{1:n}},\mu^*)\leq W_{\mathfrak{F}}(z_{1:n}) \leq  \mathcal{D}_{\mathfrak{F}}(\hat{\mu}_{z_{1:n}},\mu_\theta) + \mathcal{D}_{\mathfrak{F}}(\hat{\mu}_{y_{1:n}},\mu^*),$$
and, hence, $| W_{\mathfrak{F}}(z_{1:n}) | \leq   \mathcal{D}_{\mathfrak{F}}(\hat{\mu}_{z_{1:n}},\mu_\theta) + \mathcal{D}_{\mathfrak{F}}(\hat{\mu}_{y_{1:n}},\mu^*)$. This implies that the quantity $p_{n}(\theta)$ can be bounded below and above as follows
\vspace{-2pt}
\begin{equation*}
\begin{split}
&\int \mathds{1}\left\{ \mathcal{D}_{\mathfrak{F}}(\mu_\theta,\mu^*)    \leq \varepsilon - \mathcal{D}_{\mathfrak{F}}(\hat{\mu}_{z_{1:n}},\mu_\theta) - \mathcal{D}_{\mathfrak{F}}(\hat{\mu}_{y_{1:n}},\mu^*)  \right\} \mu_{\theta}^n(dz_{1:n}) \\
&\qquad  \leq p_n(\theta)  \leq \int \mathds{1}\left\{ \mathcal{D}_{\mathfrak{F}}(\mu_\theta,\mu^*)    \leq \varepsilon + \mathcal{D}_{\mathfrak{F}}(\hat{\mu}_{z_{1:n}},\mu_\theta) + \mathcal{D}_{\mathfrak{F}}(\hat{\mu}_{y_{1:n}},\mu^*)  \right\} \mu_{\theta}^n(dz_{1:n}).
\end{split}
\end{equation*}
Applying again Lemma~\ref{lemma_rademacher} yields $\mathbb{P}_{z_{1:n}} \left[\mathcal{D}_{\mathfrak{F}}(\hat{\mu}_{z_{1:n}},\mu_\theta) > 2 \mathfrak{R}_{\mu_\theta,n}(\mathfrak{F}) + \delta \right]\leq \exp(-n \delta^2/2b^2)$.
Therefore, setting $\delta=1/n^{1/4}$, and recalling that $ \mathfrak{R}_{\mu_\theta,n}(\mathfrak{F}) \leq  \mathfrak{R}_{n}(\mathfrak{F})$ and that we are on the event given in~\eqref{equa:proof:nonzero0}, it follows
\vspace{-5pt}
\begin{multline}
-\exp(-\sqrt{n}/2b^2)+  \mathds{1}\{ \mathcal{D}_{\mathfrak{F}}(\mu_\theta,\mu^*)    \leq \varepsilon - 4 \mathfrak{R}_{n}(\mathfrak{F}) -2/n^{1/4} \}
\\ \leq p_n(\theta) \leq \exp(-\sqrt{n}/2b^2)+  \mathds{1}\{ \mathcal{D}_{\mathfrak{F}}(\mu_\theta,\mu^*)    \leq \varepsilon +  4 \mathfrak{R}_{n}(\mathfrak{F}) + 2/n^{1/4}   \}.
\label{equa:proof:nonzero}
\end{multline}
Now, notice that the acceptance probability is defined as $p_n=\int p_n(\theta) \pi(d \theta)$. Hence, integrating with respect to $\pi(\theta)$ in the above inequalities yields, for $n$ large enough,
$$ \scalemath{1}{ \pi\{\theta: \mathcal{D}_{\mathfrak{F}}(\mu_\theta,\mu^* )\leq  \varepsilon - c_{\mathfrak{F}}\} - e_n  \leq   p_n   \leq    \pi\{\theta: \mathcal{D}_{\mathfrak{F}}(\mu_\theta,\mu^* )\leq  \varepsilon + c_{\mathfrak{F}}\} + e_n,}
$$
where $c_{\mathfrak{F}}= 4 \lim\sup \mathfrak{R}_n(\mathfrak{F})$, as in Equation \eqref{eq_pABC}, thus concluding the first part of the proof.

To proceed with the second part of the proof, notice that, by the definition of $\tilde{\varepsilon}= \inf \{ \epsilon>0: \smash{\pi\{\theta:\mathcal{D}_{\mathfrak{F}}(\mu_\theta,\mu^* ) \leq\epsilon \}}>0 \}$, the left part of the above inequality is bounded away from zero for $n$ large enough,  whenever $\varepsilon -c_{\mathfrak{F}}  > \tilde{\varepsilon}$. This implies that also the acceptance probability $p_n$ is strictly positive.  As a consequence, for such $n$, it follows that
$$ \pi_n^{(\varepsilon)}(\mbox{A}) = \frac{\int  p_n(\theta)  \mathds{1}_{\mbox{\small A}}(\theta) \pi(d\theta)   } { \int p_n(\theta)\pi(d\theta)  }= \frac{\int  p_n(\theta)  \mathds{1}_{\mbox{\small A}}(\theta) \pi(d\theta)   } { p_n  },$$
is well-defined for any event $\mbox{A}$. Then, leveraging the upper bound in~\eqref{equa:proof:nonzero} yields
\begin{align*}
& \pi_n^{(\varepsilon)}\{\theta:\mathcal{D}_{\mathfrak{F}}(\mu_\theta,\mu^* ) > \varepsilon + 4 \mathfrak{R}_n(\mathfrak{F}) \}=  \frac{\int  p_n(\theta)
\mathds{1}\left\{ \mathcal{D}_{\mathfrak{F}}(\mu_\theta,\mu^* ) > \varepsilon + 4 \mathfrak{R}_n(\mathfrak{F}) \right\}
\pi(d\theta)   } { \int p_n(\theta)\pi(d\theta)  }
\\
& \qquad \leq \frac{\int
\mathds{1}\left\{ \mathcal{D}_{\mathfrak{F}}(\mu_\theta,\mu^* ) \leq \varepsilon + 4 \mathfrak{R}_n(\mathfrak{F}) + 2/n^{1/4} \right\}
\mathds{1}\left\{ \mathcal{D}_{\mathfrak{F}}(\mu_\theta,\mu^* ) > \varepsilon + 4 \mathfrak{R}_n(\mathfrak{F}) \right\}
\pi(d\theta)   } { p_n }
\\
& \qquad \qquad + \frac{\int  \exp(-\sqrt{n}/2b^2)
\mathds{1}\left\{  \mathcal{D}_{\mathfrak{F}}(\mu_\theta,\mu^* ) > \varepsilon + 4 \mathfrak{R}_n(\mathfrak{F})\right\}
\pi(d\theta)   } {p_n  }.
\end{align*}
To conclude the proof it is now necessary to control both terms. Note that we already proved that the denominator $p_n$ is bounded away from zero for $n$ large enough. Both numerators are bounded by $1$, and going to $0$ when $n\rightarrow\infty$. Thus, by the dominated convergence theorem, both summands in the above upper bound for \smash{$ \pi_n^{(\varepsilon)}\{\theta:\mathcal{D}_{\mathfrak{F}}(\mu_\theta,\mu^* ) > \varepsilon + 4 \mathfrak{R}_n(\mathfrak{F}) \}$}~go~to~zero.  This implies, as a direct consequence, that  $\smash{\pi_n^{(\varepsilon)}}\{\theta: \mathcal{D}_{\mathfrak{F}}(\mu_\theta,\mu^* ) \leq \varepsilon + 4 \mathfrak{R}_n(\mathfrak{F})\} \rightarrow 1$, almost surely with respect to \smash{$y_{1:n}\stackrel{\scriptsize \mbox{i.i.d.}}{\sim}  \mu^*$}  as $n \rightarrow \infty$, thereby concluding the proof.
\end{proof}

\begin{proof}[Proof of Corollary~\ref{conv_rad}]
Note that by combining Equation~\eqref{eq1} in Lemma~\ref{lemma_rademacher}~with~the result $\sum_{n>0} \exp[-n \delta^2/(2b^2)] < \infty$, the Borel--Cantelli Lemma implies that both $\mathcal{D}_{\mathfrak{F}}(\hat{\mu}_{z_{1:n}},\mu_{\theta})$ and $\mathcal{D}_{\mathfrak{F}}(\hat{\mu}_{y_{1:n}},\mu^*)$ converge to 0 almost surely when $\mathfrak{R}_n(\mathfrak{F}) \rightarrow 0$ as~$n \rightarrow \infty$. Hence, since  $$ \scalemath{0.94}{-\mathcal{D}_\mathfrak{F} (\hat{\mu}_{z_{1:n}},\mu_\theta) -\mathcal{D}_\mathfrak{F}(\hat{\mu}_{y_{1:n}},\mu^*) \leq \mathcal{D}_{\mathfrak{F}}(\mu_\theta,\mu^*) - \mathcal{D}_{\mathfrak{F}}(\hat{\mu}_{y_{1:n}},\hat{\mu}_{z_{1:n}}) \leq  \mathcal{D}_\mathfrak{F} (\hat{\mu}_{z_{1:n}},\mu_\theta)+  \mathcal{D}_\mathfrak{F}(\hat{\mu}_{y_{1:n}},\mu^*)},$$
it follows that $\mathcal{D}_\mathfrak{F} (\hat{\mu}_{z_{1:n}},\hat{\mu}_{y_{1:n}}) \rightarrow\mathcal{D}_\mathfrak{F} (\mu_\theta, \mu^*) $ almost surely as $n \rightarrow \infty$. 

Combining this result with the proof of Theorem 1 in \citet{jiang2018approximate} yields the statement of Corollary~\ref{conv_rad}. Notice that, as discussed in Section~\ref{sec_abc_lim}, the limiting pseudo-posterior in Corollary~\ref{conv_rad}  is well-defined only for those $\varepsilon > \tilde{\varepsilon}$, with $\tilde{\varepsilon}$ as in Theorem~\ref{conv}. 
\end{proof}

\begin{proof}[Proof of Theorem \ref{thm_rademacher}]
Since Lemma~\ref{lemma_rademacher} and $ \mathfrak{R}_{n}(\mathfrak{F})=\sup_{\mu \in  \mathcal{P}(\mathcal{Y})} \mathfrak{R}_{\mu,n}(\mathfrak{F}) \geq \mathfrak{R}_{\mu,n}(\mathfrak{F})$ hold for every $\mu  \in  \mathcal{P}(\mathcal{Y})$, then, for every  integer $n \geq 1$ and any scalar $\delta \geq 0$, Equation~\eqref{eq1}  implies $\mathbb{P}_{x_{1:n}} \left[
\mathcal{D}_{\mathfrak{F}}(\hat{\mu}_{x_{1:n}},\mu) \leq  2 \mathfrak{R}_{n}(\mathfrak{F})+\delta \right]\geq 1-\exp(-n \delta^2/2b^2).$
Moreover, since this result holds for any $\delta \geq 0$, it follows that
$ \mathbb{P}_{x_{1:n}} [\mathcal{D}_{\mathfrak{F}}(\hat{\mu}_{x_{1:n}},\mu) \leq 2\mathfrak{R}_{n}(\mathfrak{F})+(c_1 - 2\mathfrak{R}_{n}(\mathfrak{F}))] \geq 1 - \exp[-n (c_1-2 \mathfrak{R}_{n}(\mathfrak{F}))^2/2b^2],$ for any $c_{1}  \geq 2 \mathfrak{R}_{n}(\mathfrak{F})$. Hence,
\begin{equation}
	\label{eq1_proof_1}
	\mathbb{P}_{x_{1:n}} \left[
	\mathcal{D}_{\mathfrak{F}}(\hat{\mu}_{x_{1:n}},\mu) \leq  c_1 \right]\geq 1-\exp[-n (c_1-2 \mathfrak{R}_{n}(\mathfrak{F}))^2/2b^2].
\end{equation}
Recalling the settings of Theorem~\ref{thm_rademacher}, consider the sequence $\bar{\varepsilon}_n \to 0$ as $ n \to \infty $, with $n\bar{\varepsilon}^2_n \to \infty$ and $\bar{\varepsilon}_n / \mathfrak{R}_{n}(\mathfrak{F}) \to \infty$, which is possible by Assumption~(\ref{C4}). These regimes imply that $\bar{\varepsilon}_n$ goes to zero slower than $ \mathfrak{R}_{n}(\mathfrak{F}) $
and, hence, for $n$ large enough,  $\bar{\varepsilon}_n /3 > 2 \mathfrak{R}_{n}(\mathfrak{F})$. Therefore, under Assumptions (\ref{C1})--(\ref{C3}) it is now possible to apply  \eqref{eq1_proof_1} to $y_{1:n}$, by setting $c_1=\bar{\varepsilon}_n /3$, which yields
\vspace{3pt}
\begin{equation*}
	\mathbb{P}_{y_{1:n}}\left[ \mathcal{D}_{\mathfrak{F}}(\hat{\mu}_{y_{1:n}},\mu^*) \leq 
	\bar{\varepsilon}_n/3 \right] \geq 1- \exp[-n (\bar{\varepsilon}_n/3-2 \mathfrak{R}_{n}(\mathfrak{F}))^2/2b^2].
\end{equation*}

\noindent Since $-n (\bar{\varepsilon}_n/3-2 \mathfrak{R}_{n}(\mathfrak{F}))^2=-n\bar{\varepsilon}^2_n[1/9+4(\mathfrak{R}_{n}(\mathfrak{F})/\bar{\varepsilon}_n)^2-(4/3)\mathfrak{R}_{n}(\mathfrak{F})/\bar{\varepsilon}_n]$, it follows that $-n (\bar{\varepsilon}_n/3-2 \mathfrak{R}_{n}(\mathfrak{F}))^2 \rightarrow - \infty$  when $ n \to \infty $. From the above settings we  also have that $n\bar{\varepsilon}^2_n \to~\infty$ and $\mathfrak{R}_{n}(\mathfrak{F})/\bar{\varepsilon}_n \to 0$, when  $ n \to \infty $. Therefore, as a consequence, we obtain $1- \exp[-n (\bar{\varepsilon}_n/3-2 \mathfrak{R}_{n}(\mathfrak{F}))^2/2b^2] \rightarrow 1$ as $ n \to \infty $.
Hence, in the rest of this proof, we will restrict to the event
$ \{ \mathcal{D}_{\mathfrak{F}}(\hat{\mu}_{y_{1:n}},\mu^*) \leq \bar{\varepsilon}_n/3 \} $. 

Denote with $\mathbb{P}_{\theta, z_{1:n} }$ the joint distribution of $\theta \sim \pi$ and $ z_{1:n}$ i.i.d.\ from $\mu_\theta$. By definition of conditional probability, for any $c_2$, including $c_2>2 \mathfrak{R}_{n}(\mathfrak{F})$, it follows that
\vspace{3pt}
\begin{equation}
	\label{eq1_proof_rate}
	\begin{split}
		& \pi_n^{(\varepsilon^*+\bar{\varepsilon}_n)}\left(\{\theta:  \mathcal{D}_{\mathfrak{F}}(\mu_\theta,\mu^*) > \varepsilon^* + 4\bar{\varepsilon}_n/3 + c_2\} \right)\\
		& \qquad = \frac{\mathbb{P}_{\theta, z_{1:n} } \left[ \mathcal{D}_{\mathfrak{F}}(\mu_\theta,\mu^*) >  \varepsilon^* + 4\bar{\varepsilon}_n/3 + c_2 ,   \mathcal{D}_{\mathfrak{F}}(\hat{\mu}_{z_{1:n}},\hat{\mu}_{y_{1:n}}) \leq \varepsilon^*+\bar{\varepsilon}_n \right]}
		{\mathbb{P}_{\theta, z_{1:n} }  \left[  \mathcal{D}_{\mathfrak{F}}(\hat{\mu}_{z_{1:n}},\hat{\mu}_{y_{1:n}}) \leq \varepsilon^*+\bar{\varepsilon}_n \right]}.
	\end{split}
\end{equation}
To derive an upper bound for the above ratio, we first identify an upper bound for its numerator. In addressing this goal, we leverage the triangle inequality $\mathcal{D}_\mathfrak{F} (\mu_\theta, \mu^*) \leq  \mathcal{D}_\mathfrak{F} (\hat{\mu}_{z_{1:n}},\mu_\theta) +\mathcal{D}_\mathfrak{F} (\hat{\mu}_{z_{1:n}},\hat{\mu}_{y_{1:n}})+  \mathcal{D}_\mathfrak{F}(\hat{\mu}_{y_{1:n}},\mu^*)$, since $ \mathcal{D}_\mathfrak{F} $ is a semimetric, and the previously-proved result that the event $ \{ \mathcal{D}_{\mathfrak{F}}(\hat{\mu}_{y_{1:n}},\mu^*) \leq \bar{\varepsilon}_n/3 \} $ has $\mathbb{P}_{y_{1:n}}$--probability going to $1$, thereby obtaining 
\vspace{3pt}
\begin{align*}
	&\mathbb{P}_{\theta, z_{1:n} }  [ \mathcal{D}_{\mathfrak{F}}(\mu_\theta,\mu^*)  > \varepsilon^* + 4\bar{\varepsilon}_n/3 + c_2 , \mathcal{D}_{\mathfrak{F}}(\hat{\mu}_{z_{1:n}},\hat{\mu}_{y_{1:n}}) \leq \varepsilon^*+\bar{\varepsilon}_n ]
	\\
	& \qquad \quad \leq  \mathbb{P}_{\theta, z_{1:n} } [   \mathcal{D}_\mathfrak{F} (\hat{\mu}_{z_{1:n}},\hat{\mu}_{y_{1:n}})+ \mathcal{D}_\mathfrak{F} (\hat{\mu}_{z_{1:n}},\mu_\theta) + \mathcal{D}_\mathfrak{F}(\hat{\mu}_{y_{1:n}},\mu^*) > \varepsilon^* + 4\bar{\varepsilon}_n/3 + c_2 , \\
	& \qquad \quad \qquad   \qquad   \mathcal{D}_{\mathfrak{F}}(\hat{\mu}_{z_{1:n}},\hat{\mu}_{y_{1:n}}) \leq \varepsilon^* + \bar{\varepsilon}_n ] 
	\\
	&  \qquad \quad \leq  \mathbb{P}_{\theta, z_{1:n} } \left[ \mathcal{D}_\mathfrak{F} (\hat{\mu}_{z_{1:n}},\mu_\theta) + \mathcal{D}_\mathfrak{F}(\hat{\mu}_{y_{1:n}},\mu^*) >  \bar{\varepsilon}_n/3 + c_2 \right]  \leq   \mathbb{P}_{\theta, z_{1:n} }\left[  \mathcal{D}_\mathfrak{F} (\hat{\mu}_{z_{1:n}},\mu_\theta)> c_2 \right]. 
\end{align*}
Rewriting $\mathbb{P}_{\theta, z_{1:n} }\left[  \mathcal{D}_\mathfrak{F} (\hat{\mu}_{z_{1:n}},\mu_\theta)> c_2 \right]$ 
as $\int_{\theta \in \Theta}\mathbb{P}_{z_{1:n} }\left[  \mathcal{D}_\mathfrak{F} (\hat{\mu}_{z_{1:n}},\mu_\theta)> c_2 \mid \theta \right] \pi({\rm d}\theta)$ and applying \eqref{eq1_proof_1} to $z_{1:n}$ yields 
\begin{equation*}
	\begin{split}
		{\int_{\theta \in \Theta}}&{\mathbb{P}_{z_{1:n} }}[ \mathcal{D}_\mathfrak{F} (\hat{\mu}_{z_{1:n}}{,}\mu_\theta)> c_2 {\mid} \theta] \pi({\rm d}\theta) = \int_{\theta \in \Theta}(1-{\mathbb{P}_{z_{1:n} }}[  \mathcal{D}_\mathfrak{F} (\hat{\mu}_{z_{1:n}}{,}\mu_\theta) \leq c_2 {\mid} \theta]) \pi({\rm d}\theta) \\
		&\quad \leq \int_{\theta \in \Theta}\exp[-n (c_2-2 \mathfrak{R}_{n}(\mathfrak{F}))^2/2b^2] \pi({\rm d}\theta)=\exp[-n (c_2-2 \mathfrak{R}_{n}(\mathfrak{F}))^2/2b^2].
	\end{split}
\end{equation*}
Hence, the numerator of the ratio in Equation \eqref{eq1_proof_rate} can be upper bounded by $\exp[-n (c_2-2 \mathfrak{R}_{n}(\mathfrak{F}))^2/2b^2]$ for any $c_2>2 \mathfrak{R}_{n}(\mathfrak{F})$.
As for the denominator,
defining the event $\mbox{E}_n := \{ \theta \in \Theta : \mathcal{D}_{\mathfrak{F}}(\mu_\theta,\mu^*) \leq  \varepsilon^*+\bar{\varepsilon}_n/3 \} $ and applying again the triangle inequality,
we have that
\begin{align*}
	&\mathbb{P}_{\theta, z_{1:n} }   [\mathcal{D}_{\mathfrak{F}}(\hat{\mu}_{z_{1:n}},\hat{\mu}_{y_{1:n}}) \leq \varepsilon^*{+}\ \bar{\varepsilon}_n]
	\geq \int_{{\footnotesize \mbox{E}}_n} \mathbb{P}_{z_{1:n}}\left[\mathcal{D}_{\mathfrak{F}}(\hat{\mu}_{z_{1:n}},\hat{\mu}_{y_{1:n}}) \leq \varepsilon^*{+} \ \bar{\varepsilon}_n \mid \theta \right] \pi({\rm d}\theta)
	\\
	& \qquad \geq \int_{{\footnotesize \mbox{E}}_n}  \mathbb{P}_{z_{1:n}}\left[ \mathcal{D}_{\mathfrak{F}}(\hat{\mu}_{y_{1:n}},\mu^*) + \mathcal{D}_{\mathfrak{F}}(\mu_\theta,\mu^*) +  \mathcal{D}_{\mathfrak{F}}(\hat{\mu}_{z_{1:n}},\mu_\theta)
	\leq \varepsilon^*+\bar{\varepsilon}_n \mid \theta \right] \pi({\rm d}\theta)
	\\
	& \qquad \geq \int_{{\footnotesize \mbox{E}}_n}   \mathbb{P}_{z_{1:n}}\left[ \mathcal{D}_{\mathfrak{F}}(\hat{\mu}_{z_{1:n}},\mu_\theta)
	\leq \bar{\varepsilon}_n/3 \mid \theta \right] \pi({\rm d}\theta),
\end{align*}
where the last inequality follows directly from the fact that it is possible to restrict to the event $ \{ \mathcal{D}_{\mathfrak{F}}(\hat{\mu}_{y_{1:n}},\mu^*) \leq \bar{\varepsilon}_n/3 \} $, and that we are integrating over $\mbox{E}_n := \{ \theta \in \Theta : \mathcal{D}_{\mathfrak{F}}(\mu_\theta,\mu^*) \leq \varepsilon^*+ \bar{\varepsilon}_n/3 \} $. Applying~again~\eqref{eq1_proof_1} to $z_{1:n}$, with $c_1 = \bar{\varepsilon}_n/3 > 2\mathfrak{R}_{n}(\mathfrak{F})$, the last term of the above inequality can be further lower bounded~by
\begin{align*}
&{\int_{{\footnotesize \mbox{E}}_n}}{(} 1-\exp[ -n(\bar{\varepsilon}_n/3- 2\mathfrak{R}_{n}(\mathfrak{F}))^2/2b^2 ] ) \pi({\rm d}\theta)\\
& \qquad\qquad \qquad = \pi( \mbox{E}_n)(  1-\exp[ -n(\bar{\varepsilon}_n/3- 2\mathfrak{R}_{n}(\mathfrak{F}))^2/2b^2 ] ),
\end{align*}
with $\pi( \mbox{E}_n) \geq c_\pi (\bar{\varepsilon}_n/3)^L$ by (\ref{C2}), and, as shown before, $ 1-\exp[ -n(\bar{\varepsilon}_n/3- 2\mathfrak{R}_{n}(\mathfrak{F}))^2/2b^2 ] \to 1 $, when $ n \to \infty $, which implies, for $n$ large enough, $ 1-\exp[ -n(\bar{\varepsilon}_n/3- 2\mathfrak{R}_{n}(\mathfrak{F}))^2/2b^2 ] > 1/2$. Leveraging both results, the denominator in  \eqref{eq1_proof_rate} is lower bounded by  \smash{$(c_\pi/2) \left(\bar{\varepsilon}_n/{3}\right)^L$}.
Let us now combine the upper and lower bounds derived, respectively, for the numerator and the denominator of the ratio in \eqref{eq1_proof_rate}, to obtain
\begin{equation}
	\label{eq2_proof_rate}
	\pi_n^{(\varepsilon^*+\bar{\varepsilon}_n)}(\{\theta \in \Theta:  \mathcal{D}_{\mathfrak{F}}(\mu_\theta{,}\mu^*) > \varepsilon^* {+} \ 4\bar{\varepsilon}_n/3 + c_2\})
	\leq \frac{ \exp[-n (c_2-2 \mathfrak{R}_{n}(\mathfrak{F}))^2/2b^2]
	}{(c_\pi/2) (\bar{\varepsilon}_n/3)^L  },
\end{equation}
with $\mathbb{P}_{y_{1:n}}$--probability going to $1$ as $n \rightarrow \infty$. To conclude the proof it suffices to replace $c_2$ in \eqref{eq2_proof_rate} with \smash{$2 \mathfrak{R}_{n}(\mathfrak{F})+\sqrt{(2b^2/n)\log(M_n/\bar{\varepsilon}_n^L)}$}, which is never lower than $2 \mathfrak{R}_{n}(\mathfrak{F})$. Finally, setting $M_n=n
$ yields the statement of  Theorem~\ref{thm_rademacher}.
\end{proof}

\begin{proof}[Proof of Corollary~\ref{cor_1}]
Corollary~\ref{cor_1} follows by replacing the bounds in the~proof of Corollary 1 by \citet{bernton2019approximate} with the  newly-derived ones   in~Theorem \ref{thm_rademacher}. 
\end{proof}

\begin{proof}[Proof of Corollary~\ref{corr_MMD_2}]
Recall that in the case of \textsc{mmd} with  kernels bounded by $1$ we have  $\mathfrak{R}_{n}(\mathfrak{F}) \leq n^{-1/2}$. Hence, regarding the upper and lower bounds on $p_n$ in \eqref{eq_pABC} it holds 
\begin{equation*}
\begin{split}
& \pi\{\theta: \mathcal{D}_{\textsc{mmd}}(\mu_\theta,\mu^* )\leq  \varepsilon - c_{\mathfrak{F}}\} \geq  \pi\{\theta: \mathcal{D}_{\textsc{mmd}}(\mu_\theta,\mu^* )\leq  \varepsilon - 4/\sqrt{n}\},\\
& \pi\{\theta: \mathcal{D}_{\textsc{mmd}}(\mu_\theta,\mu^* )\leq  \varepsilon + c_{\mathfrak{F}}\}  \leq  \pi\{\theta: \mathcal{D}_{\textsc{mmd}}(\mu_\theta,\mu^* )\leq  \varepsilon +4/\sqrt{n}\}.
\end{split}
\end{equation*}
Combining the above inequalities with the result in \eqref{eq_pABC}, and taking the limit for $n \rightarrow \infty$, proves the first part of the statement. The second part is a direct application of Corollary~\ref{conv_rad} to the case of  \textsc{mmd} with bounded kernels, after noticing that the aforementioned inequality $\mathfrak{R}_{n}(\mathfrak{F}) \leq n^{-1/2}$ implies $\mathfrak{R}_n(\mathfrak{F}) \rightarrow 0$ as $n \rightarrow \infty$.
\end{proof}

\begin{proof}[Proof of Corollary \ref{corr_MMD}]
To prove Corollary \ref{corr_MMD}, it suffices to plug $\bar{\varepsilon}_n=[(\log n)/n]^{\frac{1}{2}}$ and ${b=1}$ into the statement of Theorem~\ref{thm_rademacher}, and then upper-bound the resulting radius via the inequalities $\mathfrak{R}_{n}(\mathfrak{F}) \leq n^{-1/2}$ and $\log n \geq 1$. The latter holds for any $n\geq 3$ and hence for $n \rightarrow \infty$. 
\end{proof}

\begin{proof}[Proof of Proposition \ref{corr_MMD_unb}]
 We first show that, under $(A1)$--$(A3)$, Assumptions 1~and~2 made in  \citet{bernton2019approximate} are satisfied under \textsc{mmd} when $f_n(\bar{\varepsilon}_n)=1/(n\bar{\varepsilon}^2_n)$ and~$c(\theta)=\mathbb{E}_{z}\left[ k(z,z) \right] $, with $z \sim \mu_{\theta}$. 
 
Consistent with the above goal, first recall that, by standard properties of \textsc{mmd}, 
\begin{equation}
	\label{eq3_proof}
	\mathcal{D}^2_{\textsc{mmd}}(\mu_1, \mu_2)=\mathbb{E}_{x_1,x_1'}\left[ k(x_1,x_1') \right] - 2 \mathbb{E}_{x_1,x_2}\left[ k(x_1,x_2) \right] + \mathbb{E}_{x_2,x_2'}\left[ k(x_2,x_2') \right],
\end{equation}
with $x_1,x_1' \sim \mu_1$ and $x_2,x_2' \sim \mu_2$, all independently; see e.g., \citet{cherief2022finite}. Since $k(x,x')=\left<\phi(x),\phi(x')\right>_{\mathcal{H}}$ \citep[see e.g.,][]{muandet2017kernel}, the above~result~implies that
\begin{equation}
	\label{eq4_proof}
	\begin{split}
		&\mathcal{D}^2_{\textsc{mmd}}(\mu_1, \mu_2)\\
		&=\mathbb{E}_{x_1,x_1'}\left[\left<\phi(x_1),\phi(x_1')\right>_{\mathcal{H}} \right] - 2 \mathbb{E}_{x_1,x_2}\left[ \left<\phi(x_1),\phi(x_2)\right>_{\mathcal{H}} \right] + \mathbb{E}_{x_2,x_2'}\left[\left<\phi(x_2),\phi(x_2')\right>_{\mathcal{H}} \right]\\
		&=||\mathbb{E}_{x_1}[\phi(x_1)]||_{\mathcal{H}}^2-2\left<\mathbb{E}_{x_1}[\phi(x_1)],\mathbb{E}_{x_2}[\phi(x_2)]\right>_{\mathcal{H}}+||\mathbb{E}_{x_2}[\phi(x_2)]||_{\mathcal{H}}^2 \\
		&=\left\|\mathbb{E}_{x_1}\left[\phi(x_1)\right]-\mathbb{E}_{x_2}\left[\phi(x_2)\right]\right\|^2_{\mathcal{H}}.
	\end{split}
\end{equation}
Leveraging Equations \eqref{eq3_proof}--\eqref{eq4_proof} and basic Markov inequalities, for any $\bar{\varepsilon}_n 	\geq 0$, it holds 
\begin{align*}
	&\mathbb{P}_{y_{1:n}}  \left[
	\mathcal{D}_{\textsc{mmd}}(\hat{\mu}_{y_{1:n}},\mu^*) > \bar{\varepsilon}_n
	\right]
	\leq (1/\bar{\varepsilon}_n^2) \mathbb{E}_{y_{1:n}}\left[\mathcal{D}^2_{\textsc{mmd}}(\hat{\mu}_{y_{1:n}},\mu^*)\right] \\
	& =  (1/\bar{\varepsilon}_n^2)  \mathbb{E}_{y_{1:n}}[  \left\|(1/n){\textstyle{\sum\nolimits_{i=1}^n}} \phi(y_i)-\mathbb{E}_{y}\left[\phi(y)\right]\right\|_{\mathcal{H}}^2 ] 
		\leq [1/(n^2\bar{\varepsilon}_n^2)]  {\textstyle{\sum\nolimits_{i=1}^n}} \mathbb{E}_{y_i}[  \left\|\phi(y_i)\right\|_{\mathcal{H}}^2 ]
	\\
	& \leq [1/(n\bar{\varepsilon}_n^2)]  \mathbb{E}_{y_1}[  \left\|\phi(y_1)\right\|_{\mathcal{H}}^2 ]
	=  [1/(n\bar{\varepsilon}_n^2)] \mathbb{E}_{y_1}\left[ k(y_1,y_1) \right]= [1/(n\bar{\varepsilon}_n^2)] \mathbb{E}_{y}\left[ k(y,y) \right],
\end{align*}
with $y \sim \mu^*$. Since $[1/(n\bar{\varepsilon}_n^2)] \mathbb{E}_{y}\left[ k(y,y) \right] \rightarrow 0$ as $n\rightarrow\infty$ by condition $(A1)$, we have that  $\mathcal{D}_{\textsc{mmd}}(\hat{\mu}_{y_{1:n}},\mu^*) \to 0$ in $\mathbb{P}_{y_{1:n}}$--probability as $n \to \infty$, thus meeting Assumption 1 in \citet{bernton2019approximate}. Moreover, as a direct consequence of the above derivations,
$$
\mathbb{P}_{z_{1:n}}\left[
\mathcal{D}_{\textsc{mmd}}(\hat{\mu}_{z_{1:n}},\mu_{\theta}) > \bar{\varepsilon}_n
\right] \leq  [1/(n\bar{\varepsilon}_n^2)] \mathbb{E}_{z}\left[ k(z,z) \right].
$$
Thus, setting $1/(n\bar{\varepsilon}^2_n)=f_n(\bar{\varepsilon}_n)$ and $\mathbb{E}_{z}\left[ k(z,z) \right]=c(\theta) $, with $z \sim \mu_{\theta}$, ensures that $$ \mathbb{P}_{z_{1:n}}\left[
\mathcal{D}_{\textsc{mmd}}(\hat{\mu}_{z_{1:n}},\mu_{\theta}] > \bar{\varepsilon}_n
\right] \leq c(\theta)f_n(\bar{\varepsilon}_n), $$ with $f_n(u)=1/(nu^2)$ strictly decreasing in $u$ for any fixed $n$, and $f_n(u) \rightarrow 0$ as $n \rightarrow \infty$, for fixed $u$. Moreover, by Assumptions $(A2)$--$(A3)$, $c(\theta) =\mathbb{E}_{z}\left[ k(z,z) \right]$ is $\pi$-integrable and there exist a $\delta_0>0$ and a $c_0>0$ such that $c(\theta)<c_0$ for any $\theta$ satisfying $ (\mathbb{E}_{z,z'}\left[ k(z,z') \right] - 2 \mathbb{E}_{z,y}\left[ k(y,z) \right] + \mathbb{E}_{y,y'}\left[ k(y,y') \right])^{1/2}= \mathcal{D}_{\textsc{mmd}}(\mu_\theta,\mu^*) \leq \varepsilon^* + \delta_0$. This ensures that Assumption~2 in \citet{bernton2019approximate}~holds. 

Finally, note that Assumption 3 in \citet{bernton2019approximate}, is verified by our Assumption (\ref{C2}). Therefore, under the assumptions in  Proposition~\ref{corr_MMD_unb} it is possible to apply Proposition~3 in \citet{bernton2019approximate} with $f_n(\bar{\varepsilon}_n)=1/(n\bar{\varepsilon}^2_n)$, $c(\theta)=\mathbb{E}_{z}\left[ k(z,z) \right] $ and $R=M_n$, which yields the concentration result in Proposition~\ref{corr_MMD_unb}.
\end{proof}

\begin{proof}[Proof of Proposition \ref{corr_Wasserstein_unb}]
Assumptions (A1') and (A2') imply that the norm $\|  \|x\| \|_{\psi_1} $ (see Equation  (14) in~\cite{lei2020w} for a definition) is uniformly bounded when $x\sim \mu_\theta$ and $x\sim \mu^*$. Thus, we can use (15) in~\cite{lei2020w} to obtain
$$ \mathbb{P}_{z_{1:n}}\left( \mathcal{D}_{\textrm{wass}}(\mu_\theta,\hat{\mu}_{z_{1:n}}) > u \right) \leq \exp[-c' n (u-c_1 n^{-1/\max(d,3)})_+^2 ] =: f_n(u) $$
and, similarly,
$ \mathbb{P}_{y_{1:n}}\left( \mathcal{D}_{\textrm{wass}}(\mu^*,\hat{\mu}_{y_{1:n}} ) > u \right) \leq f_n(u)$,
where $c'$ depends only on the constant~$c$ in (A1') and (A2'), and $c_1$ is a universal constant. Notice that (15) in~\cite{lei2020w} requires $d>2$. If $d=1$, we can define $x'=(x,0,0)\in\mathbb{R}^3$ and apply the result in $\mathbb{R}^3$ (we can proceed similarly if $d=2$). This is why Proposition \ref{corr_Wasserstein_unb} is stated with $\max(d,3)$.
The above bounds ensure that Assumptions 1--2 in~\citet{bernton2019approximate} are met. Moreover, condition (\ref{C2})~verifies Assumption 3 in~\citet{bernton2019approximate}. Thus, we can apply Proposition 3 of~\citet{bernton2019approximate} with $f_n(\cdot)$ defined as above and under the vanishing conditions on $\bar{\varepsilon}_n$ in Proposition \ref{corr_Wasserstein_unb}.
This implies that when $n \rightarrow \infty$, and $\bar{\varepsilon}_n \rightarrow 0$ such that $f_n(\bar{\varepsilon}_n)\rightarrow 0$,  for some~${C \in (0, \infty)}$ and any $M_n \in (0, \infty)$, with $\mathbb{P}_{y_{1:n}}$--probability going to $1$ as $n \rightarrow \infty$,~it~holds
	\begin{eqnarray*}
	\pi_n^{(\varepsilon^*+\bar{\varepsilon}_n)}(\{\theta:  \mathcal{D}_{\textrm{wass}}(\mu_\theta,\mu^*) > \varepsilon^* + 4\bar{\varepsilon}_n/3 + f_n^{-1}\left( \bar{\varepsilon}_n^L/M_n \right) \}) \leq C/M_n.
	\end{eqnarray*}
Recall that our restriction $n^{-1/\max(d,3)} \ll \bar{\varepsilon}_n $, together with $n\bar{\varepsilon}^2_n \to \infty $, imply that for $n$ large enough $f_n(\cdot)$ is invertible and  $f_n(\bar{\varepsilon}_n)\rightarrow 0$. On such a range, we have that
$ f_n^{-1}\left( \bar{\varepsilon}_n^L/M_n\right) = [(1/(c'n)) \log(M_n/\bar{\varepsilon}_n^L) ]^{1/2} + c_1 n^{-{1}/{\max(d,3)}},$
which concludes the proof.
\end{proof}


\begin{proof}[Proof of Lemma~\ref{lemma_rademacher_dependent_full}]
Let $(x_t)_{t\in\mathbb{Z}}$ be a stochastic process with $\beta$-mixing coefficients $\beta(k)$, for $k\in\mathbb{N}$. Moreover, denote with $\mu^{(n)}$ and $\mu$ the joint distribution of a sample $x_{1:n}$ from  $(x_t)_{t\in\mathbb{Z}}$ and its constant marginal $\mu=\mu^{(1)}$, respectively. Then, by combining Proposition 2 and Lemma 2 in~\cite{Mohri2008} under our notation, we have that for any $b$-uniformly bounded class $\mathfrak{F}$, any integer $n \geq 1$ and any scalar $\delta \geq 0$, the  inequality
 \begin{equation*}
 \mathbb{P}_{x_{1:n}} \left[\mathcal{D}_{\mathfrak{F}}(\hat{\mu}_{x_{1:n}},\mu)  >  2 \mathfrak{R}_{\mu,n/(2K)}(\mathfrak{F})+\delta \right] \\
\leq 2\exp(-n \delta^2/ K b^2) +2(n/2K - 1) \beta(K),
\end{equation*}
holds for every integer $K>0$ such that $n/(2K)\in\mathbb{N}$, where $ \mathfrak{R}_{\mu,n/(2K)}$ is the Rademacher complexity based an i.i.d. sample of size $n/(2K)$ from $\mu$; see Definition~\ref{def_rade}. The above~concentration inequality also implies
 \begin{equation}\label{eq1-dep}
 \begin{split}
&\mathbb{P}_{x_{1:n}} \left[\mathcal{D}_{\mathfrak{F}}(\hat{\mu}_{x_{1:n}},\mu)  \leq  2 \mathfrak{R}_{\mu,n/(2K)}(\mathfrak{F})+\delta \right] \\
& \qquad \qquad \qquad \qquad \qquad  \geq 1-2\exp(-n \delta^2/ K b^2) -2(n/2K - 1) \beta(K)\\
& \qquad \qquad \qquad \qquad \qquad \qquad  \qquad \geq 1-2\exp[-n \delta^2/ (4K b^2)] -2(n/2K) \beta(K).
\end{split}
\end{equation}
Notice that, in order to ensure that both $2\exp(-n \delta^2/ 4K b^2)$ and $2(n/2K) \beta(K)$ vanish to zero (under Assumption (\ref{C1dep}) for $\beta(K)$), it is tempting to apply \eqref{eq1-dep} with $K = n^{\alpha} $~for~some $0<\alpha < 1$. Unfortunately,~there is no reason for such a $K$ to be an integer. A solution, would be to let $K = \lfloor n^\alpha \rfloor$, but in this case~$n/(2K)$ might not be an integer. To address such issues, it is necessary~to~consider a careful modification of \eqref{eq1-dep}. To this end write the Euclidean division $n=n' + r$ where $n' =2Kh \leq n$, $h = \lfloor n/(2K) \rfloor $ and $0\leq r<2K$. Then, under the common marginal~assumption and recalling also Proposition 1 in~\cite{Mohri2008} together with the triangle~inequality and the fact that the functions $f$ within $\mathfrak{F}$ are $b$-uniformly bounded, we have that
\begin{align*}
 \mathcal{D}_{\mathfrak{F}}(\hat{\mu}_{x_{1:n}},\mu)
 & = \sup_{f\in\mathfrak{F}} \left| \frac{1}{n} \sum\nolimits_{i=1}^n [f(x_i) - \mathbb{E}_{\mu} f(x)] \right|
 \\
 & \leq \sup_{f\in\mathfrak{F}} \left| \frac{1}{n} \sum\nolimits_{i=1}^{n'} [f(x_i) - \mathbb{E}_{\mu} f(x)] \right| +  \frac{1}{n} \sup_{f\in\mathfrak{F}} \left| \sum\nolimits_{i=n'+1}^{n'+r} [f(x_i) - \mathbb{E}_{\mu} f(x)] \right|
 \\
 & \leq \sup_{f\in\mathfrak{F}} \left| \frac{1}{n'} \sum\nolimits_{i=1}^{n'} [f(x_i) - \mathbb{E}_{\mu} f(x)] \right| +  \frac{2br}{n} \leq  \mathcal{D}_{\mathfrak{F}}(\hat{\mu}_{x_{1:n'}},\mu) + \frac{4bK}{n},
\end{align*}
where the last inequality follows directly from the definition of $ \mathcal{D}_{\mathfrak{F}}(\hat{\mu}_{x_{1:n'}},\mu) $ together~with the fact that  $0\leq r<2K$. Applying now \eqref{eq1-dep} to  $ \mathcal{D}_{\mathfrak{F}}(\hat{\mu}_{x_{1:n'}},\mu)$ yields
\vspace{-3pt}
 \begin{equation*}
 \begin{split}
&\mathbb{P}_{x_{1:n}} \left[\mathcal{D}_{\mathfrak{F}}(\hat{\mu}_{x_{1:n'}},\mu)+ 4bK/n  \leq  2 \mathfrak{R}_{\mu,h}(\mathfrak{F})+\delta + 4bK/n\right] \geq 1-2\exp(-h \delta^2/ 2b^2) -2h \beta(K).
\end{split}
\end{equation*}
Therefore, since $\mathcal{D}_{\mathfrak{F}}(\hat{\mu}_{x_{1:n}},\mu) \leq  \mathcal{D}_{\mathfrak{F}}(\hat{\mu}_{x_{1:n'}},\mu) + 4bK/n$ we also have that
\vspace{-3pt}
 \begin{equation*}
 \begin{split}
&\mathbb{P}_{x_{1:n}} \left[\mathcal{D}_{\mathfrak{F}}(\hat{\mu}_{x_{1:n}},\mu)   \leq  2 \mathfrak{R}_{\mu,h}(\mathfrak{F})+\delta + 4bK/n\right] \geq 1-2\exp(-h \delta^2/ 2b^2) -2h \beta(K).
\end{split}
\end{equation*}
To conclude the proof, notice that to prove convergence and concentration of the \textsc{abc} posterior it will be sufficient to let $K=\lfloor \sqrt{n} \rfloor$. Therefore, by replacing $K=\lfloor \sqrt{n} \rfloor$ in the above inequality and within the expression for $h= \lfloor n/(2K) \rfloor $ we have
\vspace{-3pt}
 \begin{multline*}
\mathbb{P}_{x_{1:n}} \Bigl[\mathcal{D}_{\mathfrak{F}}(\hat{\mu}_{x_{1:n}},\mu) \leq  2 \mathfrak{R}_{\mu,s_n}(\mathfrak{F}) + \frac{4b}{\sqrt{n}} +\delta \Bigr]\geq 1-2\exp(-s_n \delta^2/ 2b^2) -2s_n \beta(\lfloor \sqrt{n} \rfloor)
\end{multline*}
where $s_n= \lfloor n/(2\lfloor \sqrt{n} \rfloor) \rfloor $ and the term $4b/\sqrt{n}$ follows directly from the fact that~$\lfloor \sqrt{n}\rfloor /n \leq \sqrt{n}/n =1/\sqrt{n}$.
\end{proof}

\begin{proof}[Proof of Lemma~\ref{lemma_ar}]
The proof follows the arguments used in Chapter 2 of~\citet{doukhan} to study general Markov chains. In particular, when $(x_t)_{t \in \mathbb{Z}}$ is a stationary Markov chain with invariant distribution $\pi(\cdot)$ and transition kernel $P(\cdot,\cdot)$, a result proven in~\citet{davydov} and recalled in page 87--88 of \citet{doukhan} gives
$ \beta(k) = \mathbb{E}_{x\sim\pi} \| P^k(x,\cdot) - \pi(\cdot) \|_{\textsc{TV}} .$
When $-1<\theta<1$, standard results for the \textsc{ar}(1) model in Lemma~\ref{lemma_ar} lead to the invariant distribution $\pi = \mbox{N}( \psi/(1-\theta), \sigma^2/(1-\theta^2) )$. 

As for $P^k(x,\cdot)$, notice that, under such an \textsc{ar}(1) model with starting point $x$, we can write
\vspace{-19pt}
\begin{align*}
 x_k & = \theta x_{k-1} + \psi + \varepsilon_k =\theta ( \theta x_{k-2} +  \psi + \varepsilon_{k-1} ) +  \psi + \varepsilon_k  = \dots \\
     & = \theta^k x +  \psi \sum\nolimits_{l=0}^{k-1}\theta^l + \sum\nolimits_{l=0}^{k-1} \theta^l \varepsilon_{k-l}.
\end{align*}
Therefore, $ P^k(x,\cdot) = \mbox{N}( \theta^k x +  \psi \sum_{l=0}^{k-1}\theta^l , \sigma^2 \sum_{l=0}^{k-1}\theta^{2l}  ) $.  

Moreover, notice that, by direct application of standard properties of finite power series, we have  \smash{$\sum_{l=0}^{k-1}\theta^l=(1-\theta^k)/(1-\theta)$} and \smash{$\sum_{l=0}^{k-1}\theta^{2l}=\smash{(1-\theta^{2k})/(1-\theta^2)}$}. Since our goal is to derive an upper bound for $\beta(k)$~and provided that the \textsc{kl} divergence among Gaussian densities is available in closed form, let us first consider the Pinsker's inequality
$$ \| P^k(x,\cdot) - \pi(\cdot) \|_{\textsc{TV}} \leq [\mathcal{D}_{\textrm{KL}}( P^k(x,\cdot),\pi(\cdot) )/2 ]^{1/2} $$
where $\mathcal{D}_{\textrm{KL}}$ stands for the \textsc{kl} divergence. Since both $P^k(x,\cdot)$ and $\pi(\cdot)$ are Gaussian, then 
\begin{align*}
&\mathcal{D}_{\textrm{KL}}( P^k(x,\cdot),\pi(\cdot) )
\\
& = \frac{1}{2} \left[ \frac{\sigma^2 (1-\theta^{2k})   }{\sigma^2 } - 1 + \frac{[ \theta^k x +  \psi(1-\theta^k)/(1-\theta)  - \psi/(1-\theta)]^2}{ \sigma^2/(1-\theta^2) }  + \log \frac{  \sigma^2 }{\sigma^2(1-\theta^{2k}) } \right]
\\
& =  \frac{1}{2} \left[ (1-\theta^{2k}) - 1 + \frac{\theta^{2k}[  x- \psi/(1-\theta)   ]^2}{ \sigma^2/(1-\theta^2) }  + \log \frac{1}{1-\theta^{2k} } \right]
\\
& =  \frac{1}{2} \left[ - \theta^{2k} + \frac{\theta^{2k}[  x- \psi/(1-\theta)   ]^2}{ \sigma^2/(1-\theta^2) }  +\log \left( 1+ \frac{\theta^{2k}}{1-\theta^{2k} } \right) \right]
\\
& \leq \frac{1}{2} \left[- \theta^{2k} + \frac{\theta^{2k}[  x- \psi/(1-\theta)   ]^2}{ \sigma^2/(1-\theta^2) }  +  \frac{\theta^{2k}}{1-\theta^{2k} } \right] \leq \frac{\theta^{2k}}{2} \left[ -1 + \frac{[  x- \psi/(1-\theta)   ]^2}{ \sigma^2/(1-\theta^2) } +  \frac{1}{1-\theta^{2} } \right].
\end{align*}
Therefore, by leveraging the above result, together with standard properties of the expectation, we have
\begin{align*}
 \beta(k)
 & \leq \mathbb{E}_{x\sim\pi} [\mathcal{D}_{\textrm{KL}}( P^k(x,\cdot),\pi(\cdot) )/2 ]^{1/2} 
\leq [\mathbb{E}_{x\sim\pi} \mathcal{D}_{\textrm{KL}}( P^k(x,\cdot),\pi(\cdot) )/2 ]^{1/2} 
  \\
  &\leq \left[ \frac{\theta^{2k}}{4} \left[ -1 + \frac{\mathbb{E}_{x\sim\pi}[  x- \psi/(1-\theta)   ]^2}{ \sigma^2/(1-\theta^2) } +  \frac{1}{1-\theta^{2} } \right] \right]^{1/2}  = \sqrt{  \frac{\theta^{2k}}{4(1-\theta^2)}} = \frac{|\theta|^k}{2\sqrt{1-\theta^2}},
\end{align*}
which concludes the proof.
\end{proof}

\begin{proof}[Proof of Proposition~\ref{conv_rad_dep}]
Under Assumption (\ref{C1dep}), for any fixed $\delta>0$,  we have~that $\sum\nolimits_{n>0} [ 2\exp(-s_n \delta^2/2b^2) + 2 s_nC_{\beta}  \exp( - \gamma \lfloor \sqrt{n} \rfloor^{\xi})  ] < \infty.$ Therefore, combining Lemma~\ref{lemma_rademacher_dependent_full} with Assumption  (\ref{C4}), both $\mathcal{D}_{\mathfrak{F}}(\hat{\mu}_{z_{1:n}},\mu_{\theta})$ and $\mathcal{D}_{\mathfrak{F}}(\hat{\mu}_{y_{1:n}},\mu^*)$  converge to 0 almost surely as $n \to \infty$, by the Borel--Cantelli Lemma. As a result, since  $$\scalemath{0.95}{ -\mathcal{D}_\mathfrak{F} (\hat{\mu}_{z_{1:n}},\mu_\theta) -\mathcal{D}_\mathfrak{F}(\hat{\mu}_{y_{1:n}},\mu^*) \leq \mathcal{D}_{\mathfrak{F}}(\mu_\theta,\mu^*) - \mathcal{D}_{\mathfrak{F}}(\hat{\mu}_{y_{1:n}},\hat{\mu}_{z_{1:n}}) \leq  \mathcal{D}_\mathfrak{F} (\hat{\mu}_{z_{1:n}},\mu_\theta)+  \mathcal{D}_\mathfrak{F}(\hat{\mu}_{y_{1:n}},\mu^*)},$$
it holds that $\mathcal{D}_\mathfrak{F} (\hat{\mu}_{z_{1:n}},\hat{\mu}_{y_{1:n}}) \rightarrow\mathcal{D}_\mathfrak{F} (\mu_\theta, \mu^*) $ almost surely as $n \rightarrow \infty$. To conclude it suffices to apply  again the proof of Theorem 1 in \citet{jiang2018approximate}; see also the proof of Corollary~\ref{conv_rad}.
\end{proof}

\begin{proof}[Proof of Theorem~\ref{thm_rademacher_dep}]
To prove Theorem~\ref{thm_rademacher_dep} we will follow the same line of reasoning~as in the proof of Theorem~\ref{thm_rademacher}. However, in this case we leverage Lemma~\ref{lemma_rademacher_dependent_full} instead of Lemma~\ref{lemma_rademacher}. To this end, letting $\delta=c_1-2 \mathfrak{R}_{s_n}(\mathfrak{F}) - 4b/\sqrt{n}$ with $c_1 \geq 2 \mathfrak{R}_{s_n}(\mathfrak{F}) + 4b/\sqrt{n}$ and $ \mathfrak{R}_{s_n} =\sup_{\mu \in \mathcal{P}(\mathcal{Y})} \mathfrak{R}_{\mu,s_n} $, we obtain, under Assumption (\ref{C1dep}), Equation \eqref{eq1_proof_1dep}  below, instead of \eqref{eq1_proof_1}. 
\begin{multline}
	\label{eq1_proof_1dep}
	\mathbb{P}_{x_{1:n}} \left[
	\mathcal{D}_{\mathfrak{F}}(\hat{\mu}_{x_{1:n}},\mu) \leq  c_1 \right] \\ \geq 1-2\exp[- s_n (c_1-2 \mathfrak{R}_{s_n}(\mathfrak{F}) - 4b/\sqrt{n} )^2/ 2b^2]
-2 s_nC_{\beta} \exp( - \gamma \lfloor \sqrt{n} \rfloor^{\xi})  .
\end{multline}
As in Theorem~\ref{thm_rademacher}, let $c_1=\bar{\varepsilon}_n /3$ and notice that, by the settings of Theorem~\ref{thm_rademacher_dep}, for $n$ large enough $\bar{\varepsilon}_n /3>2 \mathfrak{R}_{s_n}(\mathfrak{F}) + 4b/\sqrt{n}$. Therefore, applying \eqref{eq1_proof_1dep} to $y_{1:n}$, with $c_1=\bar{\varepsilon}_n /3$, leads to the following upper bound
\begin{multline*}
	\mathbb{P}_{y_{1:n}}\left[ \mathcal{D}_{\mathfrak{F}}(\hat{\mu}_{y_{1:n}},\mu^*) \leq
	\bar{\varepsilon}_n/3 \right] \\ \geq 1-2\exp[- s_n (\bar{\varepsilon}_n /3-2 \mathfrak{R}_{s_n}(\mathfrak{F}) - 4b/\sqrt{n} )^2/ 2b^2]
-2 s_nC_{\beta} \exp( - \gamma \lfloor \sqrt{n} \rfloor^{\xi}).
\end{multline*}
Recall that, under the settings of Theorem~\ref{thm_rademacher_dep}, we have that $\sqrt{n} \bar{\varepsilon}^2_n \to \infty$ and $\bar{\varepsilon}_n/\mathfrak{R}_{s_n}(\mathfrak{F}) \to \infty$ and, therefore,
$$s_n (\bar{\varepsilon}_n /3-2 \mathfrak{R}_{s_n}(\mathfrak{F}) - 4b/\sqrt{n} )^2 \sim s_n \bar{\varepsilon}^2_n /9 \sim\sqrt{n} \bar{\varepsilon}^2_n  \to\infty.$$ Combining~this result with the fact that  $2 s_nC_{\beta} \exp( - \gamma \lfloor \sqrt{n} \rfloor^{\xi}) \to 0$ as $n \to \infty$, it follows that the lower bound for \smash{$\mathbb{P}_{y_{1:n}}\left[ \mathcal{D}_{\mathfrak{F}}(\hat{\mu}_{y_{1:n}},\mu^*) \leq \bar{\varepsilon}_n/3 \right]$} goes to $1$ as $n \to \infty$. Hence,~in~the~remaining part of the proof, we will restrict to the event
$ \{ \mathcal{D}_{\mathfrak{F}}(\hat{\mu}_{y_{1:n}},\mu^*) \leq \bar{\varepsilon}_n/3 \} $. 

Let $\mathbb{P}_{\theta, z_{1:n} }$ corresponds to the joint distribution of $\theta \sim \pi$ and $ z_{1:n}$ from \smash{$\mu_\theta^{(n)}$}. Then, as a direct consequence of the definition of conditional probability, for every positive $c_2$, including $c_2>2 \mathfrak{R}_{s_n}(\mathfrak{F}) + 4b/\sqrt{n}$, it follows that
\begin{equation}
	\label{eq1_proof_rate_dep}
	\begin{split}
		& \pi_n^{(\varepsilon^*+\bar{\varepsilon}_n)}\left(\{\theta:  \mathcal{D}_{\mathfrak{F}}(\mu_\theta,\mu^*) > \varepsilon^* + 4\bar{\varepsilon}_n/3 + c_2\} \right)\\
		& \qquad = \frac{\mathbb{P}_{\theta, z_{1:n} } \left[ \mathcal{D}_{\mathfrak{F}}(\mu_\theta,\mu^*) >  \varepsilon^* + 4\bar{\varepsilon}_n/3 + c_2 ,   \mathcal{D}_{\mathfrak{F}}(\hat{\mu}_{z_{1:n}},\hat{\mu}_{y_{1:n}}) \leq \varepsilon^*+\bar{\varepsilon}_n \right]}
		{\mathbb{P}_{\theta, z_{1:n} }  \left[  \mathcal{D}_{\mathfrak{F}}(\hat{\mu}_{z_{1:n}},\hat{\mu}_{y_{1:n}}) \leq \varepsilon^*+\bar{\varepsilon}_n \right]}.
	\end{split}
\end{equation}
To upper bound the ratio in \eqref{eq1_proof_rate_dep}, let us first derive an upper bound for~the numerator.~To this end, consider the triangle inequality $\mathcal{D}_\mathfrak{F} (\mu_\theta, \mu^*) \leq  \mathcal{D}_\mathfrak{F} (\hat{\mu}_{z_{1:n}},\mu_\theta) +\mathcal{D}_\mathfrak{F} (\hat{\mu}_{z_{1:n}},\hat{\mu}_{y_{1:n}})+  \mathcal{D}_\mathfrak{F}(\hat{\mu}_{y_{1:n}},\mu^*)$  (recall that $ \mathcal{D}_\mathfrak{F} $ is a semimetric), along with the previously-proved result that the event $ \{ \mathcal{D}_{\mathfrak{F}}(\hat{\mu}_{y_{1:n}},\mu^*) \leq \bar{\varepsilon}_n/3 \} $ has $\mathbb{P}_{y_{1:n}}$--probability going to $1$. Hence,~for~$n$~large~enough we have 
\begin{align*}
	&\mathbb{P}_{\theta, z_{1:n} }  [ \mathcal{D}_{\mathfrak{F}}(\mu_\theta,\mu^*)  > \varepsilon^* + 4\bar{\varepsilon}_n/3 + c_2 , \mathcal{D}_{\mathfrak{F}}(\hat{\mu}_{z_{1:n}},\hat{\mu}_{y_{1:n}}) \leq \varepsilon^*+\bar{\varepsilon}_n ]
	\\
	& \qquad \leq  \mathbb{P}_{\theta, z_{1:n} } [   \mathcal{D}_\mathfrak{F} (\hat{\mu}_{z_{1:n}},\hat{\mu}_{y_{1:n}})+ \mathcal{D}_\mathfrak{F} (\hat{\mu}_{z_{1:n}},\mu_\theta) + \mathcal{D}_\mathfrak{F}(\hat{\mu}_{y_{1:n}},\mu^*) > \varepsilon^* + 4\bar{\varepsilon}_n/3 + c_2 , \\
	& \qquad  \qquad  \qquad  \mathcal{D}_{\mathfrak{F}}(\hat{\mu}_{z_{1:n}},\hat{\mu}_{y_{1:n}}) \leq \varepsilon^* + \bar{\varepsilon}_n ]
	\\
	&\qquad \leq  \mathbb{P}_{\theta, z_{1:n} } \left[ \mathcal{D}_\mathfrak{F} (\hat{\mu}_{z_{1:n}},\mu_\theta) + \mathcal{D}_\mathfrak{F}(\hat{\mu}_{y_{1:n}},\mu^*) >  \bar{\varepsilon}_n/3 + c_2 \right]  \leq   \mathbb{P}_{\theta, z_{1:n} }\left[  \mathcal{D}_\mathfrak{F} (\hat{\mu}_{z_{1:n}},\mu_\theta)> c_2 \right],
\end{align*}
where  $\mathbb{P}_{\theta, z_{1:n} }\left[  \mathcal{D}_\mathfrak{F} (\hat{\mu}_{z_{1:n}},\mu_\theta)> c_2 \right]=\int_{\theta \in \Theta}\mathbb{P}_{z_{1:n} }\left[  \mathcal{D}_\mathfrak{F} (\hat{\mu}_{z_{1:n}},\mu_\theta)> c_2 \mid \theta \right] \pi({\rm d}\theta).$ 

Therefore, leveraging the above result and applying \eqref{eq1_proof_1dep} to $z_{1:n}$ yields,
\begin{equation*}
	\begin{split}
		{\int_{\theta \in \Theta}}&{\mathbb{P}_{z_{1:n} }}[ \mathcal{D}_\mathfrak{F} (\hat{\mu}_{z_{1:n}},\mu_\theta)> c_2 \mid \theta] \pi({\rm d}\theta) = \int_{\theta \in \Theta}(1-{\mathbb{P}_{z_{1:n} }}[  \mathcal{D}_\mathfrak{F} (\hat{\mu}_{z_{1:n}},\mu_\theta) \leq c_2 \mid \theta]) \pi({\rm d}\theta) \\
		&\qquad \qquad \leq 2\exp[- s_n (c_2-2 \mathfrak{R}_{s_n}(\mathfrak{F}) - 4b/\sqrt{n} )^2/ 2b^2]+2 s_nC_{\beta} \exp( - \gamma \lfloor \sqrt{n} \rfloor^{\xi}).
	\end{split}
\end{equation*}
This controls the numerator in \eqref{eq1_proof_rate_dep}.
As for the denominator,
defining the event $\mbox{E}_n := \{ \theta \in \Theta : \mathcal{D}_{\mathfrak{F}}(\mu_\theta,\mu^*) \leq  \varepsilon^*+\bar{\varepsilon}_n/3 \} $ and applying again the triangle inequality,
we have that
\begin{align*}
	&\mathbb{P}_{\theta, z_{1:n} }   [\mathcal{D}_{\mathfrak{F}}(\hat{\mu}_{z_{1:n}},\hat{\mu}_{y_{1:n}}) \leq \varepsilon^*{+}\ \bar{\varepsilon}_n]
	\geq \int_{{\footnotesize \mbox{E}}_n} \mathbb{P}_{z_{1:n}}\left[\mathcal{D}_{\mathfrak{F}}(\hat{\mu}_{z_{1:n}},\hat{\mu}_{y_{1:n}}) \leq \varepsilon^*{+} \ \bar{\varepsilon}_n \mid \theta \right] \pi({\rm d}\theta)
	\\
	& \qquad \geq \int_{{\footnotesize \mbox{E}}_n}  \mathbb{P}_{z_{1:n}}\left[ \mathcal{D}_{\mathfrak{F}}(\hat{\mu}_{y_{1:n}},\mu^*) + \mathcal{D}_{\mathfrak{F}}(\mu_\theta,\mu^*) +  \mathcal{D}_{\mathfrak{F}}(\hat{\mu}_{z_{1:n}},\mu_\theta)
	\leq \varepsilon^*+\bar{\varepsilon}_n \mid \theta \right] \pi({\rm d}\theta)
	\\
	& \qquad \geq \int_{{\footnotesize \mbox{E}}_n}   \mathbb{P}_{z_{1:n}}\left[ \mathcal{D}_{\mathfrak{F}}(\hat{\mu}_{z_{1:n}},\mu_\theta)
	\leq \bar{\varepsilon}_n/3 \mid \theta \right] \pi({\rm d}\theta).
\end{align*}
The last inequality follows from that fact that we can restrict to the event $ \{ \mathcal{D}_{\mathfrak{F}}(\hat{\mu}_{y_{1:n}},\mu^*) \leq \bar{\varepsilon}_n/3 \} $, and that we are integrating over $\mbox{E}_n := \{ \theta \in \Theta : \mathcal{D}_{\mathfrak{F}}(\mu_\theta,\mu^*) \leq \varepsilon^*+ \bar{\varepsilon}_n/3 \} $. Let us now apply again~\eqref{eq1_proof_1dep} to $z_{1:n}$, with $c_1=\bar{\varepsilon}_n /3$, to further lower bound  the last term of the above inequality~by
\begin{align*}
&{\int_{{\footnotesize \mbox{E}}_n}} [ 1-2\exp[- s_n (\bar{\varepsilon}_n /3-2 \mathfrak{R}_{s_n}(\mathfrak{F}) - 4b/\sqrt{n} )^2/ 2b^2] -2 s_nC_{\beta} \exp( - \gamma \lfloor \sqrt{n} \rfloor^{\xi})  ] \pi({\rm d}\theta)\\
&  = \pi( \mbox{E}_n) [ 1-2\exp[- s_n (\bar{\varepsilon}_n /3-2 \mathfrak{R}_{s_n}(\mathfrak{F}) - 4b/\sqrt{n} )^2/ 2b^2] -2 s_nC_{\beta} \exp( - \gamma \lfloor \sqrt{n} \rfloor^{\xi})  ].
\end{align*}
Note that, by Assumption (\ref{C2}), $\pi( \mbox{E}_n) \geq c_\pi (\bar{\varepsilon}_n/3)^L$. Moreover, as shown before, the quantity $  1-2\exp[- s_n (\bar{\varepsilon}_n /3-2 \mathfrak{R}_{s_n}(\mathfrak{F}) - 4b/\sqrt{n} )^2/ 2b^2] -2 s_nC_{\beta} \exp( - \gamma \lfloor \sqrt{n} \rfloor^{\xi})$ goes to~$1$~when $ n \to \infty $, which also implies that, for a large enough $n$, $$ 1-2\exp[- s_n (\bar{\varepsilon}_n /3-2 \mathfrak{R}_{s_n}(\mathfrak{F}) - 4b/\sqrt{n} )^2/ 2b^2] -2 s_nC_{\beta} \exp( - \gamma \lfloor \sqrt{n} \rfloor^{\xi})  > 1/2.$$ Therefore, leveraging both results, the denominator in  \eqref{eq1_proof_rate_dep} can be lower bounded by the term  \smash{$(c_\pi/2) \left(\bar{\varepsilon}_n/{3}\right)^L$}.

To proceed with the proof, let us combine the upper and lower bounds derived, respectively, for the numerator and the denominator of the ratio in \eqref{eq1_proof_rate_dep}. This yields,~for~any~integer $K$,
\begin{multline}
	\label{eq2_proof_rate_dep}
	\pi_n^{(\varepsilon^*+\bar{\varepsilon}_n)}(\{\theta:  \mathcal{D}_{\mathfrak{F}}(\mu_\theta,\mu^*) > \varepsilon^* +  4\bar{\varepsilon}_n/3 + c_2\})
	\\
	\leq \frac{ 2\exp[- s_n (c_2-2 \mathfrak{R}_{s_n}(\mathfrak{F}) - 4b/\sqrt{n} )^2/2 b^2] +2 s_nC_{\beta} \exp( - \gamma \lfloor \sqrt{n} \rfloor^{\xi})
	}{(c_\pi/2) (\bar{\varepsilon}_n/3)^L  },
\end{multline}
with $\mathbb{P}_{y_{1:n}}$--probability going to $1$ as $n \rightarrow \infty$. 

Replacing $c_2$ in \eqref{eq2_proof_rate_dep} with \smash{$2 \mathfrak{R}_{s_n}(\mathfrak{F})+\sqrt{(2b^2/s_n)\log(n/\bar{\varepsilon}_n^L)} + 4b/\sqrt{n} $ }, gives
\begin{multline*}
\smash{\pi_n^{(\varepsilon^*+\bar{\varepsilon}_n)}\Bigl(\Bigl\{\theta:  \mathcal{D}_{\mathfrak{F}}(\mu_\theta,\mu^*)} > \smash{\varepsilon^* + \frac{4\bar{\varepsilon}_n}{3}+2 \mathfrak{R}_{s_n}(\mathfrak{F}) + \frac{4b}{\sqrt{n}}+\Bigl(\frac{2b^2}{s_n}\log \frac{n}{\bar{\varepsilon}_n^L}\Bigr)^{1/2}    }  \Bigr\}\Bigr)
	\\
	\leq \frac{4 \cdot 3^L}{n c_\pi} \Bigl(1 +
	 \frac{ n s_n C_\beta \exp( - \gamma \lfloor \sqrt{n} \rfloor^{\xi})
	}{\bar{\varepsilon}_n^L  } \Bigr).
\end{multline*}
To conclude, note that
$$
 n s_n C_\beta \exp( - \gamma \lfloor \sqrt{n} \rfloor^{\xi})/\bar{\varepsilon}_n^L
= n^{1+L/2} s_n C_\beta \exp( - \gamma \lfloor \sqrt{n} \rfloor^{\xi})/( (n \bar{\varepsilon}_n^2)^{L/2} ),
$$
where the numerator goes to $0$ and the denominator goes to $\infty$ when $n\to\infty$, under the setting of Theorem~\ref{thm_rademacher_dep}. Therefore, with $\mathbb{P}_{y_{1:n}}$--probability going to $1$ as $n \to \infty$, we have
\begin{equation*}
\smash{\pi_n^{(\varepsilon^*+\bar{\varepsilon}_n)}\Bigl(\Bigl\{\theta:  \mathcal{D}_{\mathfrak{F}}(\mu_\theta,\mu^*)} > \smash{\varepsilon^* + \frac{4\bar{\varepsilon}_n}{3}+2 \mathfrak{R}_{s_n}(\mathfrak{F})+ \frac{4b}{\sqrt{n}} +\Bigl(\frac{2b^2}{s_n}\log \frac{n}{\bar{\varepsilon}_n^L}\Bigr)^{1/2}    }  \Bigr\}\Bigr)
	\\
	\leq \frac{4 \cdot 3^L}{n c_\pi},
\end{equation*}
concluding the proof.
\end{proof}


\bibliographystyle{imsart-nameyear}
\bibliography{biblio}

\begin{thebibliography}{47}

\bibitem[\protect\citeauthoryear{Agrawal and Horel}{2021}]{agrawal2021optimal}
\begin{barticle}[author]
\bauthor{\bsnm{Agrawal},~\bfnm{Rohit}\binits{R.}} \AND
  \bauthor{\bsnm{Horel},~\bfnm{Thibaut}\binits{T.}}
(\byear{2021}).
\btitle{Optimal bounds between $f$-divergences and integral probability
  metrics}.
\bjournal{Journal of Machine Learning Research}
\bvolume{22}
\bpages{1--59}.
\end{barticle}
\endbibitem

\bibitem[\protect\citeauthoryear{Bartlett, Bousquet and
  Mendelson}{2005}]{bartlett2005local}
\begin{barticle}[author]
\bauthor{\bsnm{Bartlett},~\bfnm{Peter~L}\binits{P.~L.}},
  \bauthor{\bsnm{Bousquet},~\bfnm{Olivier}\binits{O.}} \AND
  \bauthor{\bsnm{Mendelson},~\bfnm{Shahar}\binits{S.}}
(\byear{2005}).
\btitle{Local {Rademacher} complexities}.
\bjournal{Annals of Statistics}
\bvolume{33}
\bpages{1497--1537}.
\end{barticle}
\endbibitem

\bibitem[\protect\citeauthoryear{Bartlett and
  Mendelson}{2002}]{bartlett2002rademacher}
\begin{barticle}[author]
\bauthor{\bsnm{Bartlett},~\bfnm{P.~L.}\binits{P.~L.}} \AND
  \bauthor{\bsnm{Mendelson},~\bfnm{S.}\binits{S.}}
(\byear{2002}).
\btitle{Rademacher and {G}aussian complexities: Risk bounds and structural
  results}.
\bjournal{Journal of Machine Learning Research}
\bvolume{3}
\bpages{463--482}.
\end{barticle}
\endbibitem

\bibitem[\protect\citeauthoryear{Bernton et~al.}{2019}]{bernton2019approximate}
\begin{barticle}[author]
\bauthor{\bsnm{Bernton},~\bfnm{Espen}\binits{E.}},
  \bauthor{\bsnm{Jacob},~\bfnm{Pierre~E}\binits{P.~E.}},
  \bauthor{\bsnm{Gerber},~\bfnm{Mathieu}\binits{M.}} \AND
  \bauthor{\bsnm{Robert},~\bfnm{Christian~P}\binits{C.~P.}}
(\byear{2019}).
\btitle{Approximate {Bayesian} computation with the {Wasserstein} distance}.
\bjournal{Journal of the Royal Statistical Society Series B: Statistical
  Methodology}
\bvolume{81}
\bpages{235--269}.
\end{barticle}
\endbibitem

\bibitem[\protect\citeauthoryear{Birrell et~al.}{2022}]{birrell2022f}
\begin{barticle}[author]
\bauthor{\bsnm{Birrell},~\bfnm{Jeremiah}\binits{J.}},
  \bauthor{\bsnm{Dupuis},~\bfnm{Paul}\binits{P.}},
  \bauthor{\bsnm{Katsoulakis},~\bfnm{Markos~A}\binits{M.~A.}},
  \bauthor{\bsnm{Pantazis},~\bfnm{Yannis}\binits{Y.}} \AND
  \bauthor{\bsnm{Rey-Bellet},~\bfnm{Luc}\binits{L.}}
(\byear{2022}).
\btitle{$(f, \Gamma)$-divergences: Interpolating between $f$-divergences and
  integral probability metrics}.
\bjournal{Journal of Machine Learning Research}
\bvolume{23}
\bpages{1--70}.
\end{barticle}
\endbibitem

\bibitem[\protect\citeauthoryear{Bissiri, Holmes and
  Walker}{2016}]{bissiri2016general}
\begin{barticle}[author]
\bauthor{\bsnm{Bissiri},~\bfnm{Pier~Giovanni}\binits{P.~G.}},
  \bauthor{\bsnm{Holmes},~\bfnm{Chris~C}\binits{C.~C.}} \AND
  \bauthor{\bsnm{Walker},~\bfnm{Stephen~G}\binits{S.~G.}}
(\byear{2016}).
\btitle{A general framework for updating belief distributions}.
\bjournal{Journal of the Royal Statistical Society Series B: Statistical
  Methodology}
\bvolume{78}
\bpages{1103--1130}.
\end{barticle}
\endbibitem

\bibitem[\protect\citeauthoryear{Ch{\'e}rief-Abdellatif and
  Alquier}{2020}]{cherief2020mmd}
\begin{binproceedings}[author]
\bauthor{\bsnm{Ch{\'e}rief-Abdellatif},~\bfnm{Badr-Eddine}\binits{B.-E.}} \AND
  \bauthor{\bsnm{Alquier},~\bfnm{Pierre}\binits{P.}}
(\byear{2020}).
\btitle{{MMD}-{B}ayes: {R}obust {Bayesian} estimation via maximum mean
  discrepancy}.
In \bbooktitle{Symposium on Advances in Approximate Bayesian Inference}
\bpages{1--21}.
\bpublisher{PMLR}.
\end{binproceedings}
\endbibitem

\bibitem[\protect\citeauthoryear{Ch{\'e}rief-Abdellatif and
  Alquier}{2022}]{cherief2022finite}
\begin{barticle}[author]
\bauthor{\bsnm{Ch{\'e}rief-Abdellatif},~\bfnm{Badr-Eddine}\binits{B.-E.}} \AND
  \bauthor{\bsnm{Alquier},~\bfnm{Pierre}\binits{P.}}
(\byear{2022}).
\btitle{{Finite sample properties of parametric {MMD} estimation: Robustness to
  misspecification and dependence}}.
\bjournal{Bernoulli}
\bvolume{28}
\bpages{181 -- 213}.
\end{barticle}
\endbibitem

\bibitem[\protect\citeauthoryear{Davydov}{1974}]{davydov}
\begin{barticle}[author]
\bauthor{\bsnm{Davydov},~\bfnm{Yu~A}\binits{Y.~A.}}
(\byear{1974}).
\btitle{Mixing conditions for {M}arkov chains}.
\bjournal{Theory of Probability \& Its Applications}
\bvolume{18}
\bpages{312--328}.
\end{barticle}
\endbibitem

\bibitem[\protect\citeauthoryear{Doukhan}{1994}]{doukhan}
\begin{bbook}[author]
\bauthor{\bsnm{Doukhan},~\bfnm{Paul}\binits{P.}}
(\byear{1994}).
\btitle{{Mixing: Properties and Examples}}
\bvolume{85}.
\bpublisher{Springer Science \& Business Media}.
\end{bbook}
\endbibitem

\bibitem[\protect\citeauthoryear{Drovandi and
  Frazier}{2022}]{drovandi2022comparison}
\begin{barticle}[author]
\bauthor{\bsnm{Drovandi},~\bfnm{Christopher}\binits{C.}} \AND
  \bauthor{\bsnm{Frazier},~\bfnm{David~T}\binits{D.~T.}}
(\byear{2022}).
\btitle{A comparison of likelihood-free methods with and without summary
  statistics}.
\bjournal{Statistics and Computing}
\bvolume{32}
\bpages{1--23}.
\end{barticle}
\endbibitem

\bibitem[\protect\citeauthoryear{Dudley}{2018}]{dudley2018real}
\begin{bbook}[author]
\bauthor{\bsnm{Dudley},~\bfnm{Richard~M}\binits{R.~M.}}
(\byear{2018}).
\btitle{{Real Analysis and Probability}}.
\bpublisher{CRC Press}.
\end{bbook}
\endbibitem

\bibitem[\protect\citeauthoryear{Dudley, Gin{\'e} and Zinn}{1991}]{dud1991}
\begin{barticle}[author]
\bauthor{\bsnm{Dudley},~\bfnm{Richard~M}\binits{R.~M.}},
  \bauthor{\bsnm{Gin{\'e}},~\bfnm{Evarist}\binits{E.}} \AND
  \bauthor{\bsnm{Zinn},~\bfnm{Joel}\binits{J.}}
(\byear{1991}).
\btitle{Uniform and universal {G}livenko-{C}antelli classes}.
\bjournal{Journal of Theoretical Probability}
\bvolume{4}
\bpages{485--510}.
\end{barticle}
\endbibitem

\bibitem[\protect\citeauthoryear{Dutta et~al.}{2021}]{Dutta2021}
\begin{barticle}[author]
\bauthor{\bsnm{Dutta},~\bfnm{Ritabrata}\binits{R.}},
  \bauthor{\bsnm{Schoengens},~\bfnm{Marcel}\binits{M.}},
  \bauthor{\bsnm{Pacchiardi},~\bfnm{Lorenzo}\binits{L.}},
  \bauthor{\bsnm{Ummadisingu},~\bfnm{Avinash}\binits{A.}},
  \bauthor{\bsnm{Widmer},~\bfnm{Nicole}\binits{N.}},
  \bauthor{\bsnm{Künzli},~\bfnm{Pierre}\binits{P.}},
  \bauthor{\bsnm{Onnela},~\bfnm{Jukka-Pekka}\binits{J.-P.}} \AND
  \bauthor{\bsnm{Mira},~\bfnm{Antonietta}\binits{A.}}
(\byear{2021}).
\btitle{{ABCpy}: {A} high-performance computing perspective to approximate
  {Bayesian} computation}.
\bjournal{Journal of Statistical Software}
\bvolume{100}
\bpages{1–38}.
\end{barticle}
\endbibitem

\bibitem[\protect\citeauthoryear{Fearnhead and
  Prangle}{2012}]{fearnhead2012constructing}
\begin{barticle}[author]
\bauthor{\bsnm{Fearnhead},~\bfnm{Paul}\binits{P.}} \AND
  \bauthor{\bsnm{Prangle},~\bfnm{Dennis}\binits{D.}}
(\byear{2012}).
\btitle{Constructing summary statistics for approximate {Bayesian} computation:
  Semi-automatic approximate {Bayesian} computation}.
\bjournal{Journal of the Royal Statistical Society Series B: Statistical
  Methodology}
\bvolume{74}
\bpages{419--474}.
\end{barticle}
\endbibitem

\bibitem[\protect\citeauthoryear{Forbes et~al.}{2021}]{forbes2021approximate}
\begin{barticle}[author]
\bauthor{\bsnm{Forbes},~\bfnm{Florence}\binits{F.}},
  \bauthor{\bsnm{Nguyen},~\bfnm{Hien~Duy}\binits{H.~D.}},
  \bauthor{\bsnm{Nguyen},~\bfnm{Trung~Tin}\binits{T.~T.}} \AND
  \bauthor{\bsnm{Arbel},~\bfnm{Julyan}\binits{J.}}
(\byear{2021}).
\btitle{Approximate {Bayesian} computation with surrogate posteriors}.
\bjournal{hal-03139256v4}.
\end{barticle}
\endbibitem

\bibitem[\protect\citeauthoryear{Fournier and Guillin}{2015}]{fournier2015rate}
\begin{barticle}[author]
\bauthor{\bsnm{Fournier},~\bfnm{Nicolas}\binits{N.}} \AND
  \bauthor{\bsnm{Guillin},~\bfnm{Arnaud}\binits{A.}}
(\byear{2015}).
\btitle{On the rate of convergence in {W}asserstein distance of the empirical
  measure}.
\bjournal{Probability Theory and Related Fields}
\bvolume{162}
\bpages{707--738}.
\end{barticle}
\endbibitem

\bibitem[\protect\citeauthoryear{Frazier}{2020}]{frazier2020robust}
\begin{barticle}[author]
\bauthor{\bsnm{Frazier},~\bfnm{David~T}\binits{D.~T.}}
(\byear{2020}).
\btitle{Robust and efficient approximate {Bayesian} computation: A minimum
  distance approach}.
\bjournal{arXiv preprint arXiv:2006.14126}.
\end{barticle}
\endbibitem

\bibitem[\protect\citeauthoryear{Frazier, Knoblauch and
  Drovandi}{2024}]{frazier2024}
\begin{barticle}[author]
\bauthor{\bsnm{Frazier},~\bfnm{David~T}\binits{D.~T.}},
  \bauthor{\bsnm{Knoblauch},~\bfnm{Jeremias}\binits{J.}} \AND
  \bauthor{\bsnm{Drovandi},~\bfnm{Christopher}\binits{C.}}
(\byear{2024}).
\btitle{The impact of loss estimation on {G}ibbs measures}.
\bjournal{arXiv preprint arXiv:2404.15649}.
\end{barticle}
\endbibitem

\bibitem[\protect\citeauthoryear{Frazier, Robert and
  Rousseau}{2020}]{frazier2020model}
\begin{barticle}[author]
\bauthor{\bsnm{Frazier},~\bfnm{David~T}\binits{D.~T.}},
  \bauthor{\bsnm{Robert},~\bfnm{Christian~P}\binits{C.~P.}} \AND
  \bauthor{\bsnm{Rousseau},~\bfnm{Judith}\binits{J.}}
(\byear{2020}).
\btitle{Model misspecification in approximate {B}ayesian computation:
  consequences and diagnostics}.
\bjournal{Journal of the Royal Statistical Society Series B: Statistical
  Methodology}
\bvolume{82}
\bpages{421--444}.
\end{barticle}
\endbibitem

\bibitem[\protect\citeauthoryear{Frazier et~al.}{2018}]{frazier2018asymptotic}
\begin{barticle}[author]
\bauthor{\bsnm{Frazier},~\bfnm{David~T}\binits{D.~T.}},
  \bauthor{\bsnm{Martin},~\bfnm{Gael~M}\binits{G.~M.}},
  \bauthor{\bsnm{Robert},~\bfnm{Christian~P}\binits{C.~P.}} \AND
  \bauthor{\bsnm{Rousseau},~\bfnm{Judith}\binits{J.}}
(\byear{2018}).
\btitle{Asymptotic properties of approximate {Bayesian} computation}.
\bjournal{Biometrika}
\bvolume{105}
\bpages{593--607}.
\end{barticle}
\endbibitem

\bibitem[\protect\citeauthoryear{Fujisawa et~al.}{2021}]{fujisawa2021gamma}
\begin{binproceedings}[author]
\bauthor{\bsnm{Fujisawa},~\bfnm{Masahiro}\binits{M.}},
  \bauthor{\bsnm{Teshima},~\bfnm{Takeshi}\binits{T.}},
  \bauthor{\bsnm{Sato},~\bfnm{Issei}\binits{I.}} \AND
  \bauthor{\bsnm{Sugiyama},~\bfnm{Masashi}\binits{M.}}
(\byear{2021}).
\btitle{$\gamma$-{ABC}: {O}utlier-robust approximate {Bayesian} computation
  based on a robust divergence estimator}.
In \bbooktitle{International Conference on Artificial Intelligence and
  Statistics}
\bpages{1783--1791}.
\bpublisher{PMLR}.
\end{binproceedings}
\endbibitem

\bibitem[\protect\citeauthoryear{Gretton et~al.}{2012}]{gret2012}
\begin{barticle}[author]
\bauthor{\bsnm{Gretton},~\bfnm{Arthur}\binits{A.}},
  \bauthor{\bsnm{Borgwardt},~\bfnm{Karsten~M}\binits{K.~M.}},
  \bauthor{\bsnm{Rasch},~\bfnm{Malte~J}\binits{M.~J.}},
  \bauthor{\bsnm{Sch{\"o}lkopf},~\bfnm{Bernhard}\binits{B.}} \AND
  \bauthor{\bsnm{Smola},~\bfnm{Alexander}\binits{A.}}
(\byear{2012}).
\btitle{A kernel two-sample test}.
\bjournal{Journal of Machine Learning Research}
\bvolume{13}
\bpages{723--773}.
\end{barticle}
\endbibitem

\bibitem[\protect\citeauthoryear{Gutmann et~al.}{2018}]{gutmann2018likelihood}
\begin{barticle}[author]
\bauthor{\bsnm{Gutmann},~\bfnm{Michael~U}\binits{M.~U.}},
  \bauthor{\bsnm{Dutta},~\bfnm{Ritabrata}\binits{R.}},
  \bauthor{\bsnm{Kaski},~\bfnm{Samuel}\binits{S.}} \AND
  \bauthor{\bsnm{Corander},~\bfnm{Jukka}\binits{J.}}
(\byear{2018}).
\btitle{Likelihood-free inference via classification}.
\bjournal{Statistics and Computing}
\bvolume{28}
\bpages{411--425}.
\end{barticle}
\endbibitem

\bibitem[\protect\citeauthoryear{Hofmann, Sch{\"o}lkopf and
  Smola}{2008}]{hofmann2008kernel}
\begin{barticle}[author]
\bauthor{\bsnm{Hofmann},~\bfnm{Thomas}\binits{T.}},
  \bauthor{\bsnm{Sch{\"o}lkopf},~\bfnm{Bernhard}\binits{B.}} \AND
  \bauthor{\bsnm{Smola},~\bfnm{Alexander~J}\binits{A.~J.}}
(\byear{2008}).
\btitle{Kernel methods in machine learning}.
\bjournal{Annals of Statistics}
\bvolume{36}
\bpages{1171--1220}.
\end{barticle}
\endbibitem

\bibitem[\protect\citeauthoryear{Jewson, Smith and
  Holmes}{2018}]{jewson2018principles}
\begin{barticle}[author]
\bauthor{\bsnm{Jewson},~\bfnm{Jack}\binits{J.}},
  \bauthor{\bsnm{Smith},~\bfnm{Jim~Q}\binits{J.~Q.}} \AND
  \bauthor{\bsnm{Holmes},~\bfnm{Chris}\binits{C.}}
(\byear{2018}).
\btitle{Principles of {B}ayesian inference using general divergence criteria}.
\bjournal{Entropy}
\bvolume{20}
\bpages{442}.
\end{barticle}
\endbibitem

\bibitem[\protect\citeauthoryear{Jiang, Wu and
  Wong}{2018}]{jiang2018approximate}
\begin{binproceedings}[author]
\bauthor{\bsnm{Jiang},~\bfnm{Bai}\binits{B.}},
  \bauthor{\bsnm{Wu},~\bfnm{Tung-Yu}\binits{T.-Y.}} \AND
  \bauthor{\bsnm{Wong},~\bfnm{Wing~Hung}\binits{W.~H.}}
(\byear{2018}).
\btitle{Approximate {Bayesian} computation with {Kullback-Leibler} divergence
  as data discrepancy}.
In \bbooktitle{International Conference on Artificial Intelligence and
  Statistics}
\bpages{1711--1721}.
\bpublisher{PMLR}.
\end{binproceedings}
\endbibitem

\bibitem[\protect\citeauthoryear{Lei}{2020}]{lei2020w}
\begin{barticle}[author]
\bauthor{\bsnm{Lei},~\bfnm{Jing}\binits{J.}}
(\byear{2020}).
\btitle{Convergence and concentration of empirical measures under {W}asserstein
  distance in unbounded functional spaces}.
\bjournal{Bernoulli}
\bvolume{26}
\bpages{767--798}.
\end{barticle}
\endbibitem

\bibitem[\protect\citeauthoryear{Li and Fearnhead}{2018}]{lifearn2018}
\begin{barticle}[author]
\bauthor{\bsnm{Li},~\bfnm{Wentao}\binits{W.}} \AND
  \bauthor{\bsnm{Fearnhead},~\bfnm{Paul}\binits{P.}}
(\byear{2018}).
\btitle{On the asymptotic efficiency of approximate {Bayesian} computation
  estimators}.
\bjournal{Biometrika}
\bvolume{105}
\bpages{285--299}.
\end{barticle}
\endbibitem

\bibitem[\protect\citeauthoryear{Marin et~al.}{2012}]{marin2012approximate}
\begin{barticle}[author]
\bauthor{\bsnm{Marin},~\bfnm{Jean-Michel}\binits{J.-M.}},
  \bauthor{\bsnm{Pudlo},~\bfnm{Pierre}\binits{P.}},
  \bauthor{\bsnm{Robert},~\bfnm{Christian~P}\binits{C.~P.}} \AND
  \bauthor{\bsnm{Ryder},~\bfnm{Robin~J}\binits{R.~J.}}
(\byear{2012}).
\btitle{Approximate {Bayesian} computational methods}.
\bjournal{Statistics and Computing}
\bvolume{22}
\bpages{1167--1180}.
\end{barticle}
\endbibitem

\bibitem[\protect\citeauthoryear{Marin et~al.}{2014}]{marin2014relevant}
\begin{barticle}[author]
\bauthor{\bsnm{Marin},~\bfnm{Jean-Michel}\binits{J.-M.}},
  \bauthor{\bsnm{Pillai},~\bfnm{Natesh~S}\binits{N.~S.}},
  \bauthor{\bsnm{Robert},~\bfnm{Christian~P}\binits{C.~P.}} \AND
  \bauthor{\bsnm{Rousseau},~\bfnm{Judith}\binits{J.}}
(\byear{2014}).
\btitle{Relevant statistics for {Bayesian} model choice}.
\bjournal{Journal of the Royal Statistical Society Series B: Statistical
  Methodology}
\bvolume{76}
\bpages{833--859}.
\end{barticle}
\endbibitem

\bibitem[\protect\citeauthoryear{Massart}{2000}]{massart2000some}
\begin{binproceedings}[author]
\bauthor{\bsnm{Massart},~\bfnm{Pascal}\binits{P.}}
(\byear{2000}).
\btitle{Some applications of concentration inequalities to statistics}.
In \bbooktitle{Annales de la Facult{\'e} des sciences de Toulouse:
  Math{\'e}matiques}
\bvolume{9}
\bpages{245--303}.
\end{binproceedings}
\endbibitem

\bibitem[\protect\citeauthoryear{Matsubara et~al.}{2022}]{matsubara2021robust}
\begin{barticle}[author]
\bauthor{\bsnm{Matsubara},~\bfnm{Takuo}\binits{T.}},
  \bauthor{\bsnm{Knoblauch},~\bfnm{Jeremias}\binits{J.}},
  \bauthor{\bsnm{Briol},~\bfnm{Fran{\c{c}}ois-Xavier}\binits{F.-X.}} \AND
  \bauthor{\bsnm{Oates},~\bfnm{Chris~J}\binits{C.~J.}}
(\byear{2022}).
\btitle{Robust generalised {B}ayesian inference for intractable likelihoods}.
\bjournal{Journal of the Royal Statistical Society Series B: Statistical
  Methodology}
\bvolume{84}
\bpages{997--1022}.
\end{barticle}
\endbibitem

\bibitem[\protect\citeauthoryear{Miller and Dunson}{2019}]{miller2019robust}
\begin{barticle}[author]
\bauthor{\bsnm{Miller},~\bfnm{Jeffrey~W}\binits{J.~W.}} \AND
  \bauthor{\bsnm{Dunson},~\bfnm{David~B}\binits{D.~B.}}
(\byear{2019}).
\btitle{Robust {Bayesian} inference via coarsening}.
\bjournal{Journal of the American Statistical Association}
\bvolume{114}
\bpages{1113--1125}.
\end{barticle}
\endbibitem

\bibitem[\protect\citeauthoryear{Mohri and Rostamizadeh}{2008}]{Mohri2008}
\begin{binproceedings}[author]
\bauthor{\bsnm{Mohri},~\bfnm{Mehryar}\binits{M.}} \AND
  \bauthor{\bsnm{Rostamizadeh},~\bfnm{Afshin}\binits{A.}}
(\byear{2008}).
\btitle{Rademacher complexity bounds for non-i.i.d.\ processes}.
In \bbooktitle{Advances in Neural Information Processing Systems}
\bvolume{21}
\bpages{1097--1104}.
\end{binproceedings}
\endbibitem

\bibitem[\protect\citeauthoryear{Muandet et~al.}{2017}]{muandet2017kernel}
\begin{barticle}[author]
\bauthor{\bsnm{Muandet},~\bfnm{K.}\binits{K.}},
  \bauthor{\bsnm{Fukumizu},~\bfnm{K.}\binits{K.}},
  \bauthor{\bsnm{Sriperumbudur},~\bfnm{B.}\binits{B.}} \AND
  \bauthor{\bsnm{Sch{\"o}lkopf},~\bfnm{B.}\binits{B.}}
(\byear{2017}).
\btitle{Kernel mean embedding of distributions: {A} review and beyond}.
\bjournal{Foundations and Trends in Machine Learning}
\bvolume{10}
\bpages{1--141}.
\end{barticle}
\endbibitem

\bibitem[\protect\citeauthoryear{M{\"u}ller}{1997}]{muller1997integral}
\begin{barticle}[author]
\bauthor{\bsnm{M{\"u}ller},~\bfnm{Alfred}\binits{A.}}
(\byear{1997}).
\btitle{Integral probability metrics and their generating classes of
  functions}.
\bjournal{Advances in Applied Probability,}
\bvolume{29}
\bpages{429--443}.
\end{barticle}
\endbibitem

\bibitem[\protect\citeauthoryear{Nguyen et~al.}{2020}]{nguyen2020approximate}
\begin{barticle}[author]
\bauthor{\bsnm{Nguyen},~\bfnm{Hien~Duy}\binits{H.~D.}},
  \bauthor{\bsnm{Arbel},~\bfnm{Julyan}\binits{J.}},
  \bauthor{\bsnm{L{\"u}},~\bfnm{Hongliang}\binits{H.}} \AND
  \bauthor{\bsnm{Forbes},~\bfnm{Florence}\binits{F.}}
(\byear{2020}).
\btitle{Approximate {Bayesian} computation via the energy statistic}.
\bjournal{IEEE Access}
\bvolume{8}
\bpages{131683--131698}.
\end{barticle}
\endbibitem

\bibitem[\protect\citeauthoryear{Park, Jitkrittum and
  Sejdinovic}{2016}]{park2016k2}
\begin{binproceedings}[author]
\bauthor{\bsnm{Park},~\bfnm{Mijung}\binits{M.}},
  \bauthor{\bsnm{Jitkrittum},~\bfnm{Wittawat}\binits{W.}} \AND
  \bauthor{\bsnm{Sejdinovic},~\bfnm{Dino}\binits{D.}}
(\byear{2016}).
\btitle{{K2-ABC}: {Approximate Bayesian} computation with kernel embeddings}.
In \bbooktitle{International Conference on Artificial Intelligence and
  Statistics}
\bpages{398--407}.
\bpublisher{PMLR}.
\end{binproceedings}
\endbibitem

\bibitem[\protect\citeauthoryear{Ramdas, Trillos and Cuturi}{2017}]{ram2017}
\begin{barticle}[author]
\bauthor{\bsnm{Ramdas},~\bfnm{Aaditya}\binits{A.}},
  \bauthor{\bsnm{Trillos},~\bfnm{Nicol{\'a}s~Garc{\'\i}a}\binits{N.~G.}} \AND
  \bauthor{\bsnm{Cuturi},~\bfnm{Marco}\binits{M.}}
(\byear{2017}).
\btitle{On {Wasserstein} two-sample testing and related families of
  nonparametric tests}.
\bjournal{Entropy}
\bvolume{19}
\bpages{47}.
\end{barticle}
\endbibitem

\bibitem[\protect\citeauthoryear{Sejdinovic
  et~al.}{2013}]{sejdinovic2013equivalence}
\begin{barticle}[author]
\bauthor{\bsnm{Sejdinovic},~\bfnm{Dino}\binits{D.}},
  \bauthor{\bsnm{Sriperumbudur},~\bfnm{Bharath}\binits{B.}},
  \bauthor{\bsnm{Gretton},~\bfnm{Arthur}\binits{A.}} \AND
  \bauthor{\bsnm{Fukumizu},~\bfnm{Kenji}\binits{K.}}
(\byear{2013}).
\btitle{Equivalence of distance-based and {RKHS-based} statistics in hypothesis
  testing}.
\bjournal{Annals of Statistics}
\bvolume{41}
\bpages{2263--2291}.
\end{barticle}
\endbibitem

\bibitem[\protect\citeauthoryear{Sriperumbudur
  et~al.}{2012}]{sriperumbudur2009integral}
\begin{barticle}[author]
\bauthor{\bsnm{Sriperumbudur},~\bfnm{Bharath~K}\binits{B.~K.}},
  \bauthor{\bsnm{Fukumizu},~\bfnm{Kenji}\binits{K.}},
  \bauthor{\bsnm{Gretton},~\bfnm{Arthur}\binits{A.}},
  \bauthor{\bsnm{Sch{\"o}lkopf},~\bfnm{Bernhard}\binits{B.}} \AND
  \bauthor{\bsnm{Lanckriet},~\bfnm{Gert~RG}\binits{G.~R.}}
(\byear{2012}).
\btitle{On the empirical estimation of integral probability metrics}.
\bjournal{Electronic Journal of Statistics}
\bvolume{6}
\bpages{1550--1599}.
\end{barticle}
\endbibitem

\bibitem[\protect\citeauthoryear{Talagrand}{1994}]{tal1994}
\begin{barticle}[author]
\bauthor{\bsnm{Talagrand},~\bfnm{Michel}\binits{M.}}
(\byear{1994}).
\btitle{The transportation cost from the uniform measure to the empirical
  measure in dimension $ \geq 3$}.
\bjournal{Annals of Probability}
\bvolume{22}
\bpages{919--959}.
\end{barticle}
\endbibitem

\bibitem[\protect\citeauthoryear{Villani}{2021}]{villani2021topics}
\begin{bbook}[author]
\bauthor{\bsnm{Villani},~\bfnm{C{\'e}dric}\binits{C.}}
(\byear{2021}).
\btitle{{Topics in Optimal Transportation (Second Edition)}}.
\bpublisher{American Mathematical Society}.
\end{bbook}
\endbibitem

\bibitem[\protect\citeauthoryear{Wainwright}{2019}]{wainwright2019high}
\begin{bbook}[author]
\bauthor{\bsnm{Wainwright},~\bfnm{M.~J.}\binits{M.~J.}}
(\byear{2019}).
\btitle{{High-Dimensional Statistics: A Non-Asymptotic Viewpoint}}.
\bpublisher{Cambridge University Press}.
\end{bbook}
\endbibitem

\bibitem[\protect\citeauthoryear{Wang, Kaji and Rockova}{2022}]{wang2022}
\begin{barticle}[author]
\bauthor{\bsnm{Wang},~\bfnm{Yuexi}\binits{Y.}},
  \bauthor{\bsnm{Kaji},~\bfnm{Tetsuya}\binits{T.}} \AND
  \bauthor{\bsnm{Rockova},~\bfnm{Veronika}\binits{V.}}
(\byear{2022}).
\btitle{Approximate {B}ayesian computation via classification}.
\bjournal{Journal of Machine Learning Research}
\bvolume{23}
\bpages{1--49}.
\end{barticle}
\endbibitem

\bibitem[\protect\citeauthoryear{Weed and Bach}{2019}]{weed2019sharp}
\begin{barticle}[author]
\bauthor{\bsnm{Weed},~\bfnm{J.}\binits{J.}} \AND
  \bauthor{\bsnm{Bach},~\bfnm{F.}\binits{F.}}
(\byear{2019}).
\btitle{Sharp asymptotic and finite-sample rates of convergence of empirical
  measures in {W}asserstein distance}.
\bjournal{Bernoulli}
\bvolume{25}
\bpages{2620--2648}.
\end{barticle}
\endbibitem

\end{thebibliography}

\end{document}